%% file: main-paper.tex
\documentclass[12pt]{article}
\usepackage{template}

\newcommand\reallywidehat[1]{\arraycolsep=0pt\relax%
	\begin{array}{c}
		\stretchto{
			\scaleto{
				\scalerel*[\widthof{\ensuremath{#1}}]{\kern-.5pt\bigwedge\kern-.5pt}
				{\rule[-\textheight/2]{1ex}{\textheight}} 
			}{\textheight} %
		}{0.5ex}\\           
		#1\\                 
		\rule{-1ex}{0ex}
	\end{array}
}

\title{Simultaneous prediction and community detection for networks with application to neuroimaging}

\author[$1$]{Jes\'us Arroyo\thanks{jesus.arroyo@jhu.edu}}
\author[$2$]{Elizaveta Levina\thanks{elevina@umich.edu}}
\affil[$1$]{\small Center for Imaging Science, Johns Hopkins University}
\affil[$2$]{\small Department of Statistics, University of Michigan}
\date{}

\begin{document}
\maketitle

\begin{abstract}
	
Community structure in networks is observed in many different domains, and unsupervised community detection has received a lot of attention in the literature.   Increasingly the focus of network analysis is shifting towards using network information in some other prediction or inference task rather than just analyzing the network itself.   In particular, in neuroimaging applications brain networks are available for multiple subjects and the goal is often to predict a phenotype of interest. Community structure is well known to be a feature of brain networks, typically corresponding to different regions of the brain responsible for different functions.     There are standard parcellations of the brain into such regions, usually obtained by applying clustering methods to brain connectomes of healthy subjects.  However,  when the goal is predicting a phenotype or distinguishing between different conditions, these static communities from an unrelated set of healthy subjects may not be the most useful for prediction.   Here we present a method for supervised community detection, aiming to find a partition of the network into communities that is most useful for predicting a particular response.     We use a block-structured regularization penalty combined with a prediction loss function, and compute the solution with a combination of a spectral method and an ADMM optimization algorithm. We show that the spectral clustering method recovers the correct communities under a weighted stochastic block model.   The method performs well on both simulated and real brain networks, providing support for the idea of task-dependent brain regions.
	
\end{abstract}

\input{intro}

\input{method}

\input{algorithm}

\input{theory}

\input{simulations}

\input{data}

\input{discussion}

\section*{Acknowledgements}
This research was supported in part by NSF grants DMS-1521551 and DMS-1916222, and
a Dana Foundation grant to E. Levina. The authors would like to thank our collaborator Stephan F. Taylor and his lab for providing a processed version of the data, Chandra Sripada and Daniel Kessler for helpful discussions, and the computational resources and services provided by Advanced Research Computing at the University of Michigan.
\input{appendix}

\bibliographystyle{apalike}
\bibliography{Biblio}

\end{document}

%% file: intro.tex
\section{Introduction}
\label{sec:intro}

While the study of networks has traditionally been driven by social sciences applications and focused on understanding the structure of single network,  new methods for statistical analysis of multiple networks have started to emerge  \citep{ginestet2017hypothesis,narayan2015two,arroyo2016graphclass,athreya2017statistical,Le2018,Levin2019}.   To a large extent, the interest in multiple network analysis is driven by neuroimaging applications, where brain networks are constructed from raw imaging data, such as fMRI, to represent connectivity between a predefined set of nodes in the brain \citep{bullmore2009complex}, often referred to as regions of interest (ROIs).    Collecting data from multiple subjects has made possible population level  studies of the brain under different conditions, for instance, mental illness.   Typically this is accomplished by using the network as a covariate for predicting or conducting a test on a phenotype, either a category such as disease diagnosis, or a quantitative measurement like attention level.   

Most of the previous work on using networks  as covariates in supervised prediction problems has followed one of two general approaches.   One approach reduces the  network to global features that summarize its structure, such as average degree, centrality, or clustering coefficient \citep{bullmore2009complex},  but these features fail to capture local structure and might not contain useful information for the task of interest.    The other approach vectorizes the adjacency matrix, treating all edge weights as individual features, and then applies standard prediction tools for multivariate data, but this can reduce both accuracy and  interpretability if network structure is not considered \citep{arroyo2016graphclass}.  More recently, unsupervised dimensionality reduction methods for multiple networks  \citep{tang2017nonparametric,wang2017joint,Arroyo2019} have been used to embed the networks in a low dimensional space before fitting a prediction rule, but this two-stage approach  may not be optimal.

Communities are a structure commonly found in networks from many domains, including social, biological and technological networks, among others \citep{girvan2002community,Fortunato2016}. A community is typically defined as a group of nodes with similar connectivity patterns.  A common interpretation of community is as a group of nodes that have stronger connections within the group than to the rest of the network. The problem of finding the communities in a network can be defined as an unsupervised clustering of the nodes, for which many statistical models and even more algorithms have been proposed, and theoretical properties obtained in simpler settings; see \cite{Abbe2017} and \cite{Fortunato2016} for recent reviews.

\begin{figure}
	\begin{subfigure}{.6\textwidth}
		\includegraphics[height=2.5in]{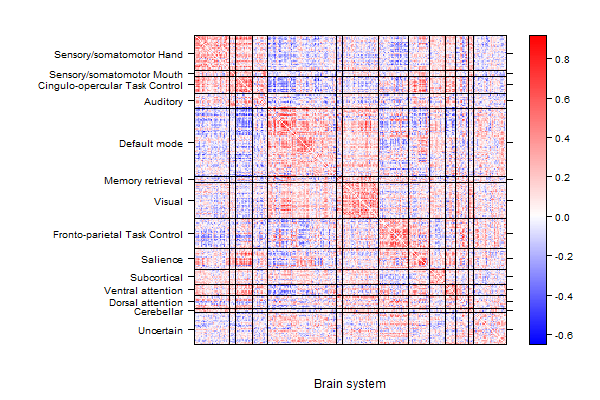}
		\caption{Edge connectivity matrix of a single healthy subject. Nodes are divided by brain systems, labeled on the left.}
	\end{subfigure}
	\begin{subfigure}{0.4\textwidth}
		\includegraphics[height=2.5in]{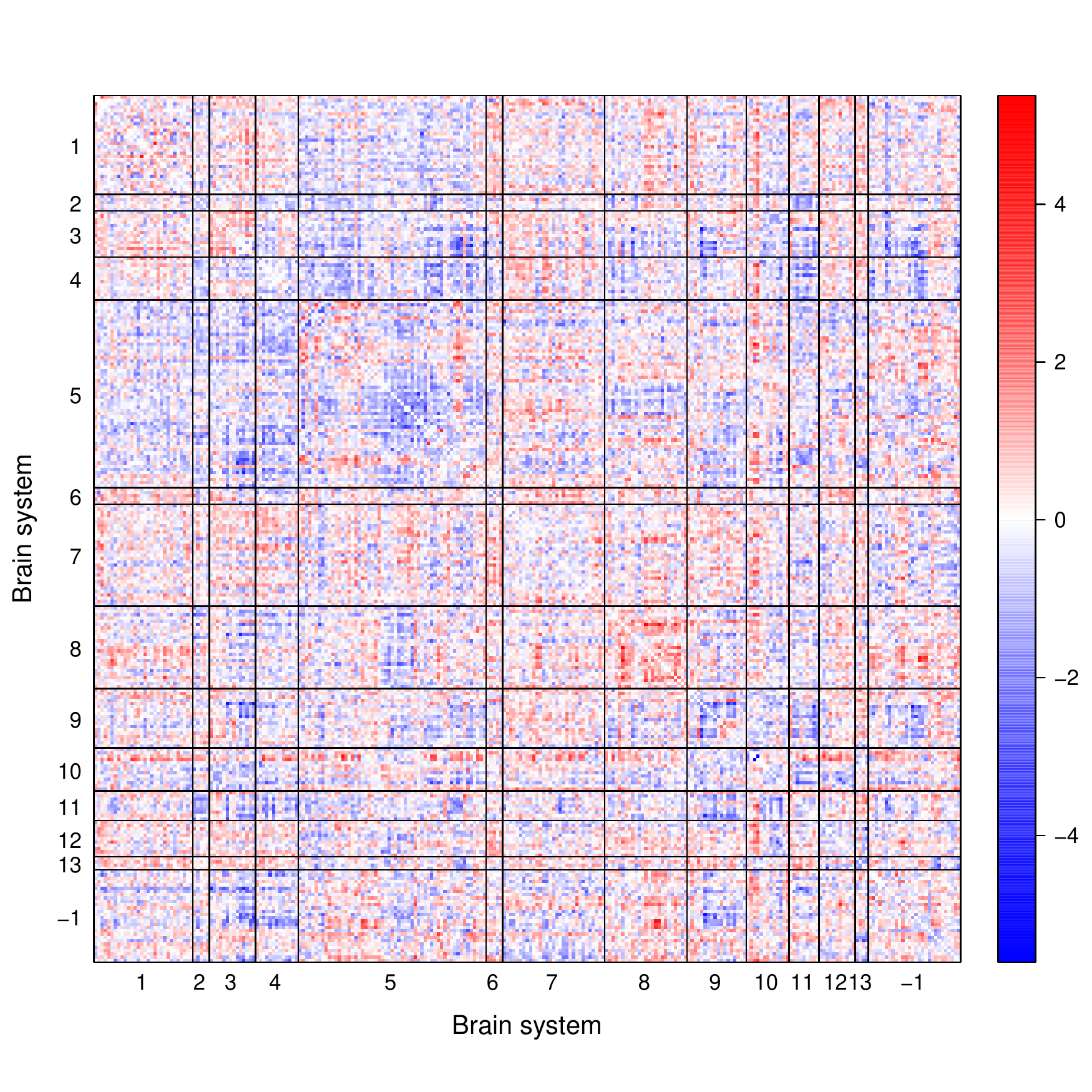}
		\caption{The t-statistics of edge difference between healthy and schizophrenic subjects.}
	\end{subfigure}
	\caption{COBRE data edge statistics.}
        \label{fig:subject-vs-ttest} 
\end{figure}


In brain networks, communities correspond to functional brain systems \citep{fox2005human,chen2008revealing,bullmore2009complex}. The nodes in a community can be thought of  as activating together during tasks and having similar functionality in the brain. Thus, each set of edges connecting two communities or brain systems, referred to as a network \emph{cell}, tends to have homogeneous connectivity levels.  This pattern can be seen in  Figure \ref{fig:subject-vs-ttest}(a), which shows the fMRI brain network of a healthy subject from the COBRE data (described in more detail in Section \ref{sec:data-CD}).  In this example, the nodes of the network are defined according to the \cite{power2011functional} parcellation  which identifies 264 regions of interest (ROIs) in the brain. These nodes are divided into 14 communities (see labels in the figure) which mostly correspond to known functional systems of the brain.

Analyzing fMRI brain networks at the level of individual nodes or edges is typically noisy and hard to interpret;  interpreting the results in terms of brain systems and corresponding network cells is much more desirable in this application.   There is evidence that organization and connectivity between brain systems is associated with subject-level phenotypes of interest, such as age, gender, or mental illness diagnosis \citep{meunier2009age,sripada2014disrupted,kessler2016growth}, or can even  reconfigure according to task-induced states in short timescales \citep{Salehi2019}. 
However, there is no universal agreement on how to divide the brain into systems, and different community partitions have been proposed  \citep{power2011functional,ThomasYeo2011}. These partitions are typically constructed either at the subject level by applying a community detection algorithm to a single network, or at the sample level by estimating a common community structure for all the networks, typically from healthy subjects \citep{thirion2014fmri}.    Several parcellations have been obtained and mapped to known brain systems this way.   \citep{Fischl2004,Desikan2006, Craddock2012}   However, this approach does not account for other covariates of the subjects, or the fact that the functional systems may rearrange themselves depending on the task or condition \citep{smith2012temporally,fair2007development,sripada2014lag}. 

When studying the association between the brain networks and a subject-level phenotype, the choice of parcellation is important, as the resolution of some parcellations may not capture relevant associations.   Figure \ref{fig:subject-vs-ttest}(b) shows the two sample $t$-test statistic values for each edge, for the difference in mean connectivity  between schizophrenic  subjects and healthy controls from the COBRE data.    Clearly, the $t$-statistics values are not as homogeneous as the connectivity values themselves, and some of the pre-determined network cells have both significantly positive and significantly negative groups of edge differences,  making it difficult to interpret the cell as a whole. In particular, the default mode network (brain system 5) is a region that has been strongly associated with schizophrenia \citep{broyd2009default}, but interpreting this system as a single unit based on this parcellation can be misleading, since it contains regions indicating both positive and negative effects. 

In this paper, we develop a new method that learns the most relevant community structure simultaneously with solving the corresponding prediction problems.    We achieve this by enforcing a block-constant constraint on edge coefficients, in order to identify a grouping of nodes into clusters that both gives a good prediction accuracy for the phenotype of interest, and gives network cells that are homogeneous in their effect on the response variable.   The solution is obtained with a combination of a spectral method and an efficient iterative optimization algorithm based on ADMM. We study the theoretical performance of the spectral clustering method, and show good performance on both simulated networks and schizophrenia data, obtaining a new set of regions  that give a parsimonious and interpretable solution with good prediction accuracy.

The rest of this paper is organized as follows.  In Section \ref{sec:supervisedCD}, we present our block regularization approach to enforcing community structure in prediction problems.  Section \ref{sec:optimization-CD} presents an algorithm to solve the corresponding optimization problem.  In Section \ref{sec:theory} we study the theoretical performance of our method. In Section \ref{sec:simulations-CD}, we evaluate the performance for both recovering community structure and prediction accuracy on simulated data.  Section \ref{sec:data-CD} presents results for the COBRE schizophrenia dataset.   We conclude with a discussion and future work in Section \ref{sec:discus}.

%% file: method.tex
\section{Supervised community detection \label{sec:supervisedCD}}
We start by setting up notation.    Since the motivating application is to brain networks constructed from fMRI data, we focus on weighted undirected networks with no self-loops, although our approach can be easily extended to other network settings. We observe a sample of $m$ networks with $n$ labeled nodes that match across all networks  $A^{(1)}, \ldots, A^{(N)}$, and their associated response vector $\mathcal{Y}=(Y_1,\ldots,Y_N)$, with $Y_m\in\Bbb{R}$, $m = 1, \dots, N$.    Each network here is represented by its weighted  adjacency matrix $A^{(m)} \in \Bbb{R}^{n\times n}$, satisfying $A^{(m)} = (A^{(m)})^T$ and $\text{diag}( A^{(m)})= 0$. 

The inner product between two matrices $U$ and $V$ 
is denoted by $\left\langle U,V\right\rangle=\operatorname{Tr}(V^TU)$.  Let   $\|M\|_p=\left(\sum_{i=1}^{n_1}\sum_{j=1}^{n_2} M_{ij}^p\right)^{1/p}$ be the entry-wise $\ell_p$ norm of a matrix $M\in\Bbb{R}^{n_1\times n_2}$;  in particular,  $\|\cdot\|_2 = \|\cdot\|_F$ is the Frobenius norm.  Given a symmetric matrix $M\in\real^{n\times n}$ with eigenvalues  $\lambda_1, \ldots, \lambda_n$ ordered so that $|\lambda_1|>\ldots>|\lambda_n|$, we denote by $\lambda_{\max}(M) := \lambda_{1}$ and $\lambda_{\min}(M) = \lambda_N$, the largest and smallest eigenvalues in absolute value.

For simplicity, we focus on linear prediction methods, considering that interpretation the most important goal in the motivating application;  however, the linear predictor can be replaced by another function as long as it is convex in the parameters of interest.   For a given matrix $A$, the corresponding response $Y$ will be predicted using a linear combination of the entries of $A$.   We can define a matrix of coefficients $B\in\Bbb{R}^{n\times n}$, an intercept $b\in\Bbb{R}$ and a prediction  loss function $\ell$ by 
\begin{equation}
\ell(B) = \sum_{m=1}^N\tilde{\ell}\left(Y_m, \langle A^{(m)}, B\rangle + b\right), \label{eq:lossfunction-innerprod}
\end{equation}
where $\tilde{\ell}$ is a prediction loss determined by the problem of interest; in particular, this framework includes generalized linear models, and can be used for continuous,  binary, or categorical responses.   The entry $B_{ij}$ of $B$ is the coefficient of $A_{ij}$, and since the networks are undirected with no self-loops, we require $B=B^T$ and $\text{diag}(B)=0$ for identifiability.
The conventional approach is to minimize an objective function that consists of a loss function plus a regularization penalty $\Omega_\lambda$, with a tuning parameter $\lambda$ which controls the amount of regularization.  The penalty is important for making the solution unique for small sample sizes and high dimensions and for imposing structure on coefficients;  popular choices include ridge, lasso, or the elastic net penalties, all available in the glmnet R package \citep{friedman2009glmnet}.

\begin{remark}
	The intercept $b$ in \eqref{eq:lossfunction-innerprod}  is important for accurate prediction, but since in many situations it can be removed by centering and otherwise is easy to optimize over, we omit the intercept in all derivations that follow, for simplicity of notation.  
\end{remark}

As discussed in the Introduction,  communities in brain networks  correspond to groups of nodes with similar functional connectivity, and thus it is reasonable to assume that edges within a network ``cell'' (edges connecting a given pair of communities or edges within one community) have a similar effect on the response. This structure  can be explicitly enforced in the matrix of coefficients $B$.
Suppose the nodes are partitioned into $K$ groups $\mathcal{C}_1,\ldots, \mathcal{C}_K\subset\{1,\ldots,n\}$ such that $\mathcal{C}_i\cap \mathcal{C}_j=\emptyset$ and $\bigcup_{k=1}^K\mathcal{C}_{k}=\{1,\ldots, n\}$.   If the value of $B_{ij}$ depends only on the community assignments of nodes $i$ and $j$, we can represent all the coefficients in $B$ with a $K \times K$ matrix $C$, with $B_{ij} = C_{kl} \text{ if }i\in\mathcal{C}_k \text{  and }j\in\mathcal{C}_l$. 
Equivalently, define a binary membership matrix $Z\in\{0,1\}^{n\times K}$ setting $Z_{ik}=1$ if $i\in\mathcal{C}_k$, and 0 otherwise.  Then $B$ can be written as
\begin{equation}
B = ZCZ^T. \label{eq:B-blockstructure}
\end{equation}
This enforces equal coefficients for all the edges within one network cell (see Figure \ref{fig:B-blockconstant}).    This definition looks similar to the stochastic block model (SBM) \citep{holland1983stochastic}, with the crucial difference that here $B$ is not a matrix of edge probabilities, but the matrix of coefficients of a linear predictor, and hence its entries can take any real values.

Let $\mathcal{Z}_{n,K}$  
be the set of all $n \times K$ membership matrices.   Suppose for the moment  we are given a membership matrix $Z\in\mathcal{Z}$.  We can enforce cell-constant coefficients by adding a constraint on $B$ to the optimization problem, solving 
\begin{align}
\min_{C} \  & \{ \ell(B) + \Omega (B) \} \label{eq:blockConstrainedProblem-noZ}\\
\text{subject to} \  & B=ZCZ^T, \  C\in\Bbb{R}^{K\times K}, \ C=C^T,\nonumber
\end{align}
or, using the fact $\left\langle A, Z CZ^T\right\rangle = \left\langle Z^TAZ,  C\right\rangle $,  we can restate the optimization problem in terms of $C$ as 
\begin{eqnarray}
       \min_C  & \left\{ \sum_{m=1}^N \tilde\ell \left(Y_m, \langle Z^TA^{(m)}Z, C\rangle \right) + \Omega_{\lambda}(ZCZ^T)\right\} . \label{eq:lossfunction-innerprod-C}
\end{eqnarray}
This  reduces the number of coefficients to estimate from $n(n-1)/2$ to only  $K(K+1)/2$, which allows for much better interpretation of network cells' effects and for faster optimization.

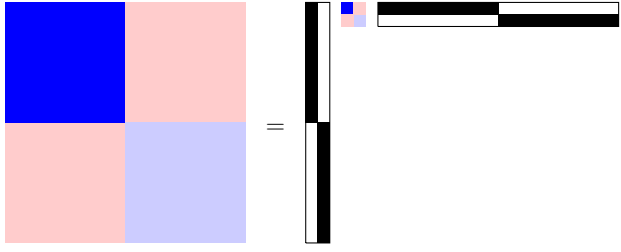
\begin{figure}
	\centering
	\begin{tikzpicture}[scale=0.8, transform shape]
	\draw (5,0) -- (5,4) -- (5.4,4) -- (5.4,0) -- (5,0);
	
	\fill[red!20!white] (0,0) rectangle (2,2);
	\fill[blue!] (0,2) rectangle (2,4);
	\fill[red!20!white] (2,2) rectangle (4,4);
	\fill[blue!20!white] (2,0) rectangle (4,2);
	
	\node[xshift=4.5cm,yshift=1.9cm] (equal)  {=};
	
	\fill[black!] (5,2) rectangle (5.2,4);
	\fill[black!] (5.2,0) rectangle (5.4,2);
	
	\fill[blue!] (5.6,3.8) rectangle (5.8,4);
	\fill[blue!20!white] (5.8,3.6) rectangle (6,3.8);
	\fill[red!20!white] (5.6,3.6) rectangle (5.8,3.8);
	\fill[red!20!white] (5.8,3.8) rectangle (6,4);
	
	\draw (6.2,3.6) -- (6.2,4) -- (10.2,4) -- (10.2,3.6) -- (6.2,3.6);
	
	\fill[black!] (6.2,3.8) rectangle (8.2,4);
	\fill[black!] (8.2,3.6) rectangle (10.2,3.8);
	\end{tikzpicture}\\
	\caption{Factorization of the coefficient matrix $B$ into $ZCZ^T$. 
        } \label{fig:B-blockconstant}
\end{figure}

For many choices of the penalty $\Omega_\lambda$, such as lasso or ridge,  the optimization problem \eqref{eq:lossfunction-innerprod-C} is a standard prediction loss plus penalty problem with only $K(K+1)/2$ different parameters, which is easy to solve.    When the number of parameters $K(K+1)/2$ is large relative to the sample size $m$, regularization (setting $\lambda > 0$) is required for the solution to be well-defined.

In reality, the community membership matrix is $Z$ is unknown, and our goal is to find a partition into communities that will give the best prediction.  Thus we need to optimize the objective function over $Z$ and $C$ jiontly, solving
\begin{align}
\min_{Z, C}   & \left\{ \sum_{m=1}^N \ell \left(Y_m, \langle Z^TA^{(m)}Z, C\rangle \right) + \Omega_{\lambda}(ZCZ^T)\right\} 
\label{eq:blockConstrainedProblem}\\
\text{subject to }   & C\in\Bbb{R}^{K\times K} , \ C = C^T 
 Z\in\{0,1\}^{n\times K} \ , \ Z\textbf{1}_K = \textbf{1}_n.\nonumber
\end{align}

The formulation  \eqref{eq:blockConstrainedProblem} is aimed at finding the best community assignments for predicting a particular response, recognizing that neuroscientists now understand these communities are dynamic and vary with time and activity \citep{Salehi2019}.  Enforcing block structure on the coefficients has the effect of grouping  edges with similar predictive function into cells by clustering the associated nodes.    Approaches for simultaneously predicting a response and clustering predictors have been proposed; for example, \cite{bondell2008simultaneous} introduced a penalty that achieves this goal via fused lasso.    Our goal, however, is not just clustering predictors (edges); it is partitioning the brain network into meaningful regions, which requires clustering {\em nodes}. 

The number of ``communities'' $K$ in \eqref{eq:blockConstrainedProblem} can be viewed as a tuning parameter that controls the amount of regularization.   In practice, the value of $K$ is unknown and can be chosen by cross-validation, as is commonly done with all regularization tuning parameters.   Alternatively, one can view our method as a tool to enforce an interpretable structure on the coefficients while maintaining a reasonable prediction performance.  Then it makes sense to set $K$ to match those found in commonly used brain atlases, typically 10-20, in order to obtain an interpretable solution that can be easily compared with an existing brain atlas.

%% file: algorithm.tex
\section{Algorithms for block-structured regularization \label{sec:optimization-CD}}

Solving the problem \eqref{eq:blockConstrainedProblem} exactly  is computationally infeasible, at least naively, since $Z$ can take on $K^n$ different values. Instead, we propose an iterative optimization algorithm based on the alternating direction method of multipliers (ADMM) \citep{boyd2011distributed}.  Since the objective function is not convex, having a good initial value is crucial in practice, even though in principle one can initialize with any membership matrix $Z$.   Some reasonable choices for the initial value include one of the previously published brain parcellations, or the output of some unsupervised community detection method designed for a sample of networks, but both of these ignore the response $\mathcal{Y}$.    We initialize instead with a solution computed by a particular spectral clustering algorithm which provides a fast approximate solution to the problem while taking the response into account, and can be used to initialize the ADMM optimization procedure.

\subsection{Spectral clustering for the sum of squares loss}
The spectral clustering algorithm works with the sum of squares loss;  we will discuss initialization for other losses at the end of this section.   Suppose we want to solve the constrained problem \eqref{eq:blockConstrainedProblem} with the sum of squares loss function.   Assuming for simplicity that the network matrices and the responses are centered,  that is, $\bar{A} = \frac{1}{N}\sum_{m=1}^N A^{(m)}=0$ and $\bar{Y}=\frac{1}{N}\sum_{i=1}^N Y_m=0$,   we can write the loss function as
\begin{equation}
\ell(B) = \frac{1}{2m}\sum_{i=1}^m \left(Y_i - \text{Tr}\left( A^{(i)} B\right) \right)^2. \label{eq:leastSquaresBlocks}
\end{equation}

We first present the very simple algorithm, and then explain the intuition behind it.   Let $\hat{\Sigma}^{\mathcal{A},\mathcal{Y}}$ be the $n\times n$ matrix of coefficients from the simple linear regression of the response on the weight of the edge $(i,j)$, defined by 
\begin{equation}
\left(\hat{\Sigma}^{\mathcal{A},\mathcal{Y}}\right)_{ij} = {\widehat{\text{Cov}}(A_{ij}, Y)} =  \frac{1}{N}\sum_{m=1}^NY_mA^{(m)}. \label{eq:sigmaAY-block}
\end{equation}
Next, we perform spectral clustering on $\hat{\Sigma}^{\mathcal{A},\mathcal{Y}}$. That is, we first compute the $K$ leading eigenvectors of $\hat{\Sigma}^{\mathcal{A}, \mathcal{Y}}$, denoted by $V\in\mathbb{R}^{n\times K}$, and then cluster the rows of $V$ using  $K$-means, as summarized in Algorithm \ref{alg:SC-leastsquares}.  Cluster assignments give a membership matrix $\hat{Z}^{(0)}$, which can be used either as a regularizing constraint in problem \eqref{eq:blockConstrainedProblem-noZ}, or as an initial value for Algorithm \ref{alg:ADMM-SCD},  introduced in the next section. 

\begin{algorithm}
	\caption{Spectral clustering for the sum of squares loss}
	\begin{algorithmic} 
		\Input Training sample $\{(A^{(1)},Y_1),\ldots,(A^{(N)},Y_N)\}$ centered and standardized; number of communities $K$.
		\begin{enumerate}
			\item Compute $\hat{\Sigma}^{\mathcal{A},\mathcal{Y}} $ as in equation \eqref{eq:sigmaAY-block}.
			\item Compute $V$,  
                          the $n \times K$  matrix of  $K$ leading eigenvectors of $\hat{\Sigma}^{\mathcal{A},\mathcal{Y}}$ corresponding to the $K$ largest (in absolute value) eigenvalues.
                          
			\item Run $K$-means to cluster the rows of $V$ into $K$ groups. 
		\end{enumerate}
		\Output $\hat{Z}^{(0)}=\hat{Z}$, the membership matrix constucted from $K$-means clusters.  
	\end{algorithmic}
	\label{alg:SC-leastsquares}
\end{algorithm}

The intuition behind this approach comes from the case of uncorrelated predictors, in the sense that for any pair of different edges $(i_1, j_1)$ and $(i_2, j_2)$
$$\sum_{m=1}^N A_{i_1j_1}^{(m)}A^{(m)}_{i_2j_2} =0.$$
In this case, it is well known that the least squares solution with no constraints is given by $\hat B = \hat{\Sigma}^{\mathcal{A}, \mathcal{Y}}$.   Then to obtain the best predictor with the block-structured constraints, we can solve the optimization problem
\begin{eqnarray}
\min_{Z, C}  & \frac{1}{2}\left\| ZCZ^T -  \Sigma^{\mathcal{A}, \mathcal{Y}} \right\|_F^2 \label{eq:blockConstrainedProblem-leastsquares}\\
\text{subject to} & Z\in\mathcal{Z}_{n,K}, \quad C\in\real^{K\times K}.\nonumber
\end{eqnarray}
This problem is still computationally hard since $Z$ has binary entries.    Instead, we solve~\eqref{eq:blockConstrainedProblem-leastsquares} over the space of continuous  matrices $\tilde{Z}\in\real^{n\times K}$, and then project the solution to the discrete constraint space  in~\eqref{eq:blockConstrainedProblem-leastsquares}. 
The continuous solution is given by the leading $K$ eigenvectors of  $\Sigma^{\mathcal{A}, \mathcal{Y}}$, another well-known fact.   Let $\tilde Z$ be the $n \times K$ matrix of eigenvectors corresponding to the $K$ largest eigenvalues of $\Sigma^{\mathcal{A}, \mathcal{Y}}$, in absolute value. 
To project $\tilde{Z}$ onto the feasible set of membership matrices, we find the best Frobenius norm approximation to it by a matrix with $K$ unique rows, which is equivalent to minimizing the corresponding $K$-means loss function.  This procedure(Algorithm~\ref{alg:SC-leastsquares}) is analogous to spectral clustering in the classical community detection problem (e.g., \cite{rohe2011spectral}), with the crucial difference in that it uses the response values $\mathcal{Y}$ and not just the network itself.  


In general, it might be unrealistic to assume that the edges are uncorrelated; nevertheless, many methods constructed based on this assumption have surprisingly good performance in practice even when this assumption does not hold, the so-called naive Bayes phenomenon \citep{bickel2004some}.  In Section~\ref{sec:theory}, we will show that even when there is a moderate correlation between edges in different cells,  this spectral clustering algorithm can accurately recover the community memberships.

\begin{remark}
For a general loss function of the form in Equation \eqref{eq:lossfunction-innerprod}, we compute the matrix $\widetilde{\Sigma}^{\mathcal{A},\mathcal{Y}}$ by solving the marginal problem for each coefficient of $B$ given by
\begin{equation}
\left(\widetilde{\Sigma}^{\mathcal{A},\mathcal{Y}}_{ij}, \tilde b_{ij}  \right)= \argmin_{B_{ij}, b} \sum_{m=1}^N\tilde{\ell}(Y_m, A^{(m)}B_{ij} + b). \label{eq:marginal}
\end{equation}
For the least squares loss function, $\widehat{\Sigma}^{\mathcal{A},\mathcal{Y}}=\widetilde{\Sigma}^{\mathcal{A},\mathcal{Y}}$, but more generally the matrix $\widetilde{\Sigma}^{\mathcal{A}, \mathcal{Y}}$ might be more expensive to compute.
 For generalized linear models (GLMs),  $\widehat{\Sigma}^{\mathcal{A},\mathcal{Y}}$ can be viewed as the estimate of the first iteration of the standard GLM-fitting algorithm, Iteratively Reweighted Least Squares, to solve Equation~\eqref{eq:marginal}. Thus, we also use Algorithm~\ref{alg:SC-leastsquares} for GLMs.

\end{remark}



\subsection{Iterative optimization with ADMM}


The ADMM is a popular flexible method for solving convex optimization problems with linear constraints which makes it easy to incorporate additional penalty functions. Although ADMM is limited to linear constraints, some heuristic extensions to more general settings with non-convex constraints   have been proposed \citep{diamond2016general}, showing good empirical performance.   Here we introduce an iterative ADMM algorithm for approximately solving \eqref{eq:blockConstrainedProblem}. 

Let $\mathcal{W}_{n,K}$ be the set of $n\times n$ matrices with $K$ blocks, which can be written as 
\[\mathcal{W}_{n,K}=\left\{B\in\mathbb{R}^{n\times n}\left|B=ZCZ^T, Z\in\mathcal{Z}_{n,K}, C\in\mathbb{R}^{K\times K}\right.\right
\}.\]  
For  $W\in\mathcal{W}_{n,K}$ and a dual variable $V\in\mathbb{R}^{n\times n}$, the augmented Lagrangian of the problem can be written as 
\begin{equation}
L_{\rho}(B, V, W) = \ell(B) + \Omega(B) +  \left\langle V, B-W\right\rangle + \frac{\rho}{2}\|B-W\|^2_F,
\end{equation}
with $\rho > 0$ a parameter of the optimization algorithm. Given initial values $B^{(0)}, V^{(0)}$ and $W^{(0)}$, the ADMM updates are as follows:   
\begin{align}
  B^{(t)} & =  \argmin_{B\in\mathbb{R}^{n\times n}} L_\rho(B, V^{(t-1)}, W^{(t-1)}) , 
            \label{eq:alg_lossstep}\\
  W^{(t)} & =  \argmin_{W\in\mathcal{W}} L_\rho(B^{(t)}, V^{(t-1)}, W) , 
            \label{eq:alg_clusteringstep}\\
V^{(t)}  & =  V^{(t-1)} + \rho\left(B^{(t)}-W^{(t)}\right). \label{eq:dualADMM}
\end{align}
Equation \eqref{eq:alg_lossstep} depends on the loss $\ell$ and the penalty $\Omega$, and can be expressed as
\begin{align}
B^{(t)} 
& =  \argmin_{B\in\mathbb{R}^{n\times n}}\left\{  \ell(B) + \Omega(B) + \frac{\rho}{2}\left\|B-\left(W^{(t-1)}-\frac{1}{\rho}V^{(t-1)}\right)\right\|_F^2\right\} . \label{eq:solve-for-B-ADMM}
\end{align}
If $\ell$ and $\Omega$ are convex, as is generally the case, then \eqref{eq:solve-for-B-ADMM} is also convex, and in some cases it is possible to express it as a regression loss with a ridge penalty. Step \eqref{eq:alg_clusteringstep} is a non-convex combinatorial problem because of the membership matrix $Z$.  It can be rewritten as 
\begin{equation}
W^{(t)} = \argmin_{W\in\mathcal{W}}\left\| W - \left(B^{(t)}+\frac{1}{\rho}V^{(t-1)}\right)\right\|_F^2,\label{eq:optimize-W}
\end{equation}
which is analogous to problem \eqref{eq:blockConstrainedProblem-leastsquares}, and we can again 
approximately solve it by spectral clustering on the matrix $B+\frac{1}{\rho}V$. 
Once we obtain the membership matrix $Z^{(t)}$ from spectral clustering, the solution of equation \eqref{eq:optimize-W} is given by
\begin{equation}
W^{(t)} = Z^{(t)} (Z^{(t)^T}Z^{(t)})^{-1} \left(B^{(t)}+\frac{1}{\rho}V^{(t-1)}\right)(Z^{(t)^T}Z^{(t)})^{-1} Z^{(t)^T}.
\label{eq:solution-W-Z}
\end{equation}
Finally, step \eqref{eq:dualADMM} updates the dual variables in closed form. We iterate this algorithm until the primal and the dual residuals,
$\frac{1}{n}\|B^{(t)}-W^{(t)}\|_F$ and $\frac{\rho}{n}\|W^{(t)}-W^{(t-1)}\|_F$, are both smaller than some tolerance $\epsilon^{\text{TOL}}$ \citep{boyd2011distributed}.
 The steps of the ADMM algorithm are summarized in Algorithm  \ref{alg:ADMM-SCD}. Because this is a non-convex problem, it is important to start with a good initial value, so we use Algorithm~\ref{alg:SC-leastsquares} to initialize the method.

\begin{algorithm}
	\caption{Iterative optimization with ADMM}
	\begin{algorithmic} 
		\Input $\{(A^{(1)},Y_1),\ldots,(A^{(N)},Y_N)\}$,  $K$, $\rho$
		\Initialize   Set $Z=Z^{(0)}$ from Algorithm~\ref{alg:SC-leastsquares}, set $W^{(0)}$ to be the solution of problem \eqref{eq:blockConstrainedProblem} using, and let $V^{(0)}=0_{n\times n}$.
		\Iterate for $t=1,2,\ldots$ until convergence 
		\begin{enumerate}
			\item Compute $B^{(t)}$ using \eqref{eq:solve-for-B-ADMM}.
			\item Update $W^{(t)}$ using spectral clustering:
			\begin{enumerate}
				\item Set $V$ be the $K$ leading eigenvectors of $B^{(t)} + \frac{1}{\rho}V^{(t-1)}$; 
				\item Run $K$-means to cluster the rows of $V$ into $K$ groups $\hat{C}_1,\ldots,\hat{C}_K$ ; 
				\item Set $Z^{(t)}_{ik}=1$ if $i\in\hat{\mathcal{C}}_k$ and 0 otherwise; 
				\item Update $W^{(t)}$ using \eqref{eq:solution-W-Z}.
			\end{enumerate}
			\item Update $V^{(t)}$ using \eqref{eq:dualADMM}.
		\end{enumerate}
		\Output $\hat{B}=B^{(t)}$ and $\hat{Z}=Z^{(t)}$
	\end{algorithmic}
	\label{alg:ADMM-SCD}
\end{algorithm}

The algorithmic parameter $\rho$ in Algorithm \ref{alg:ADMM-SCD} controls the size of the primal and dual steps. A larger $\rho$ forces $B^{(t)}$ to stay closer to $W^{(t-1)}$, and then the community assignments  $Z^{(t-1)}$ are less likely to change at each iteration.  For convex optimization, the ADMM is guaranteed to converge to the optimal value for any $\rho>0$.   In our case, if $\rho$ is too small, the algorithm may not converge, whereas if $\rho$ is too large, the algorithm may never move away  from the initial community assignment $Z^{(0)}$. In practice, we run Algorithm \eqref{alg:ADMM-SCD} for a set of different values of $\rho$, and choose the solution $\hat{B}$ that gives the best objective function value in equation \eqref{eq:blockConstrainedProblem}.


%% file: theory.tex
\section{Theory \label{sec:theory}}

We obtain theoretical performance guarantees for the spectral clustering method  when the sample of networks is distributed according to a weighted stochastic block model and our motivating model holds, i.e., the response follows a linear model with a block-constant matrix of coefficients. This is reflected in the following assumption.

\begin{assumption}[Linear model with block-constant coefficients] \label{assump:linear-model}  Responses $Y_m$, $m = 1, \dots, N$ follow the linear model
	\[Y_m = \left\langle A^{(m)}, B\right\rangle + \sigma \epsilon_m,\]
	where $\sigma>0$ is an unknown constant,  $\epsilon_1,\ldots, \epsilon_N$ are independent sub-Gaussian random variables  with $$\Exp{\epsilon_m}=0, \quad\quad\quad \orlicz{\epsilon_m}=1 , $$
        and the matrix of coefficients $B\in\real^{n\times n}$ takes the form $B=ZCZ^T$, where $Z\in\mathcal{Z}_{n,K}$ is a membership matrix,  and $C\in\real^{K\times K}$ satisfies $\|C\|_F=\Theta(1)$.
\end{assumption}
Assumption \ref{assump:linear-model}  implies that the rank of $B$ is at most $K$, which corresponds to the number of communities.  For the theoretical developments only, we further assume that the value of $K$ is known, and the  goal is to recover the community membership matrix $Z$.   In practice, $K$ is selected by cross-validation.  

As is common in the high-dimensional regression literature, we assume that the predictor variables (edge weights) are distributed as sub-Gaussian random variables
\citep{zhou2009restricted, Raskutti2010,Rudelson2013}. Recall that the sub-Gaussian norm of a random variable is defined as
\begin{equation*}
    \orlicz{X} = \sup_{p\geq 1} p^{-1/2}\Exp{|X|^p}^{1/p}.
\end{equation*}
A random variable $X$ is said to be sub-Gaussian if $\orlicz{X} < \infty$.  
Using this notation, the weighted sub-Gaussian stochastic block model distribution for a single graph is defined as follows. 
\begin{definition}[Weighted sub-Gaussian stochastic block model]
	Given a positive integer $ K$, a community membership matrix $Z \in\mathcal{Z}_{n,K}$ with vertex community labels ${z}_1,\ldots, {z}_n\in[ K]$ such that ${z}_u=k$ if ${Z}_{uk}=1$, and  matrices $R, \Psi \in\real^{K\times K}$,  we say that an $n \times n$ matrix $A$ follows a weighted sub-Gaussian stochastic block model distribution,  denoted by $A\sim \operatorname{sGSBM}(R, \Psi, \{{Z}\})$ if for every $1 \le v < u \le n$, $A_{uv}$ is an independent sub-Gaussian random variable with
	\begin{equation*}
	  \Exp{A_{uv}} = R_{{z}_u {z}_v},  \quad  \operatorname{Var}\left(A_{uv}\right) = \Psi_{ {z}_u {z}_v}.
	\end{equation*}
\end{definition}

This definition is an extension of the classical stochastic block model with Bernoulli edges \citep{holland1983stochastic}, allowing for a much larger class of distributions for edge weights.   
Weighted stochastic block models with general edge distributions have been considered in the context of a single-graph community detection problem \citep{Xu2017}, as well as in the analysis of averages of multiple networks \citep{Levin2019}. In contrast, our model is for a regression setting, and requires control of the covariance structure, which we can achieve with just the sub-Gaussian edge distribution assumption.   This is formalized in the following assumption on the sample.

\begin{assumption}[Network sample with  common communities] \label{assump:network-distr}
Given $K$ and $Z\in\mathcal{Z}_{n,K}$ as defined in Assumption~\ref{assump:linear-model}, there are symmetric $K \times K$ matrices $R^{(1)},\ldots, R^{(N)}$, $\Psi^{(1)}, \ldots, \Psi^{(N)}$ such that for each $m\in[N]$
  \begin{equation*}
  A^{(m)}\sim \mathrm{sGSBM}(R^{(m)}, \Psi^{(m)}, Z).
  \end{equation*}
\end{assumption}

For the purpose of the theoretical analysis, we consider the matrices $R^{(m)}, \Psi^{(m)}, m\in[N]$  \emph{fixed} parameters, and the edges of the graphs $A^{(m)}_{uv}, m\in[N], u,v\in[n], u>v$ are \emph{random} variables distributed according to Assumption~\ref{assump:network-distr}. Let $n_1, \ldots, n_K$ be the shared community sizes, with $\sum_{k=1}^K n_k =n$, and let $n_{\min} := \min_kn_k$ and $n_{\max} := \max_k n_k$. 



The next assumption is  introduced to simplify notation, essentially without loss of generality, because we can always standardize the data. 
\begin{assumption}[Standardized covariates]
\label{assumption:centered}
For each $u,v\in[n]$, $u>v$,
\begin{align}
\Exp{ \frac{1}{N}\sum_{m=1}^N A^{(m)}_{uv} }   = 0, \  \  \  
\Exp{\frac{1}{N}\sum_{m=1}^N \left(A^{(m)}_{uv}\right)^2 }   = 1,
\end{align} 
that is, the expected sample mean and variance of the edges conditional on the sGSBM parameters $R^{(1)},\ldots, R^{(N)},\Psi^{(1)}, \ldots, \Psi^{(N)}$ are 0 and 1, respectively.  
\end{assumption}


Next, define matrices $\Pi, \Psi\in\real^{K\times K}$  corresponding to the cell  and edge sample variances, that is,
\begin{equation*}
\Pi_{ij} := \frac{1}{N}\sum_{m=1}^N \left(R^{(m)}_{ij}\right)^2, \ \ \  \Psi_{ij} := \frac{1}{N}\sum_{m=1}^N \Psi^{(m)}_{ij}.
\end{equation*}
Assumption~\ref{assumption:centered}  implies that
\begin{align*}
 \frac{1}{N}\sum_{m=1}^N R^{(m)} = 0, \ \ \  \Pi_{ij}  + \Psi_{ij}  = 1.
\end{align*}


By Assumption~\ref{assump:network-distr}, all the networks have the same community structure but potentially different expected connectivity matrices $R^{(m)}, m\in[N]$, allowing for individual differences.  While the edge variables within each network are independent,  the connectivity matrices induce a covariance between the edge variables across the sample, corresponding to the covariance of the linear regression design matrix. To define this covariance, we first introduce some notation. Let $p= n(n+1)/2$ be the total number of edges.
To represent the covariance between edges in matrix form, 
we write $[u,v]$ to denote the index in a vector of length $p$ corresponding to edge $(u,v)$, $u<v$. Denote by $\mathcal{P}:=\{(u,v)\in[n]\times [n], u<v\}$  the index set of edges pairs, and by $\mathcal{Q}:=\{(h,l)\in[K]\times [K],h\leq l \}$  the index set of cell pairs. Using this notation,  $\Sigma^{\mathcal{A}}\in\mathbb{R}^{p\times p}$ is defined as the covariance matrix of the edges across the networks, such that for a pair of edge variables $(u,v)$ and $(s,t)$ in $\mathcal{P}$, $\Sigma^{\mathcal{A}}_{[u,v], [s,t]} := \e\left[\frac{1}{N}\sum_{m=1}^N{A}^{(m)}_{uv}{A}^{(m)}_{st}\right]$. As a consequence of the block model distribution of the networks, the covariance between edge variables across the samples only depends on the community membership of the nodes involved, and is given by 
\begin{align}
	\Sigma^{\mathcal{A}}_{[u,v], [s,t]} &  = \left\{ \begin{array}{ll}
		\frac{1}{N}\sum_{m=1}^NR^{(m)}_{z_uz_v}R^{(m)}_{z_sz_t} & \text{if }(u,v)\neq (s,t)\text{ and }(z_u, z_v)\neq (z_s, z_t),\\
		& \\
		\Pi_{z_uz_v} & \text{if }(u,v)\neq (s,t)\text{ and }(z_u,z_v)  = (z_s, z_t),\\
		& \\
		\Pi_{z_uz_v} + \Psi_{z_uz_v} & \text{if }(u,v)= (s,t). \\
	\end{array}\right. \label{eq:covariance-communities}
\end{align}

The next proposition gives the specific form of the expectation of $\hat{\Sigma}^{\mathcal{A},\mathcal{Y}}$ defined in Equation~\eqref{eq:sigmaAY-block} under our model. 

\begin{proposition} \label{proposition:expected-sigma-AY}  Suppose that Assumptions \ref{assump:linear-model}, \ref{assump:network-distr} and \ref{assumption:centered}  hold. Define the symmetric matrix $F\in\real^{K\times K}$ as
\begin{equation*}
F_{jk} := 2\Psi_{jk} C_{jk} + 2\sum_{(s,t)\in\mathcal{P}}\left(\frac{1}{N}\sum_{m=1}^NR^{(m)}_{jk}R^{(m)}_{z_sz_t}\right) C_{z_sz_t}.
\end{equation*}
Then, the expected value of  $\hat{\Sigma}^{\mathcal{A}, \mathcal{Y}}$ can be written as 
\begin{equation}
	\Exp{\hat{\Sigma}^{\mathcal{A}, \mathcal{Y}}} = ZF Z^T - \mathrm{diag}(ZFZ^T). \label{eq:expected=sigma-A-Y}
\end{equation}
\end{proposition}

Proposition \ref{proposition:expected-sigma-AY} shows the expected value of $\hat{\Sigma}^{\mathcal{A}, \mathcal{Y}}$ has the same community structure as the matrix of coefficients $B$, and thus the leading  eigenvectors of $\Exp{\hat{\Sigma}^{\mathcal{A}, \mathcal{Y}}}$  contain all the necessary information for recovering communities, provided that   $F$ is full rank. 
We state this as a formal condition on $F$.  
\begin{condition} 
\label{condition:bounded-covariance} 
There exist a constant $\kappa>0$  such that  $|\lambda_{\min}\left(F\right)| \geq \kappa$.
\end{condition}
This is a mild condition which is sufficient for consistent community detection with spectral clustering, and it typically holds whenever $\Psi\circ C$ is full rank.   If the smallest eigenvalue of $F$ grows with the size of the graph at the rate of $\Omega(n^2_{\min})$, we will show that spectral clustering achieves a faster consistency rate for community detection.  In particular, this faster rate is achieved when the covariance across the sample between edges from different cells is sufficiently small; the next proposition gives a sufficient condition for that.

\begin{proposition} \label{prop:condition-covariance}
Suppose that $n_{\min}^2=\omega(n_{max})$ and there exist constants $\delta\in(0,1], \gamma >0$ such that $|\lambda_{\min}(\Pi\circ C')|\geq \gamma$ and
\begin{equation}
\max_{(u,v)\in\mathcal{P}} \left|
\sum_{\substack{(h,l)\in\ \mathcal{Q},\\ (h,l)\neq (z_u, z_v) }}\  \sum_{\substack{(s,t)\in\mathcal{P},\\(z_s, z_t) = (h,k)}}  \Sigma^{\mathcal{A}}_{[u,v][s,t]} B_{st} \right| \leq (1-\delta)\frac{n^2_{\min}}{K}|\lambda_{\min}(\Pi\circ C')|,
\label{eq:covarianceR-bound}
\end{equation}
 where $C'_{jk} := (1+ \mathbbm{1}\{j\neq k\})C_{jk}$.
Then, Condition~\ref{condition:bounded-covariance} holds, and
\begin{equation*}
    |\lambda_{\min}(F)| = \Omega\left(n^2_{\min}\right) .  
\end{equation*}
\end{proposition}

The left hand side of Equation~\eqref{eq:covarianceR-bound} is the sum of off-diagonal entries of $\Sigma^\mathcal{A}$ corresponding to pairs of edges located in different cells, multiplied by  the coefficient matrix $B$. Equation~\eqref{eq:covarianceR-bound} trivially holds when these entries of $\Sigma^{\mathcal{A}}$ are all zero.  This condition on the structure of the covariance between predictors is similar in spirit to conditions for support recovery for sparse  regularized estimators \citep{wainwright2009sharp,zhao2006model}.

\begin{remark} \label{remark:SBM} Assumption~\ref{assump:network-distr} requires the community structure in all the graphs to be the same.  A shared community structure is a  common assumption in studying multiple networks with labeled vertices  \citep{holland1983stochastic,han2014consistent}; in our application of interest, this assumption is motivated by the fact that each vertex corresponds to the same location in the brain for all the networks, obtain from a common brain atlas, and the organization of nodes into brain systems, once they are mapped onto a common atlas, remains the same across subjects.   Having the same community structure in the coefficients effectively allows us to look for the influence of network cells, representing connectivity within or between brain systems, instead of between individual nodes.  Alternatively, Assumption \ref{assump:network-distr} can be replaced with assumptions about the structure of the covariance  $\Sigma^\mathcal{A}$ analogous to Proposition \ref{prop:condition-covariance}.
\end{remark}

The next theorem quantifies the distance between the eigenspace of $\hat{\Sigma}^{\mathcal{A},\mathcal{Y}}$ and the eigenspace of its expectation, which will in turn allow us to bound the error in community detection.    Denote by $\psi := \max_{j,k}\Psi_{jk}$ to the largest individual edge variance and by $\pi^\ast := \max_{j,k\in[K]}\Pi_{jk}$ to the largest cell variance.

\begin{theorem} \label{thm:eigenvectors}
Suppose that weighted networks $A^{(1)},\ldots,A^{(N)}$ satisfy Assumptions \ref{assump:network-distr} and \ref{assumption:centered}, and  responses $Y_1,\ldots,Y_N$ follow the linear model of Assumption \ref{assump:linear-model}. Let $\hat{V}$ and $V$ be the $K$ leading eigenvectors of $\hat{\Sigma}^{\mathcal{A}, \mathcal{Y}}$  and $\Exp{\hat{\Sigma}^{\mathcal{A}, \mathcal{Y}}}$, respectively.  Suppose Condition  \ref{condition:bounded-covariance} holds for all $n>n'$ for some constant $n' > 0$.
Set $\lambda^\ast := \left|\lambda_{\min}\left(\frac{1}{n^2_{\min}}F\right)\right|$. Then, 
\begin{equation}
\Exp{\min_{O\in\mathcal{O}_K}\|\hat{V} - VO\|_F} \lesssim \frac{\sqrt{Kn}(K^2n_{\max}^2\pi^\ast\psi +
     \sigma +
     n_{\max}\psi^2)}{n_{\min}\sqrt{N}}\min\left\{ \frac{1}{n^2_{\min}\lambda^\ast}, 1\right\},
     \label{eq:theorem-eigenvec-distance}
\end{equation} 
where $\mathcal{O}_K$ denotes the set of $K\times K$ orthogonal matrices.
\end{theorem}

The right hand side of Equation~\eqref{eq:theorem-eigenvec-distance} bounds the distance between the subspaces of the eigenvectors as the minimum of two different terms. The first one corresponds to the setting  $|\lambda_{\min}(F)|=\Omega(n^2_{\min})$, which happens, for example, when Equation~\eqref{eq:covarianceR-bound} holds. The expected subspace estimation error achieves a faster rate of convergence in this regime. On the other hand, when $\lambda^\ast$ goes to zero but Condition \ref{condition:bounded-covariance} still holds, we still get convergence, albeit at a slower rate.  

The final step in recovering communities involves clustering the rows of the matrix $\hat{V}$ into $K$ groups using $K$-means.  Theorem~\ref{thm:eigenvectors} leads to a bound on the clustering error of $K$-means \citep{rohe2011spectral} measured by the distance between the estimated membership matrix $\hat{Z}$ and $Z$, completing the proof of consistency of the spectral clustering method.

\begin{theorem} \label{theorem:communityerror}
Suppose that the conditions of Theorem \ref{thm:eigenvectors} hold,  and $\frac{n_{\max}}{n_{\min}}=O(1)$. 
Then, 
\begin{equation*}
    \Exp{\min_{O\in\mathcal{P}_K}\|\hat{Z} - ZO\|_F} \lesssim \min \left\{\frac{K^3(1+\sigma/n^2)}{\sqrt{N}\lambda^\ast}, \frac{K(n^2+\sigma)}{\sqrt{N}} \right\},
\end{equation*}
where $\mathcal{P}_K$ denotes the set of $K\times K$ permutation matrices. 
\end{theorem}
The rates show, as one would expect, that we do better when the sample size $N$ grows, and worse as $K$ grows (with everything else fixed), since the number of coefficients to estimate is proportional to $K^2$.  On the other hand, when the eigenvalues of $F$ are sufficiently large,
if $K$ and $N$ are fixed, the rate improves as the number of vertices $n$ grows, since that gives us more edges per network cell to estimate each regression coefficient.  


{\bf Example.} 
To illustrate the theory, consider a simple example of networks with $K=2$ communities of equal size, and $R^{(m)}_{11} = R^{(m)}_{22} = p^{(m)}$ and $R^{(m)}_{21} = R^{(m)}_{12} = q^{(m)}$. Define the variables
\begin{align*}
	\xi_1:=& \frac{1}{N}\sum_{m=1}^N (R_{11}^{(m)})^2 = \frac{1}{N}\sum_{m=1}^N (R_{22}^{(m)})^2 = 
	\frac{1}{N}\sum_{m=1}^N\left(p^{(m)}\right)^2,\\
	\xi_2:=& \frac{1}{N}\sum_{m=1}^N (R_{12}^{(m)})^2 =  \frac{1}{N}\sum_{m=1}^N\left(q^{(m}\right)^2,\\
	\xi_3:=& \frac{1}{N}\sum_{m=1}^N R_{11}^{(m)}R^{(m)}_{12} = \frac{1}{N}\sum_{m=1}^N R_{22}^{(m)}R^{(m)}_{12} = \frac{1}{N}\sum_{m=1}^Np^{(m)}q^{(m)},
\end{align*}
It is straightforward to calculate that $\Psi_{11} = \Psi_{22} = 1 - \xi_1$, $\Psi_{12} = \Psi_{21} = 1-\xi_2$, and 
the entries of $F$, defined in Proposition~\ref{proposition:expected-sigma-AY}, are given by 
\begin{equation}
F_{jk} = \left\{\begin{array}{ll}
2\Psi_{11}C_{jj}+ \left[\xi_1\frac{n(n-2)}{4}(C_{11} + C_{22}) + \xi_3 \frac{n^2}{2}C_{12} \right] & \text{if }j=k,\\
2\Psi_{12}C_{12} +  + \left[\xi_3\frac{n(n-2)}{4}(C_{11} + C_{22}) + \xi_2 \frac{n^2}{2}C_{12} \right]& \text{if }j\neq k.
\end{array}
\right. \label{eq:matrix-F}
\end{equation}
Let $\lambda := |\lambda_{\min}(C)|$ be the smallest eigenvalue of the coefficients matrix $C$ and consider the following scenarios on the cell covariance.
\begin{enumerate}
	\item[a)] Networks are an i.i.d.\ sample, that is, $p^{(m)}=p$ and $q^{(m)} = q$ for all $m \in[N]$. The centering required by by Assumption~\ref{assumption:centered} enforces $p^{(m)}=q^{(m)}=0$ , which implies that $F = 2C$. In this scenario, all edges are uncorrelated across the sample, and Condition~\ref{condition:bounded-covariance} holds if $\lambda=\Omega(1)$, which according to Theorem~\ref{theorem:communityerror} implies that the expected error in the estimation of the memberships is
	\begin{equation}
		\Exp{\min_{O\in\mathcal{P}_2}\|\hat{Z} - ZO\|_F} \lesssim \frac{n^2+\sigma}{\sqrt{N}}. \label{eq:example1-rate}
	\end{equation}
	
	\item[b)]  The distributions vary across the sample, but $p^{(m)}= q^{(m)}$ for all $m$. In this scenario, there is no community structure in the networks, which causes edges across different cells to be correlated, 
	 with $\xi_1 = \xi_2 = \xi_3$, and hence $F=2(1-\xi_1)C + \xi'\textbf{1}\textbf{1}^T$, with $\xi' \asymp \xi_1n^2\|C\|_1$. Using Weyl's inequality, the smallest eigenvalue of $F$ is bounded below by 
	\begin{equation*}
		|\lambda_{\min}(F)| \gtrsim 2(1-\xi_1)\lambda - \xi_1n^2\|C\|_1.
	\end{equation*}
	Condition~\ref{condition:bounded-covariance} is satisfied if $\xi_1= O\left(\frac{\lambda}{\lambda + n^2\|C\|}\right)$, in which case the expected membership error bound is the same as in Equation~\eqref{eq:example1-rate}. In this scenario of strong correlation between edges on different cells, 
	the proportion of the edge variance corresponding to the community effect $\xi_1$ needs to shrink as the number of vertices increases for the spectral clustering algorithm to guarantee recovery.  
	
	\item[c)] The distributions vary across the sample, and $p^{(m)} \neq q^{(m)}$ for all $m$.  In this scenario,  spectral clustering method can take advantage of the community structure if the cells are weakly correlated. Using Equation~\eqref{eq:matrix-F}, the eigenvalues of $F$ can be obtained exactly for a general $C$, but to avoid cumbersome calculations, we consider two special full-rank cases. Suppose first that $C = 2I$, so for a given network the response is the sum of edges connecting nodes within communities.   The smallest eigenvalue of $F$ is then given by
	\begin{equation*}
	\lambda_{\min}(F) = \frac{8\Psi_{11} + \xi_1n(n-1) - |\xi_3|n(n-1)}{2},
	\end{equation*}
	and hence, if $\xi_1\geq \gamma$ and $|\xi_3|\leq (1-\delta) \xi_1$ for some $\delta, \gamma\in(0,1]$ then $\lambda_{\min}(F)=\Omega(n^2)$, which results in
	\begin{equation}
	\Exp{\min_{O\in\mathcal{P}_2}\|\hat{Z} - ZO\|_F} \lesssim \frac{1+\sigma/n^2}{\sqrt{N}}, \label{eq:example3-rate}
	\end{equation}
	a better rate of recovering communities.   
	Next, consider $C = 2(\textbf{1}\textbf{1}^T-I)$,  so that the response is proportional to the sum of edges connecting nodes in different communities. Here, we can apply Proposition~\ref{prop:condition-covariance}, and as long as $|\xi_2|\geq \gamma$ and $|\xi_3| \leq (1-\delta)|\xi_2|/2$ for some constants $\delta, \gamma\in(0,1]$, we get the same improved rate as in Equation~\eqref{eq:example3-rate}.
\end{enumerate}

%% file: simulations.tex
\section{Numerical results on simulated networks \label{sec:simulations-CD}}

In this section, we use simulated data to evaluate the performance of our method on both predicting a response and recovering the community structure.   
We generate networks from the sGSBM  with $n=40$ nodes with $K=4$ equal size communities, resulting in  $p=780$ distinct edges.   Given a subject-specific expected connectivity matrix $R^{(m)}\in\Bbb{R}^{K\times K}$, each edge $(u,v)$ of the network $A^{(m)}$, with $u>v$, is generated independently from a Gaussian distribution, 
\begin{equation}
A^{(m)}_{uv}\sim  \mathcal{N}(R^{(m)}_{c_uc_v}, s^2),
\end{equation}
with $s=0.1$. 
The connectivity matrix $R^{(m)}$ for subject $i$ is set to 
\begin{equation}
R^{(m)} = \left(\begin{array}{cccc}
0.3 + tU_m & 0.3 & 0.1 & 0.1\\
0.3  & 0.3+ tU_m & 0.1 & 0.1\\
0.1 & 0.1 & 0.3 + tU_m & 0.3\\
0.1 & 0.1 & 0.3 & 0.3 + tU_m \\
\end{array}\right),
\end{equation}
where $U_m$ is a random variable uniformly distributed on $(-0.5, 0.5)$ and $t\in\Bbb{R}$ is  a parameter. When $t=0$, the networks have only two communities, and as $t$ increases, the four communities become more distinguishable. Given a subject's network $A^{(m)}$, the response $Y_m$ is generated from the linear model, 
\begin{equation}
Y_m = \langle A^{(m)}, B \rangle + \epsilon_m,
\end{equation}
where $\epsilon_m\sim \mathcal{N}(0,\sigma^2)$ are i.i.d.\ noise. The matrix of coefficients $B$ shares the community structure of the networks and is defined as 
\begin{equation*}
B_{uv} = \left\{\begin{array}{cl}
1 & \text{if }c_u=c_v,\\
0 & \text{otherwise}.
\end{array} \right.
\end{equation*}

We fit the model on a training sample of $N$ networks (to be specified) by solving the optimization problem with the least squares loss function.    Since $K$ is small compared to the sample size, we set the penalty tuning parameter $\lambda=0$;  enforcing a block structure on the coefficients already provides a substantial amount of regularization, and in particular allows us to handle the case $p > n$  without additional penalties.  Tuning this regularization is achieved by choosing $K$ via 5-fold cross-validation. We evaluate out-of-sample prediction performance by the relative mean squared error (MSE),   computed as $(N\widehat{\text{Var}}(Y))^{-1}\sum_{m=1}^N(Y_m-\hat Y_m)^2$.
The parameter $K$ is chosen by 5-fold cross-validation.   Since the estimated $K$ may be different from $K=4$ used to generate the data, we measure community detection performance by  co-clustering error, which calculates the proportion of pairs of nodes that are incorrectly assigned to the same community.   For any two membership matrices $Z$ and $\tilde{Z}$, the co-clustering error is defined as
\begin{equation*}
E(Z, \tilde{Z}) = \frac{1}{n^2} \sum_{i=1}^n\sum_{j=1}^n |(ZZ^T - \tilde{Z}\tilde{Z}^T)_{ij}|.
\end{equation*}

We initialize the solution with spectral clustering (Algorithm \ref{alg:SC-leastsquares}) and then run ADMM (Algorithm \ref{alg:ADMM-SCD}).   We compare to lasso and ridge regression \citep{friedman2009glmnet}, two generic regularized linear regression methods, as prediction benchmarks.    Since $p > n$, we do have to regularize lasso and ridge by tuning $\lambda$, whereas our method is only tuned by choosing $K$.   If anything, this will give an advantage to lasso and ridge, since further tuning $\lambda$ for our method could only improve its performance.  An oracle solution is also included as a gold standard,  obtained by solving \eqref{eq:blockConstrainedProblem-noZ} using the true communities.

Unsupervised community detection methods for multiple networks often start by combining by averaging the networks, but these methods generally are designed for networks with homogeneous structure, andwill fail in our setting since $\mathbb{E}[A^{(m)}]$ is a matrix with only two communities.  \cite{bhattacharyya2017spectral} proposed a community detection method for samples of networks that can work in our setting,  based on performing  spectral clustering on the sum of squared adjacency matrices $\tilde{A} = \sum_{m=1}^N(A^{(m)^2} - \text{diag}(A^{(m)^2}))$. We use this method as a benchmark for the community detection part of the problem.

The difficulty of the problem is controlled by multiple parameters, inlcuding the noise level $\sigma$, the strength of the community structure (controlled by $t$), and the sample size $N$.   In each experiment, we vary one of these parameters while keeping the other two constant.   The constant values are set to $\sigma=1$, $t=0.025$ and $N=150$. Each scenario is repeated 50 times.

\begin{figure}
	\includegraphics[width=\textwidth]{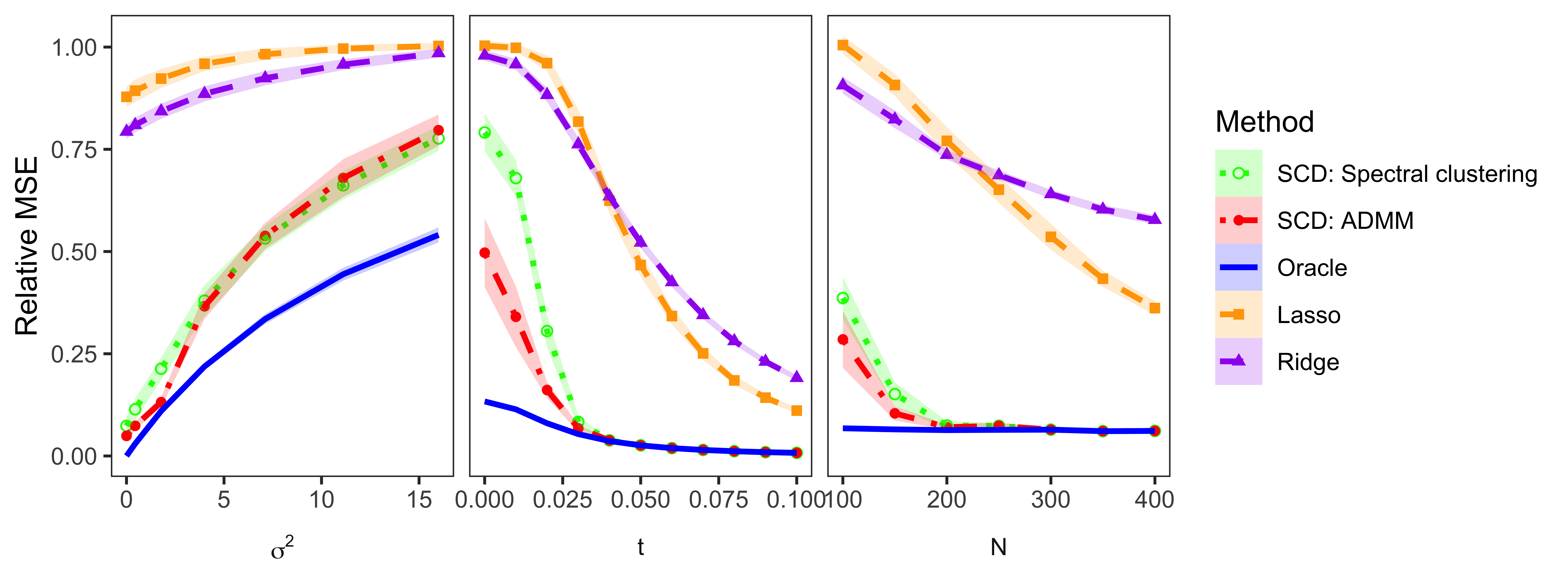}
	\includegraphics[width=\textwidth]{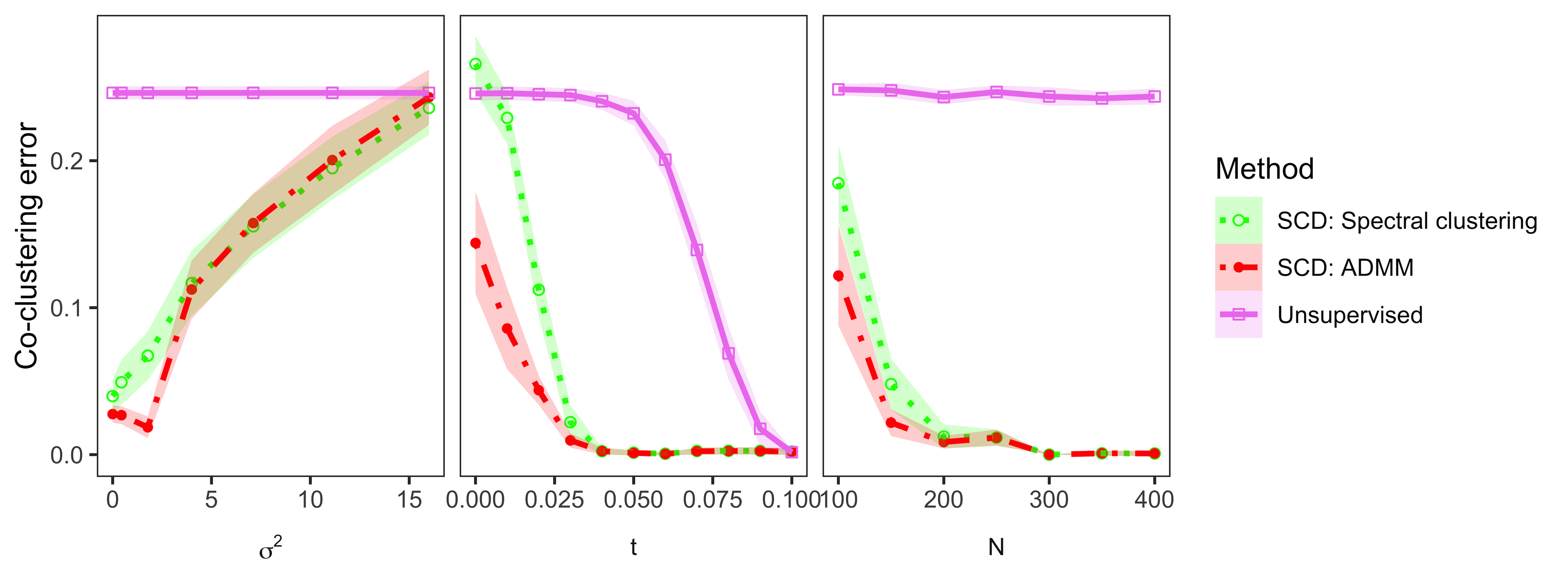}
	\caption{Results for prediction (top) and community detection (bottom). Each plot shows the corresponding error averaged over 50 replications as a function of one of the parameters $\sigma$, $t$, and $N$, keeping the other two fixed.}
	\label{fig:simulations-SCD}
\end{figure}

Figure \ref{fig:simulations-SCD} shows the average performance of the methods as a function of model parameters.  In general, all methods perform better when the noise level is lower,  the community structure  is stronger, and the sample size larger, as one would expect.   Our method outperforms both lasso and ridge, since it takes advantage of the underlying structure.   
ADMM usually improves on the initial value provided by SC, and when the signal is strong enough, both methods recover the community structure correctly, performing as well as the oracle.   Unsupervised community detection does not use the response values, and hence its performance does not depend on $\sigma$.    Community detection becomes easier as $t$ increases, but the unsupervised method requires a much larger $t$ than our supervised community detection algorithm to achieve the same performance.   The sample size has virtually no effect on unsupervised community detection, but supervised detection improves very quickly as the sample size grows.

%% file: data.tex
\section{Supervised community detection in fMRI brain networks \label{sec:data-CD}}

Here we apply the proposed method to classification of brain networks from healthy and schizophrenic subjects from the COBRE dataset, described below.  Understanding the connection between schizophrenia and brain connectivity is somewhat important from the clinical perspective, since schizophrenia is a famously heterogenous condition with many subtypes \citep{Arnedo2015},
and especially important for the scientific goal of understanding which regions of the brain are implicated in schizophrenia, potentially an important step towards developing new treatments. 

The COBRE dataset \citep{Aine2017,wood2014harnessing} includes  resting-state fMRI data for 54 schizophrenic patients and 70 healthy control subjects.  The data underwent  preprocessing and registration steps   (see \cite{arroyo2016graphclass} for details)    to generate a weighted correlation network for each person in the study. The Power parcellation \citep{power2011functional} was employed to define the regions of interest (ROIs), resulting in a total of $n=263$ nodes in the brain (see left panel of Figure \ref{fig:powervsnew}).

We train our method using logistic regression loss and the lasso penalty, needed since the sample size $N=124$ only allows us to fit up to $K=15$ communities without regularization. As is commonly done with logistic regression with a large number of predictors, we also add a small ridge penalty for numerical stability. The objective function  is given by 
\begin{equation*}
\ell(B) + \Omega_\lambda(B) = \frac{1}{m}\sum_{i=1}^m\log\left(1+\exp(-Y_i(\langle A^{(i)},B\rangle+b))\right)+\frac{\gamma}{2}\|B\|_F^2 \ +  \ \lambda\|B\|_1.
\end{equation*}
The parameter $\lambda$ in the lasso penalty controls sparsity of the solution and is  selected by cross-validation. The value of $\gamma$ is fixed at $10^{-5}$. The details on the optimization solver for this problem are described in the Appendix~\ref{sec:elastic-net}. An implementation of the code can be found at \url{https://github.com/jesusdaniel/glmblock/}.

Analyzing the brain connectome at the level of brain systems is usually preferrable for interpretation purposes, since nodes or edges are too granular for a meaningful interpretation. In the Power parcellation, the 263 nodes are partitioned into 14 groups which are related to known functional systems in the brain \citep{power2011functional}. The nodes labeled according to these communities are shown in Figure \ref{fig:powercomms}. We use these pre-determined communities as a baseline by solving the constrained problem \eqref{eq:blockConstrainedProblem-noZ}, and compare the performance of our supervised community detection method using the same number of $K=14$ communities.

Figure \ref{fig:powervsnew} shows the 263 nodes of the data colored according to the parcellation by \cite{power2011functional} (left) and the communities we found with our method with $K=14$ (right). In both cases, we chose $\lambda$ by cross-validation, and evaluated prediction accuracy using 10-fold nested cross-validation. The accuracy for each method is reported in Figure \ref{fig:powervsnew}. Our supervised method uses the same number of parameters  but has a noticeably higher accuracy.   For ease of interpretation, Figure \ref{fig:newcomms} shows a separate plot for each community.    These node clusters are mostly well concentrated in space, suggesting they may correspond to meaningful structural or functional regions.  In Figure \ref{fig:sankey}, a Sankey diagram shows a comparison of the new and the old community assignments. Many of the Power communities are partitioned into smaller communities with supervised community detection. Figure \ref{fig:Bnew} shows the fitted coefficients ordered according to the supervised communities. Note that the communities F, H and I are mostly composed by nodes from the default mode network (community 5 in Power parcellation), which has been previously linked to schizophrenia \citep{broyd2009default,Whitfield-Gabrieli2009}
but now we have found cells with both positive and negative coefficients within those communities.

\begin{figure}
	\centering
	\begin{minipage}{0.48\textwidth}
		\centering
		\includegraphics[width=\textwidth]{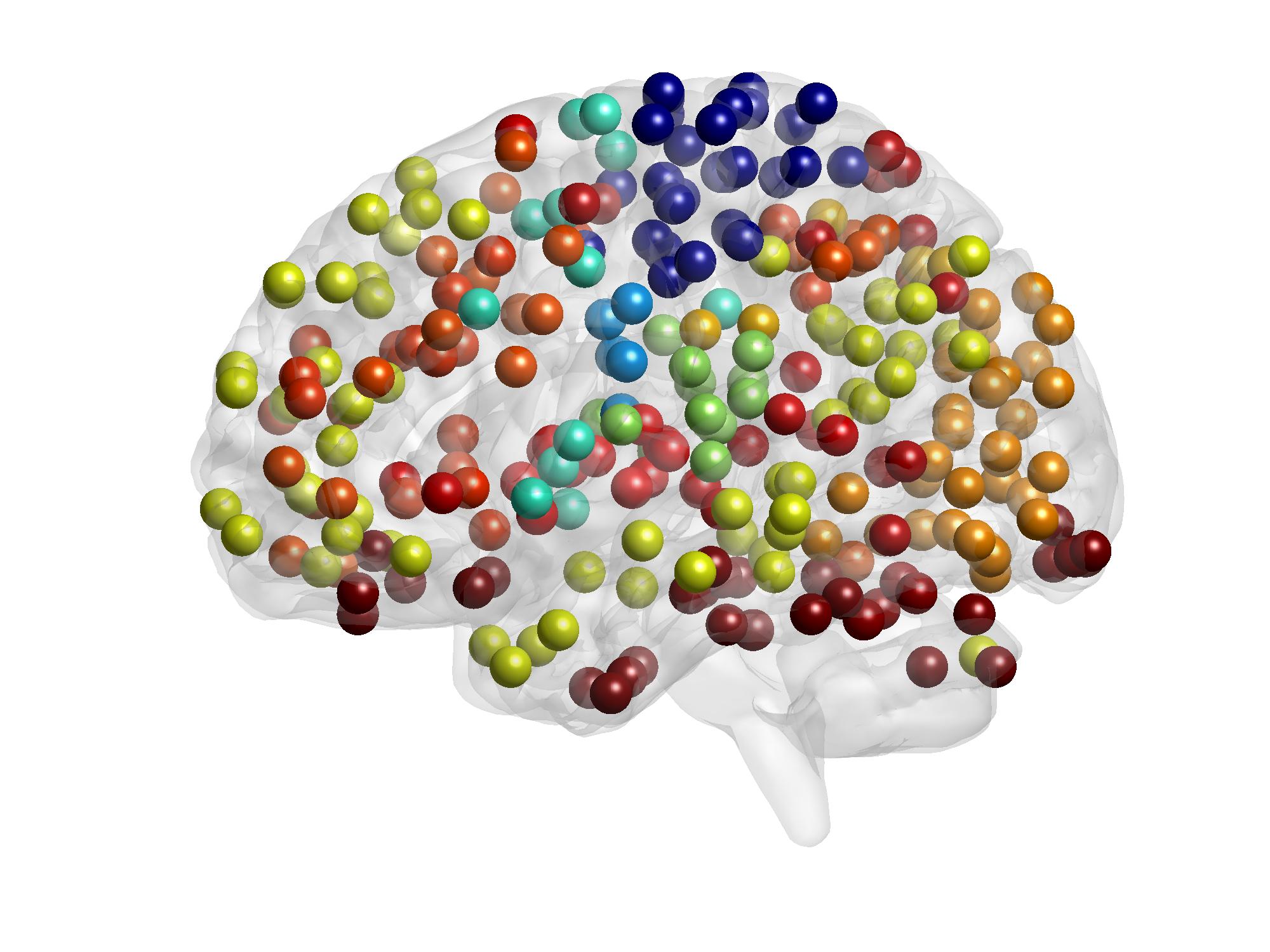}
		Power brain systems\\
		CV accuracy: 62\%
	\end{minipage}
	\begin{minipage}{0.48\textwidth}
		\centering
		\includegraphics[width=\textwidth]{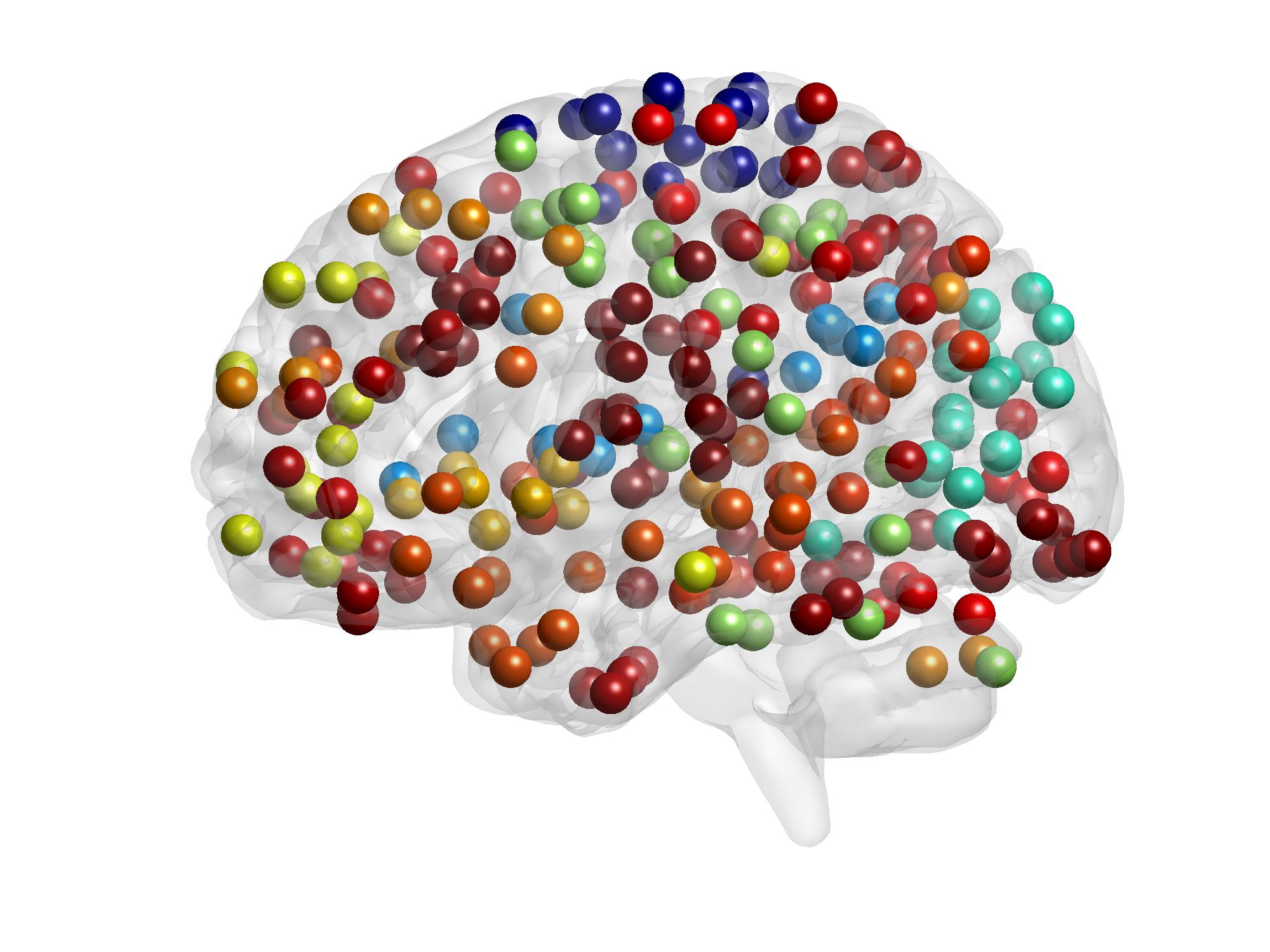}
		Supervised community detection\\
		CV accuracy: 79\%
	\end{minipage}
	\caption[Baseline communities (Power et al., 2011) and communities found by our supervised community detection method.]{Baseline communities \citep{power2011functional} and communities found by our supervised community detection method.}
	\label{fig:powervsnew}
\end{figure}

\begin{figure}
	\begin{tabular}{ccccc}
		\includegraphics[width=0.17\textwidth]{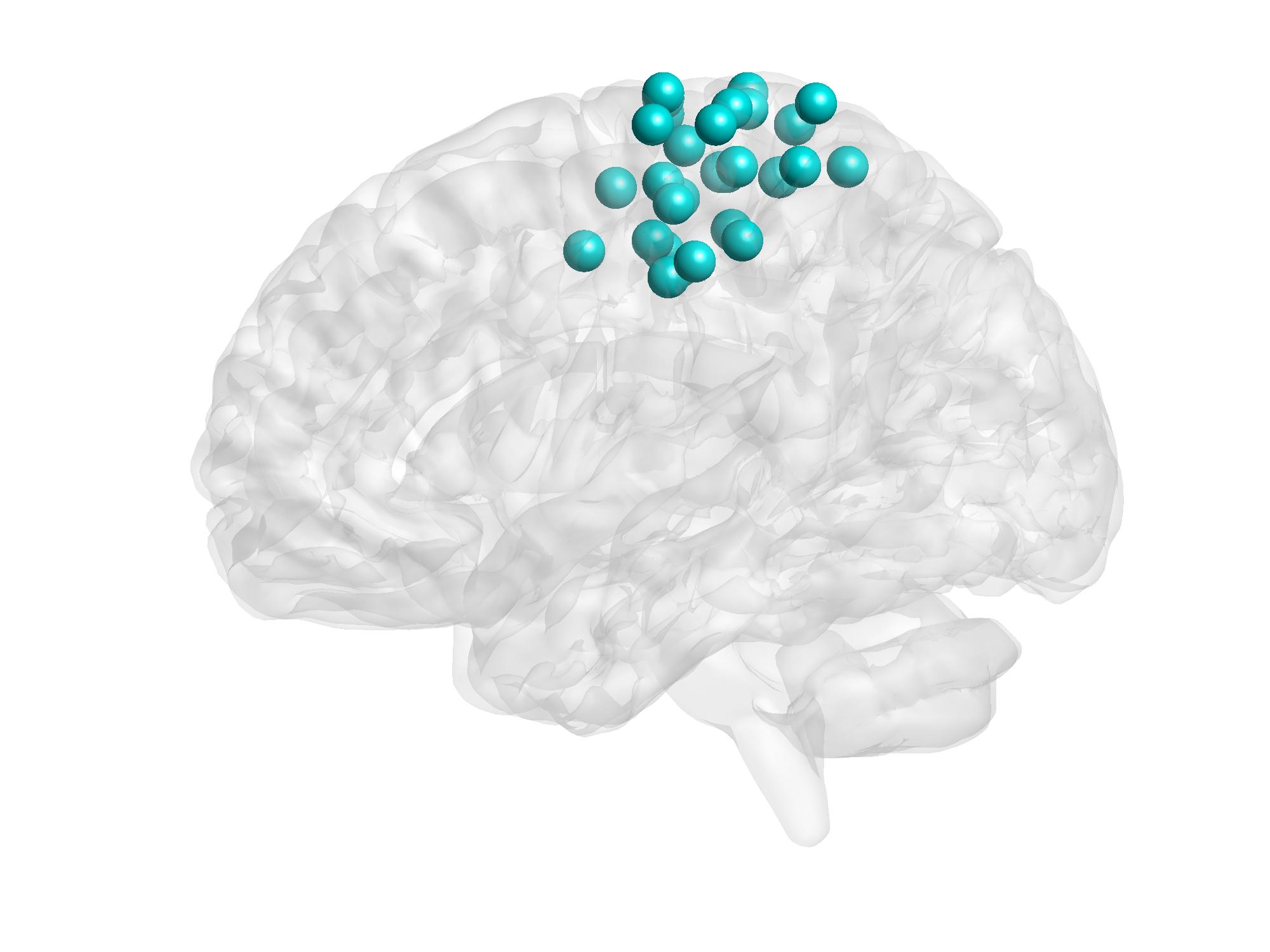} & \includegraphics[width=0.17\textwidth]{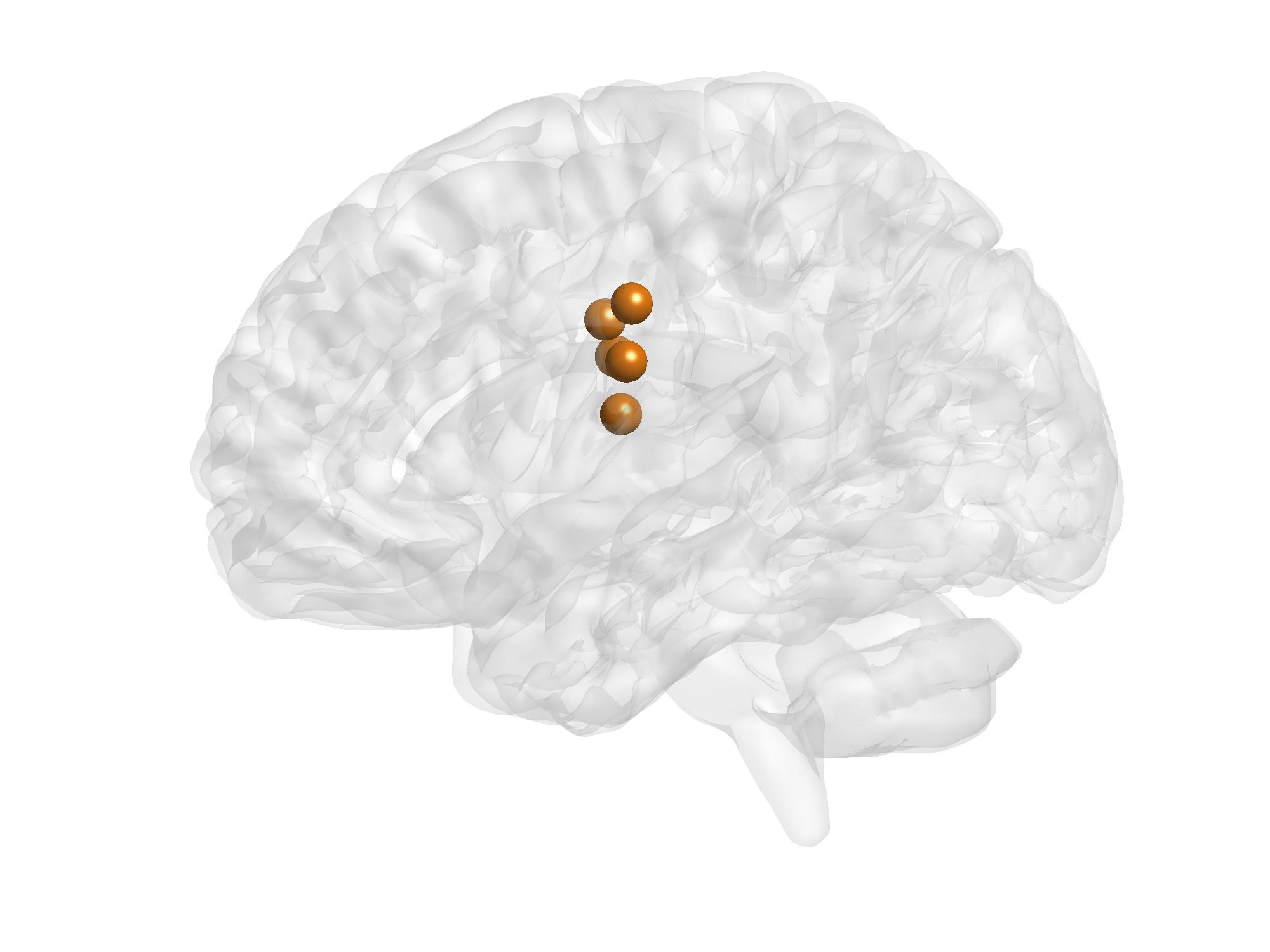} & \includegraphics[width=0.17\textwidth]{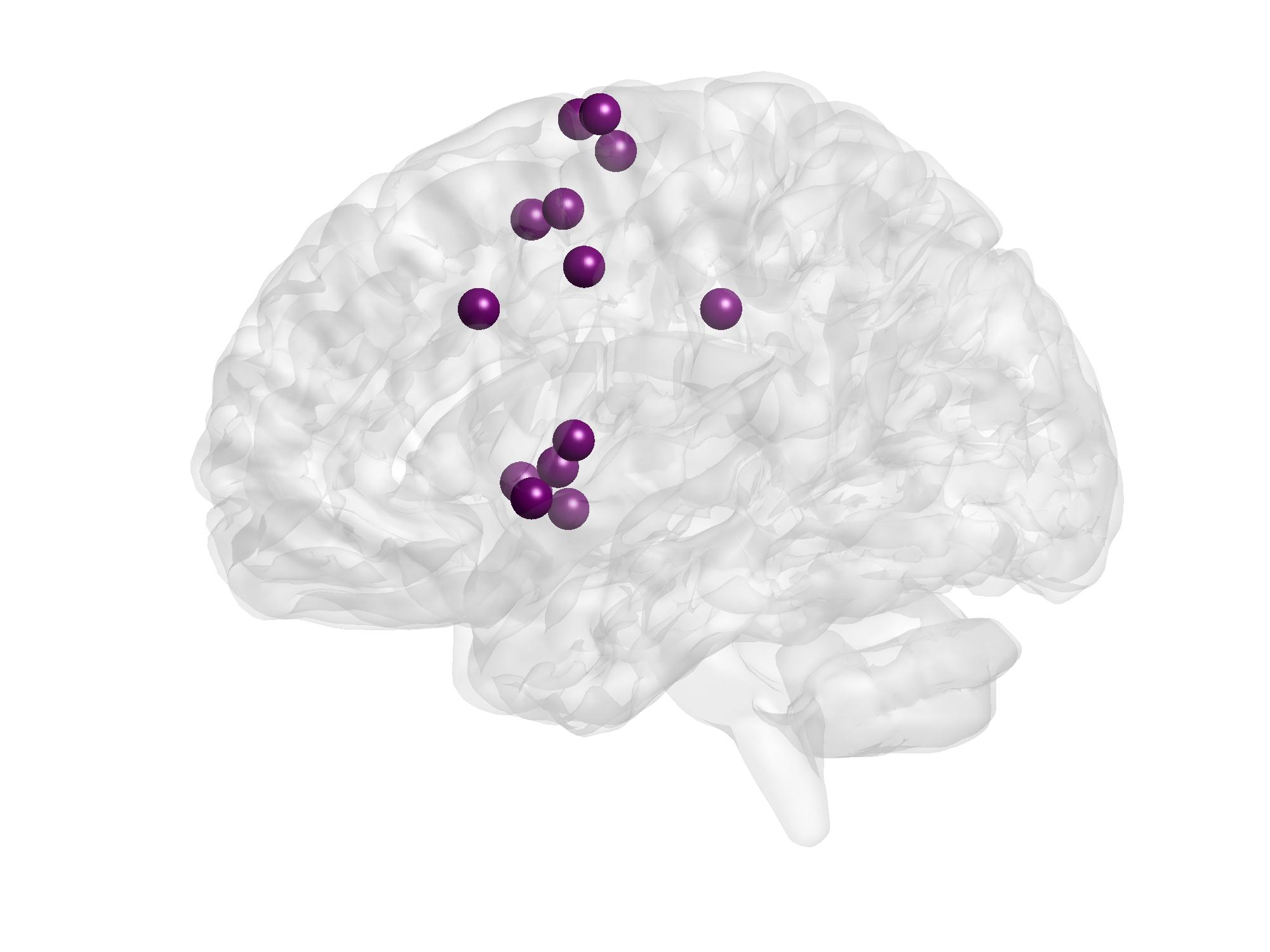} & \includegraphics[width=0.17\textwidth]{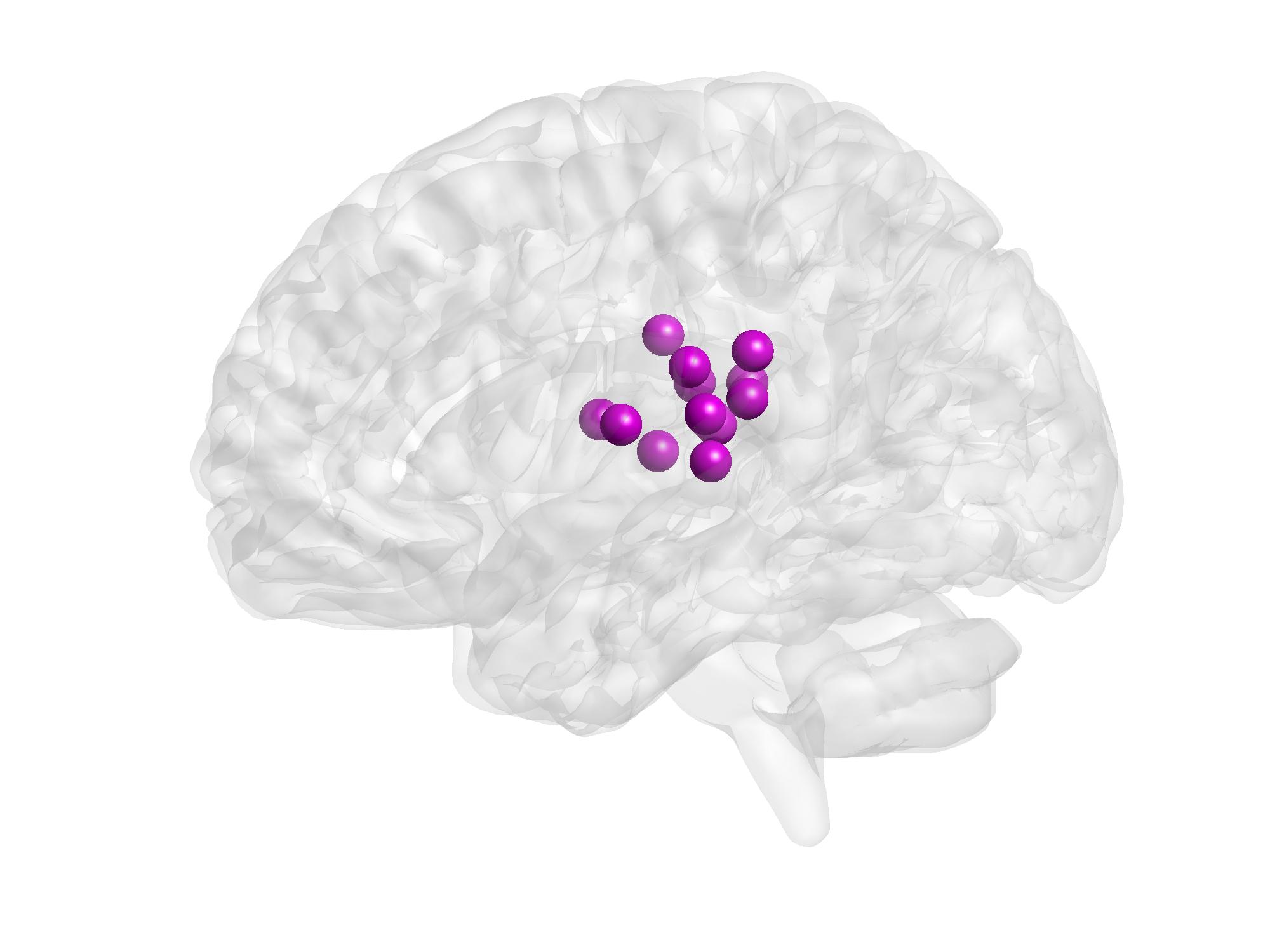} & \includegraphics[width=0.17\textwidth]{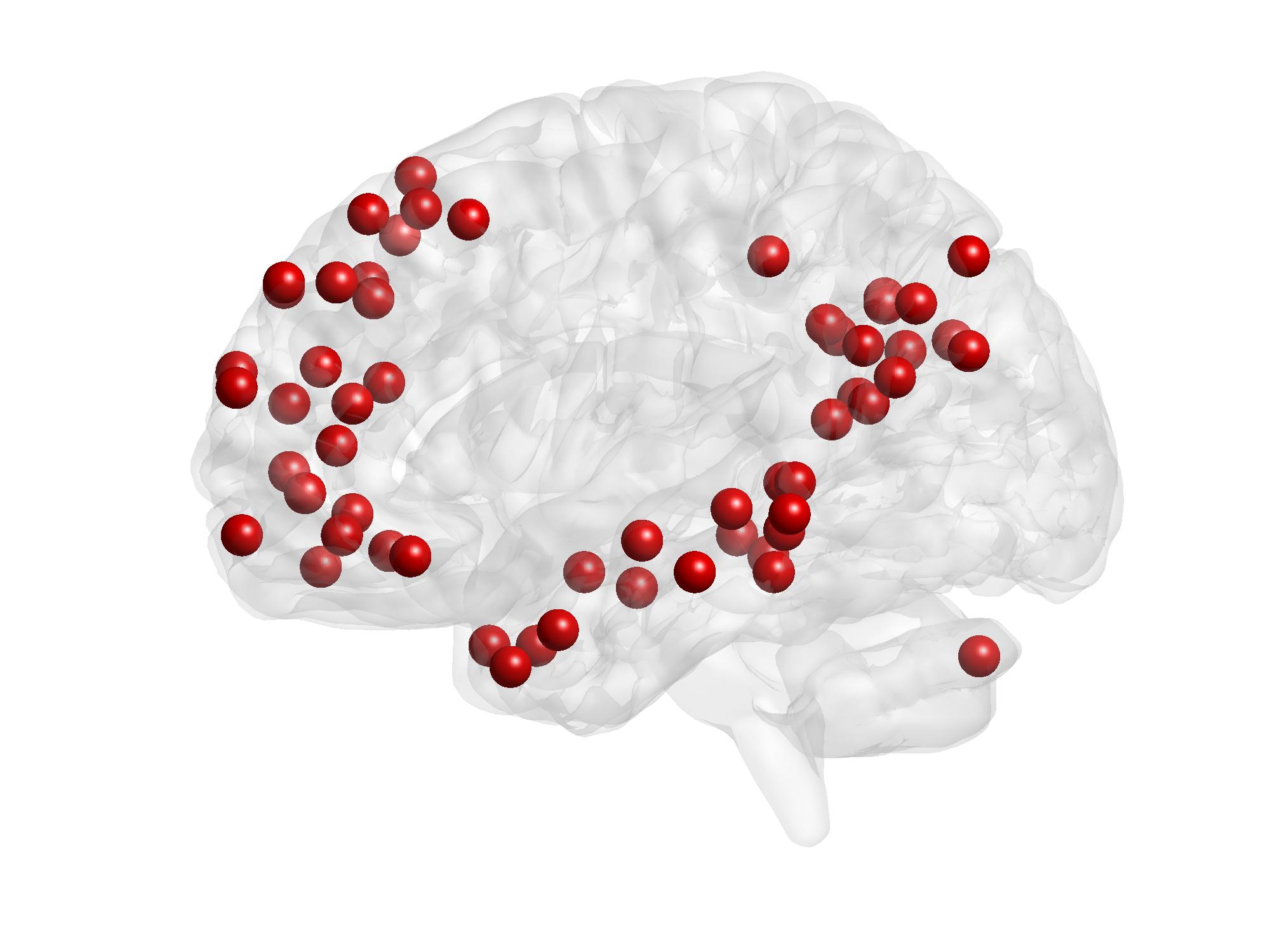} \\ 
		1 & 2 & 3 & 4 & 5\\ 
		\includegraphics[width=0.17\textwidth]{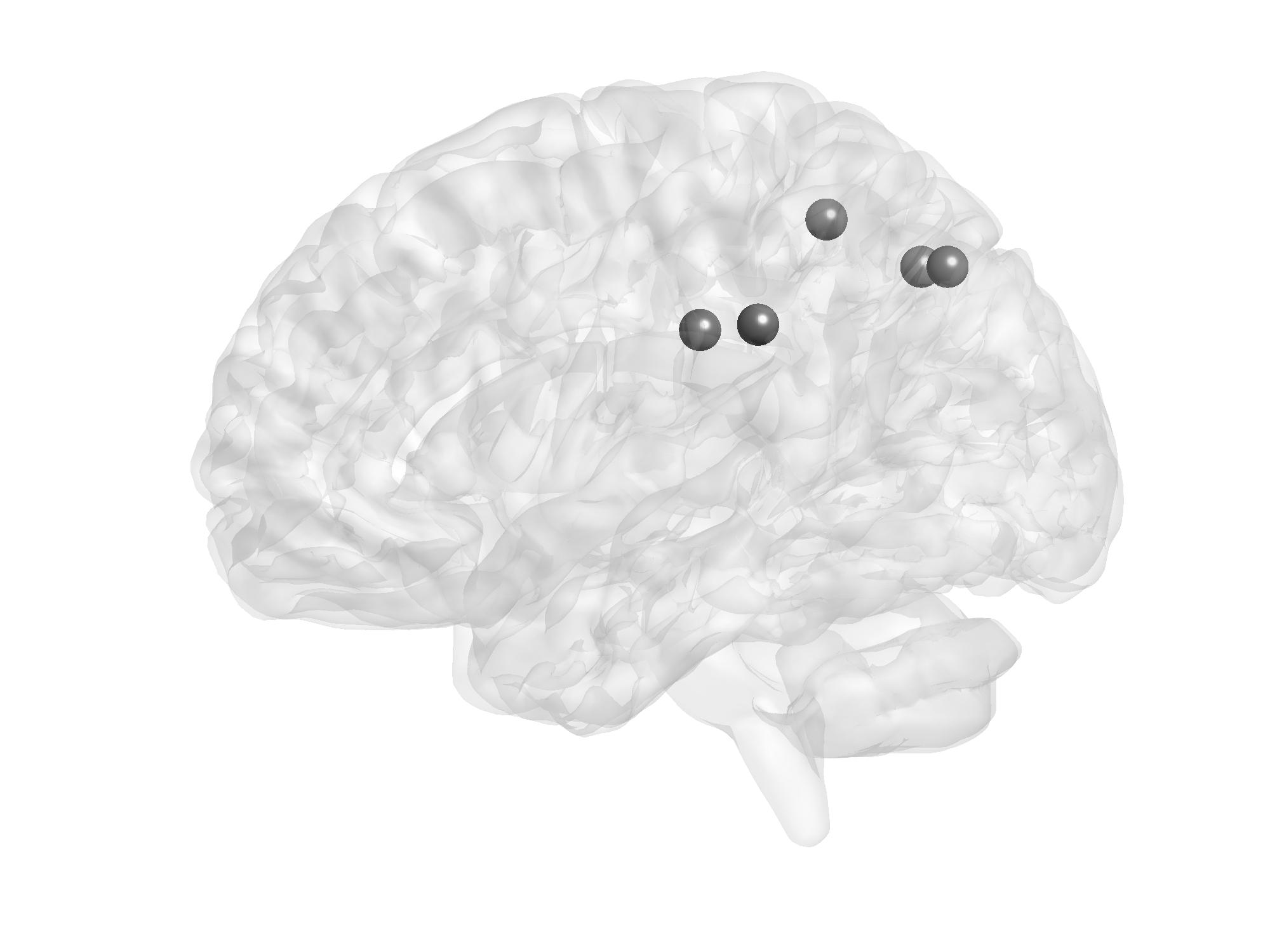}& \includegraphics[width=0.17\textwidth]{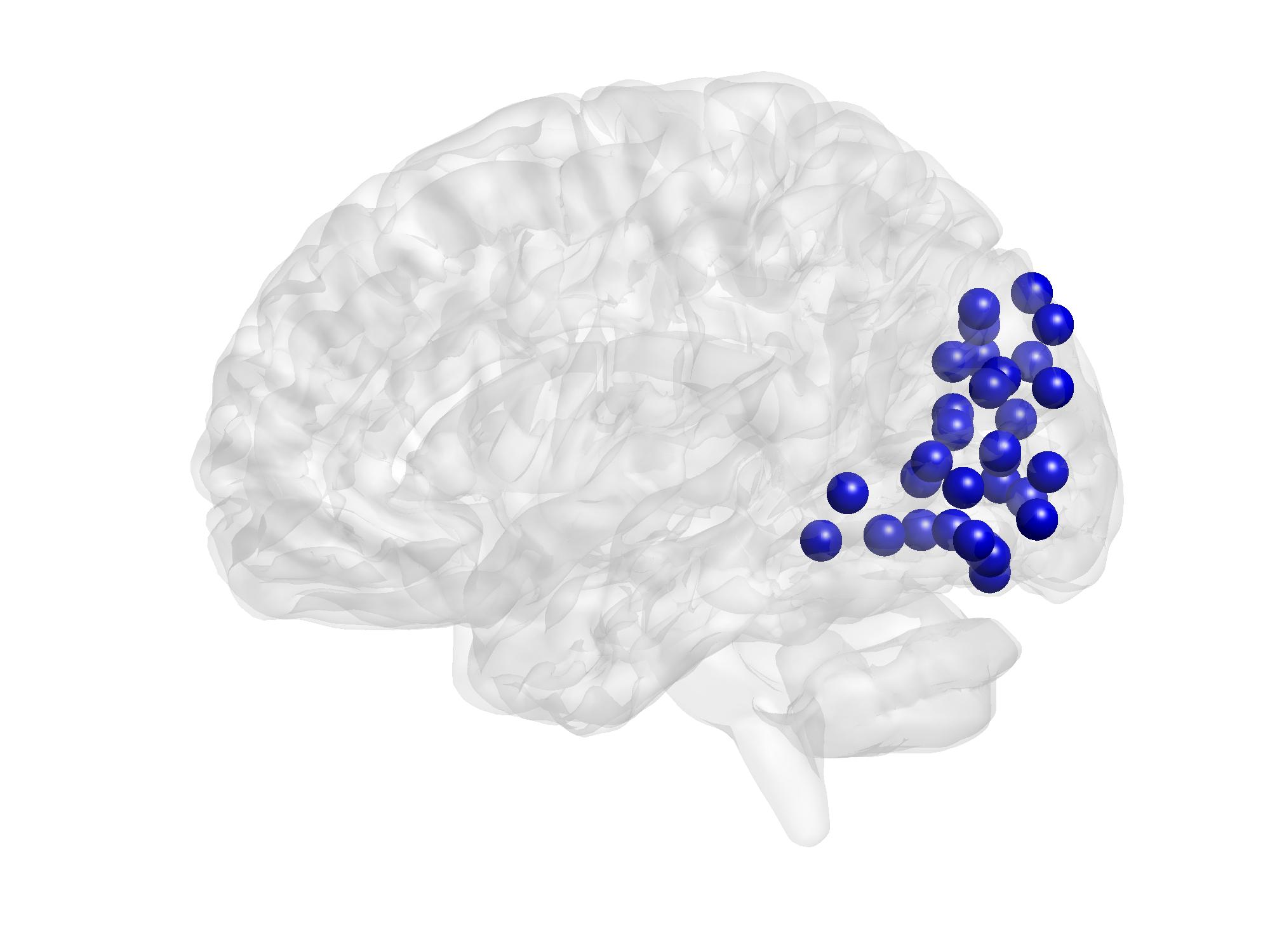} & \includegraphics[width=0.17\textwidth]{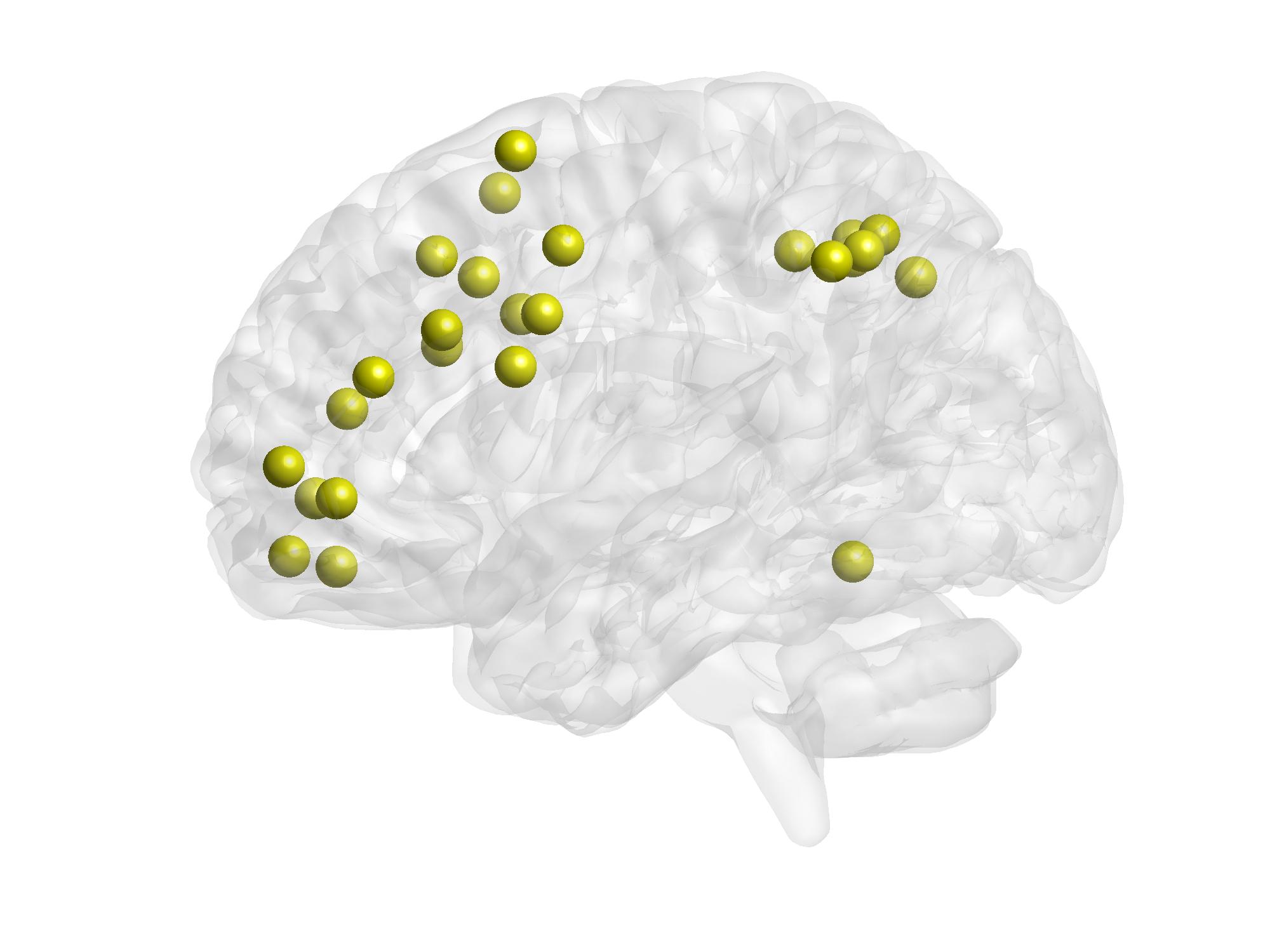} & \includegraphics[width=0.17\textwidth]{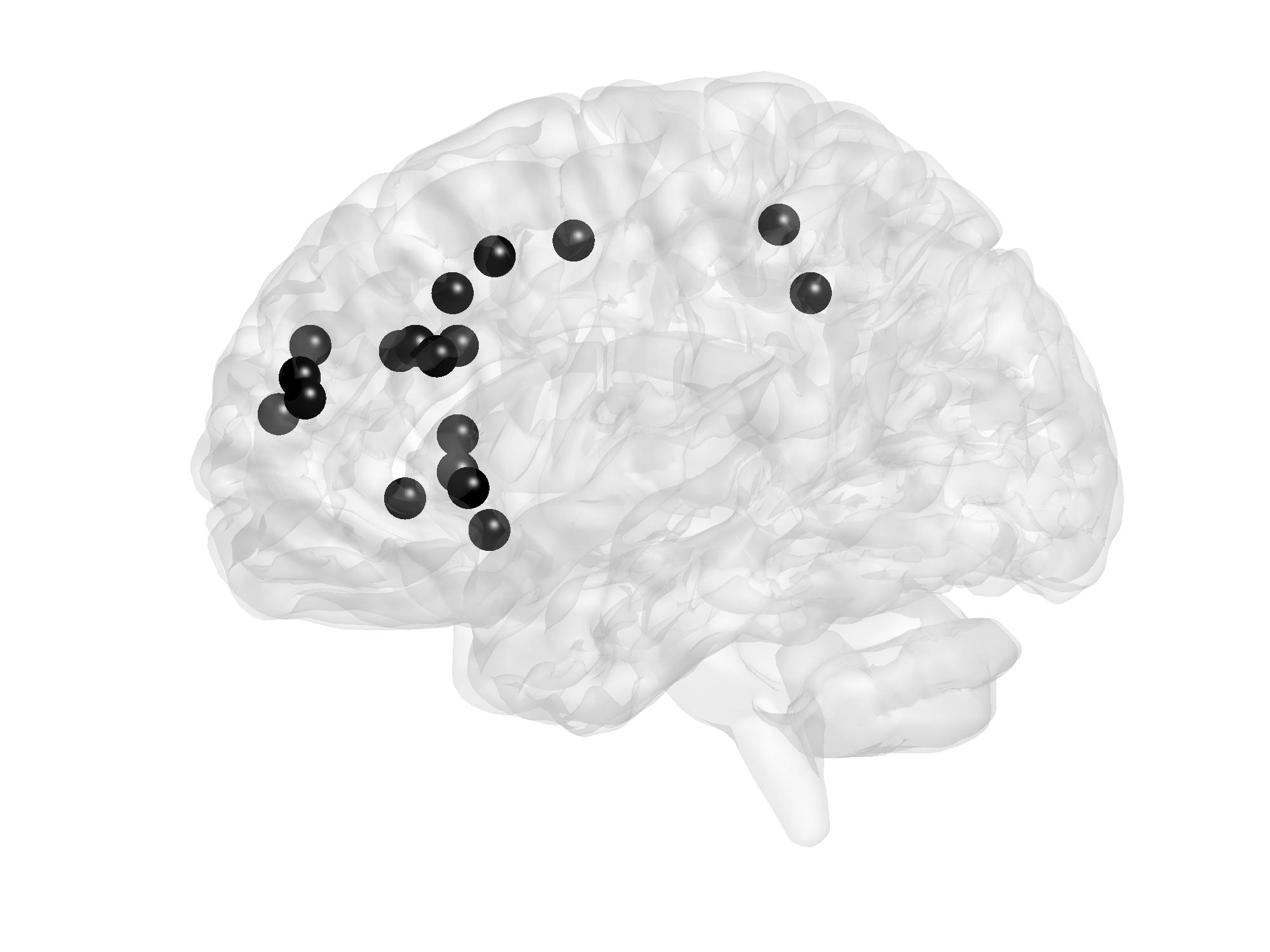} & \includegraphics[width=0.17\textwidth]{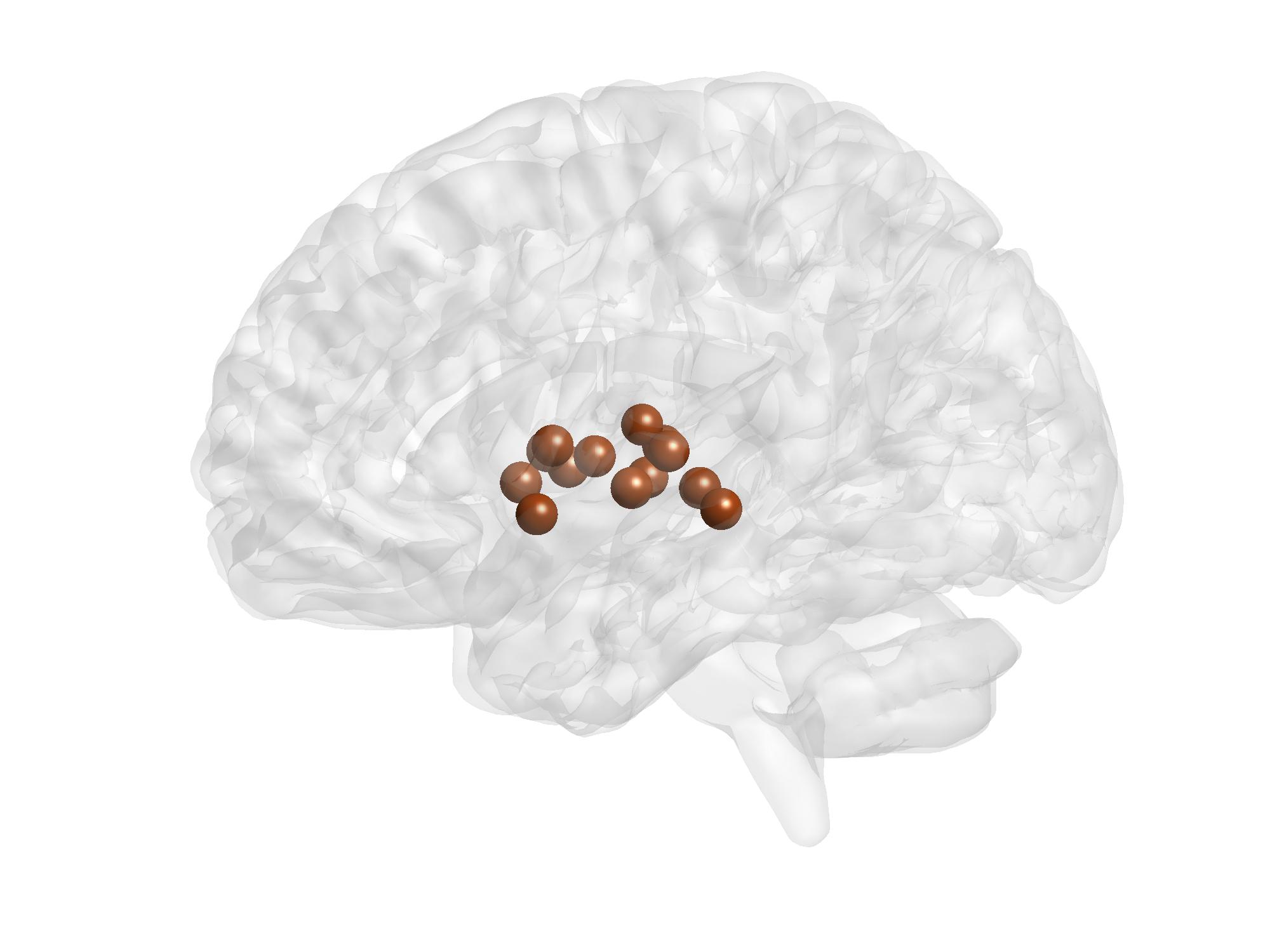} \\
		6 & 7 & 8 & 9 & 10\\ 
		\includegraphics[width=0.17\textwidth]{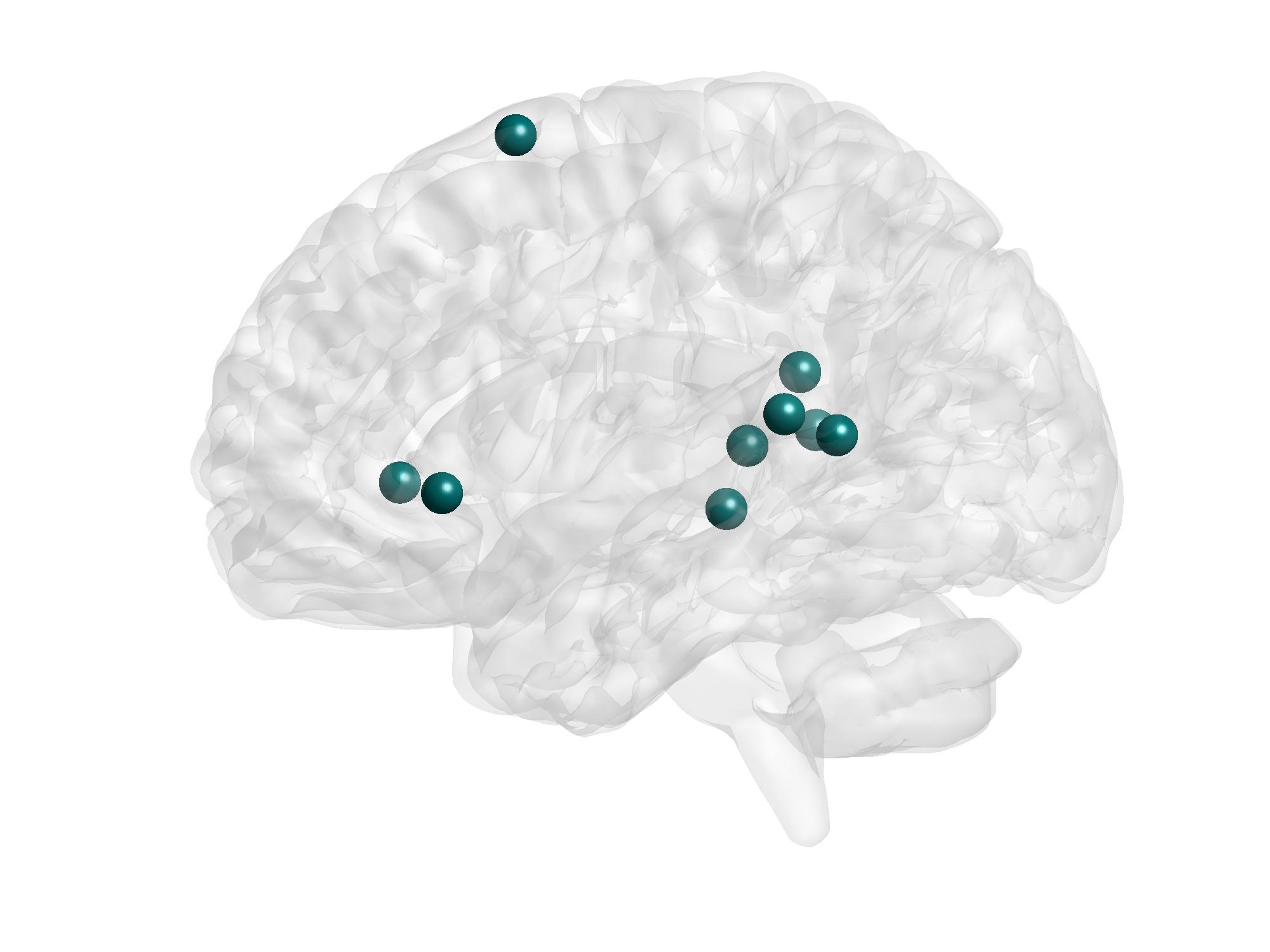}& \includegraphics[width=0.17\textwidth]{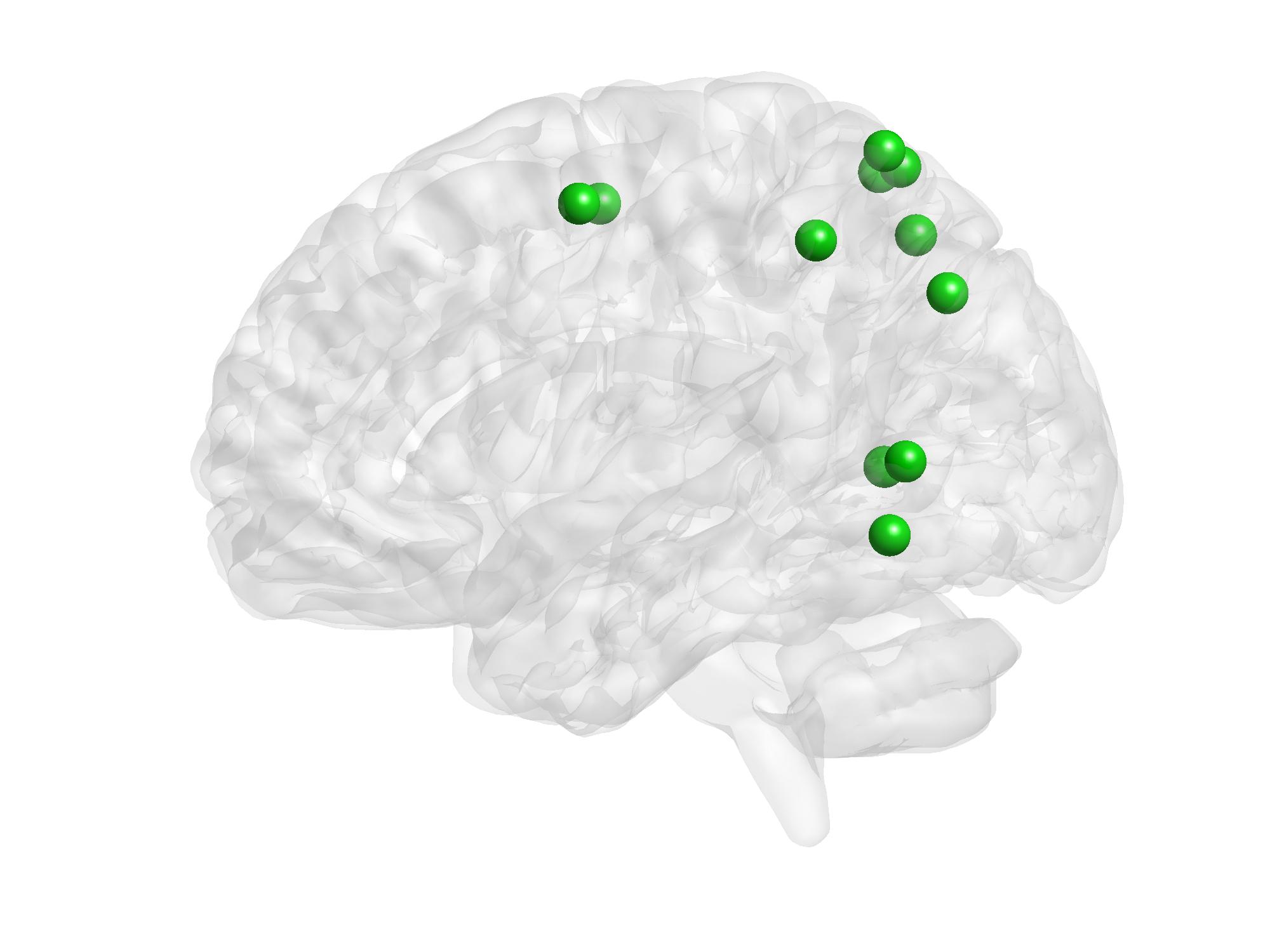} & \includegraphics[width=0.17\textwidth]{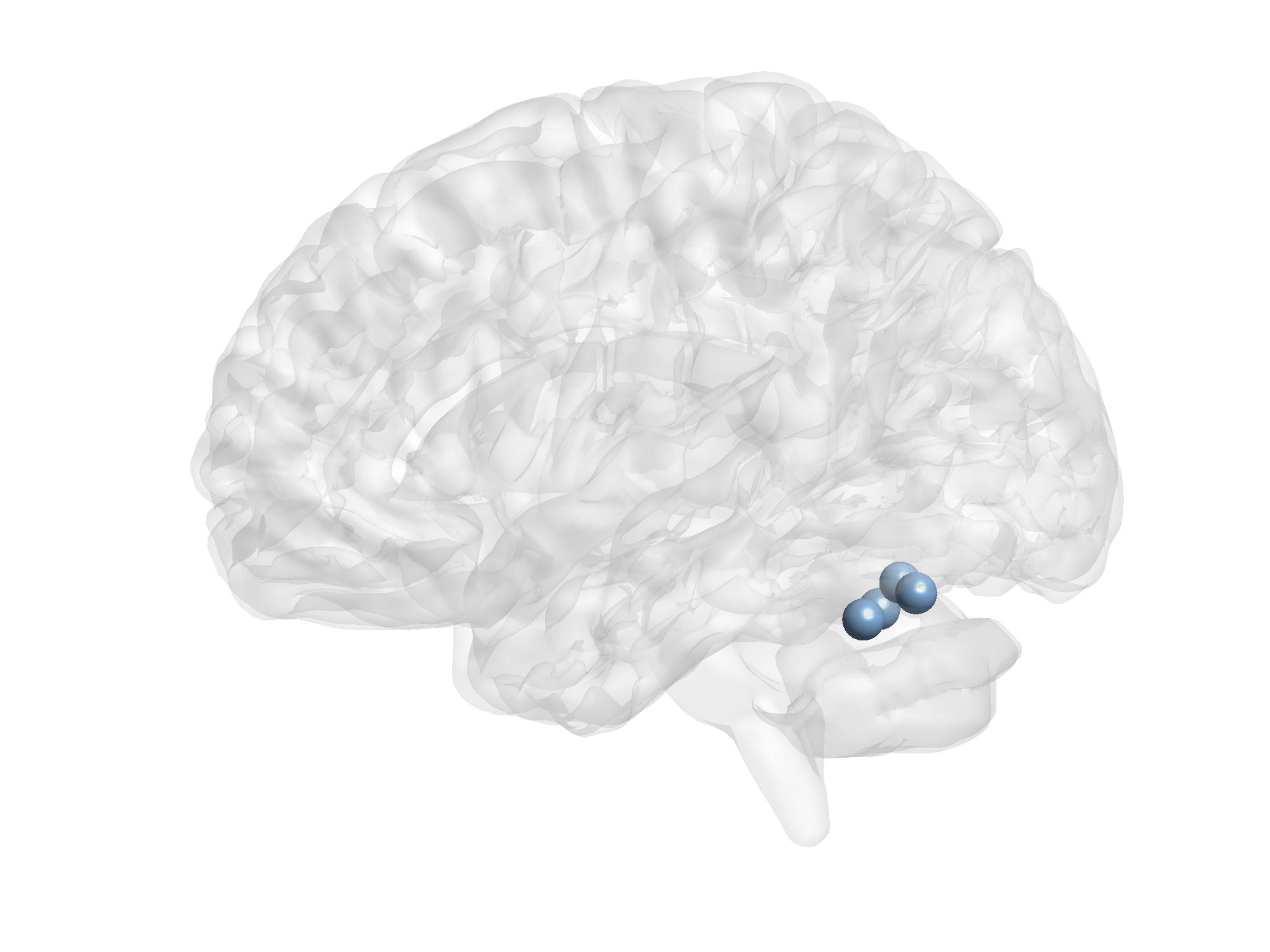} & \includegraphics[width=0.17\textwidth]{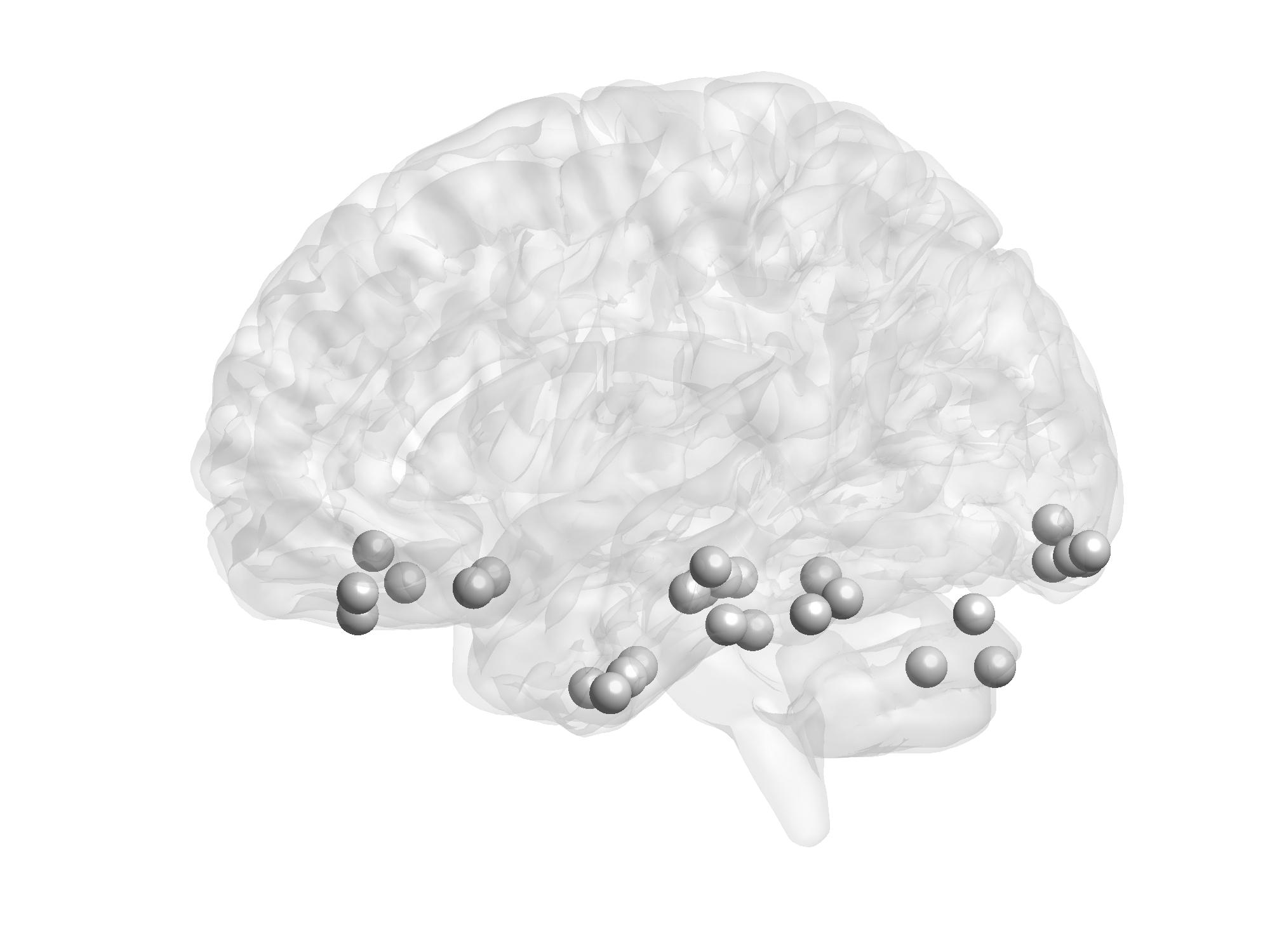} & \\
		11 & 12 & 13 & -1 & \\ 
	\end{tabular} 
	\caption[Individual communities proposed by Power et al. (2011).]{Individual communities proposed by \cite{power2011functional}.}
	\label{fig:powercomms}
\end{figure}

\begin{figure}
	\begin{tabular}{ccccc}
		\includegraphics[width=0.17\textwidth]{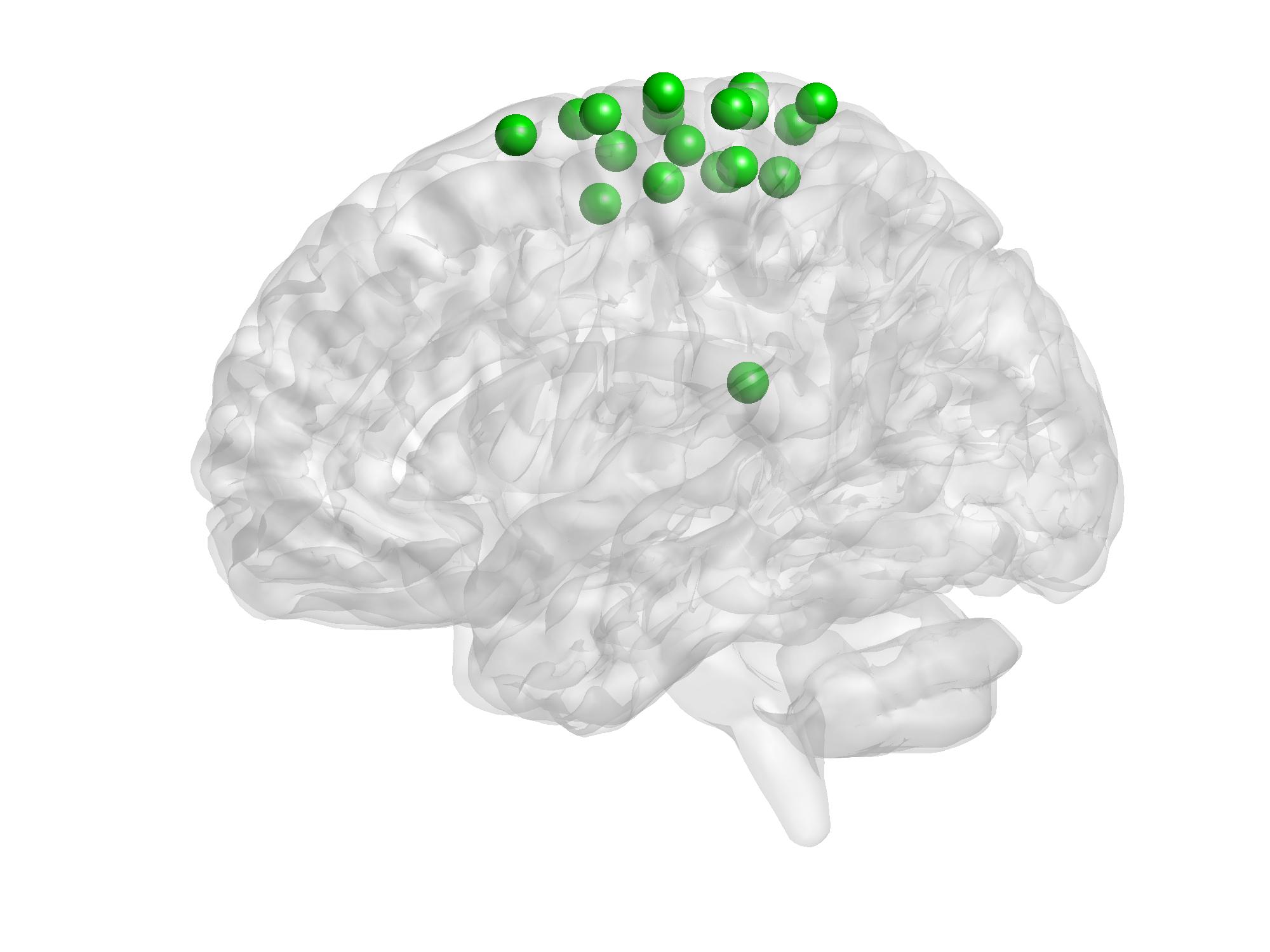} & \includegraphics[width=0.17\textwidth]{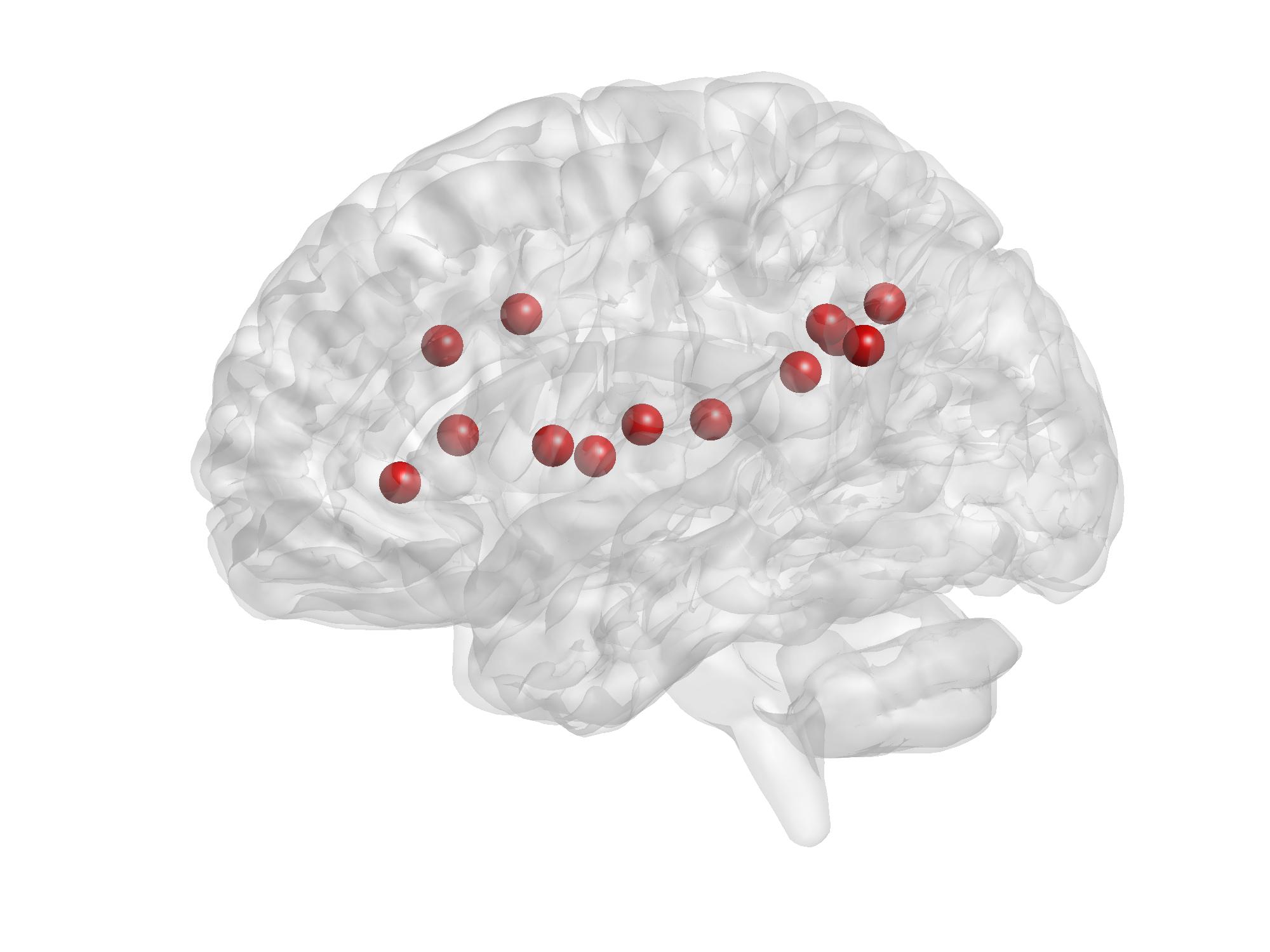} & \includegraphics[width=0.17\textwidth]{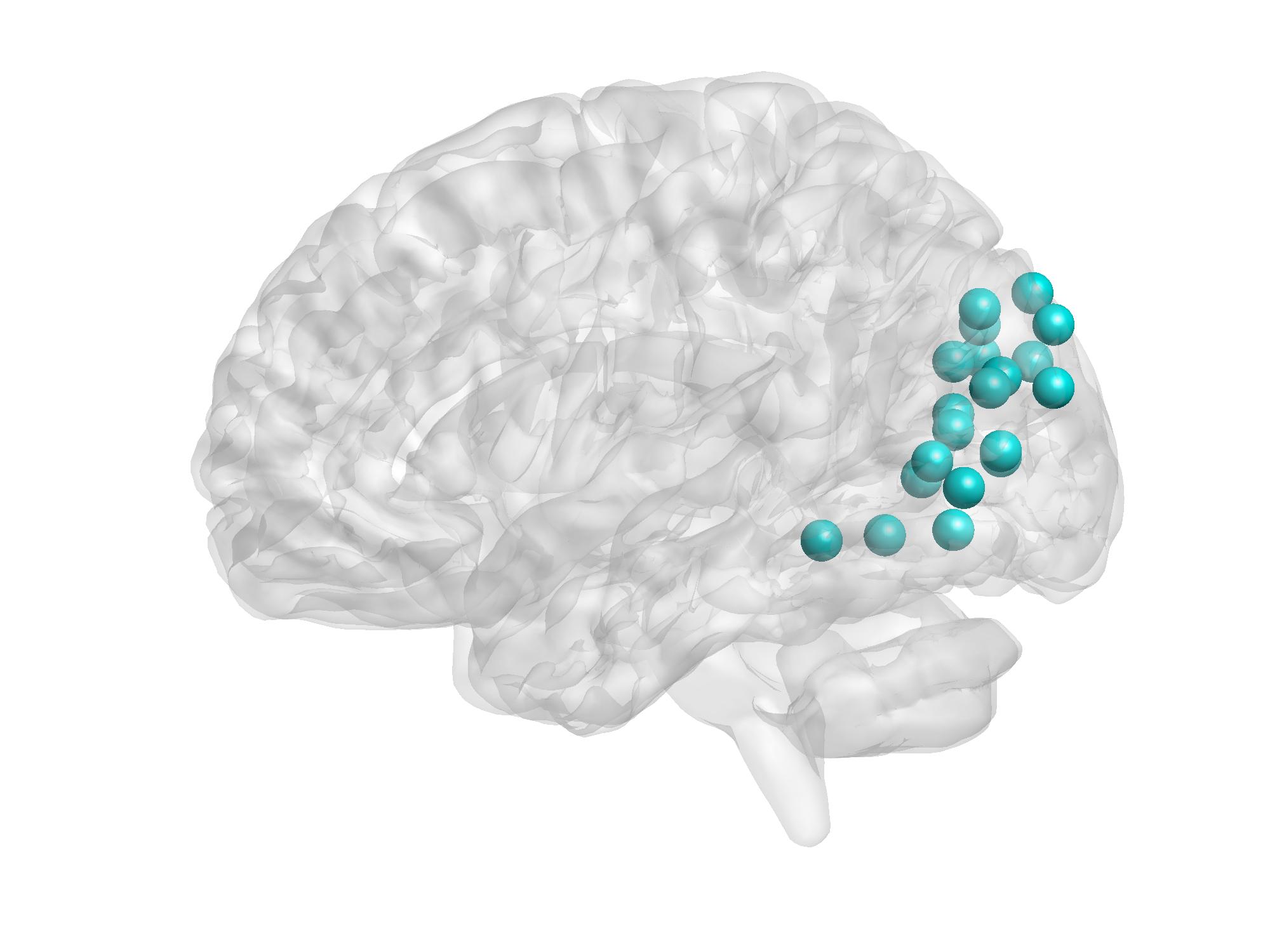} & \includegraphics[width=0.17\textwidth]{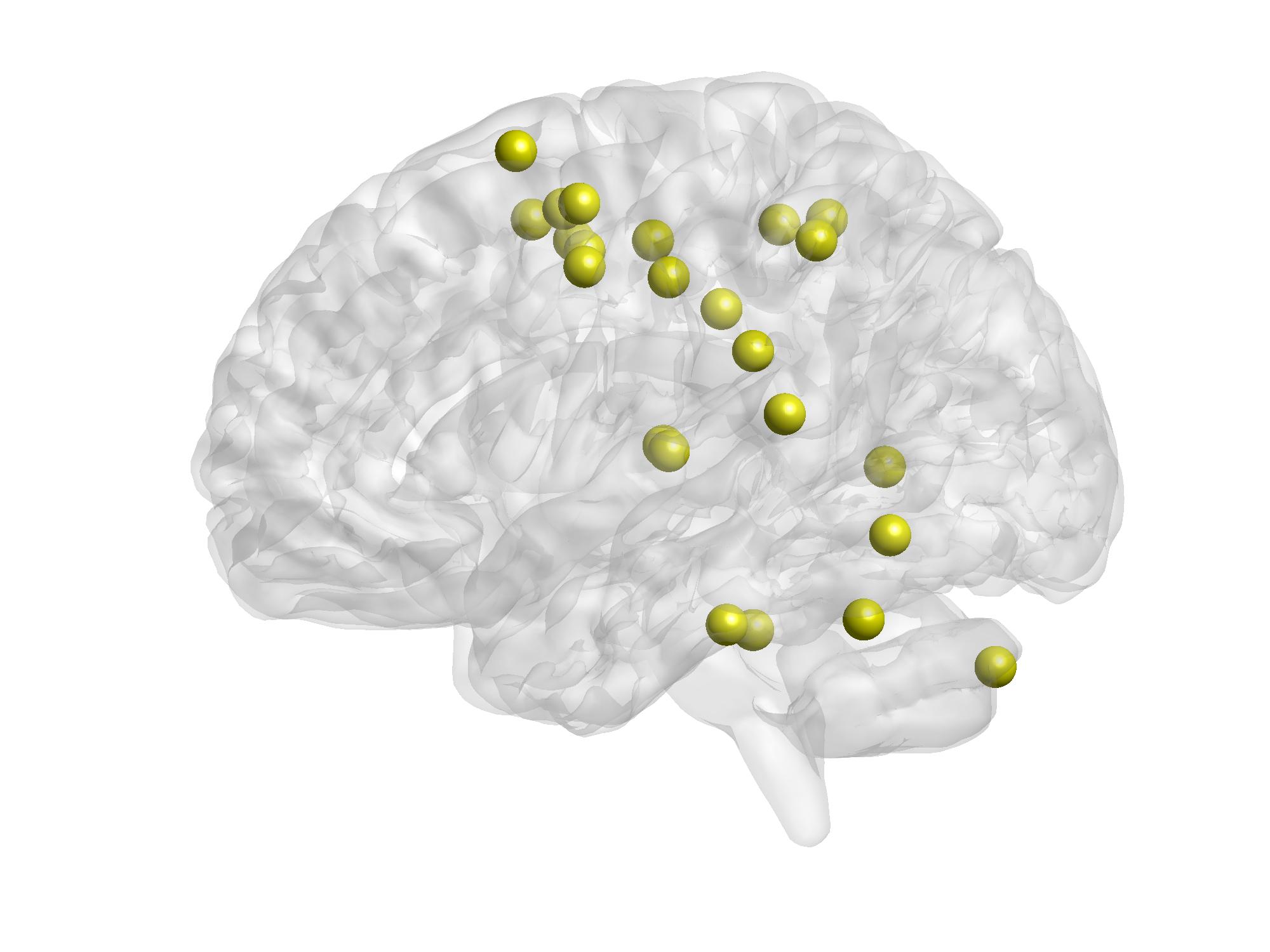} & \includegraphics[width=0.17\textwidth]{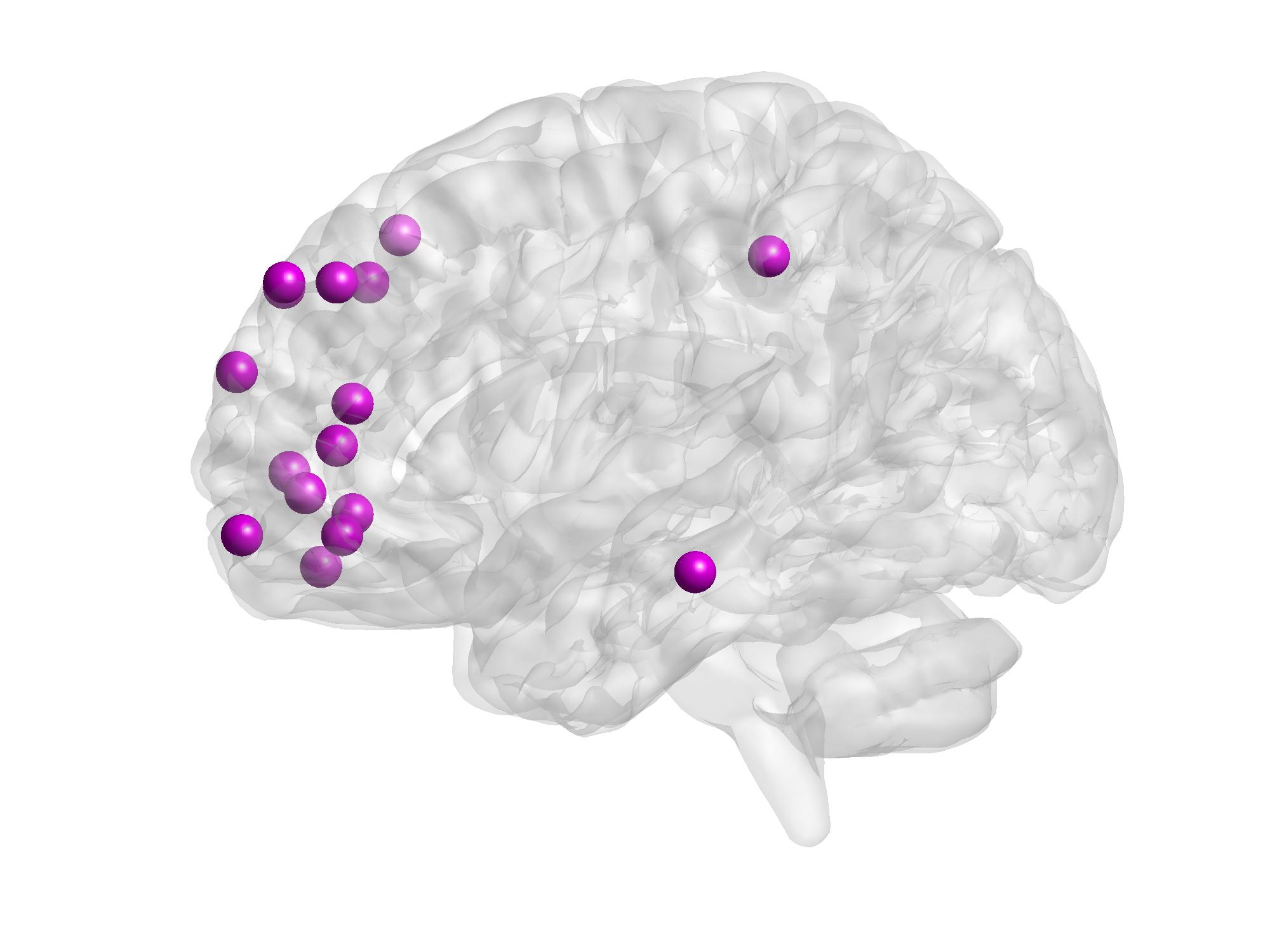} \\ 
		A & B & C & D & E\\ 
		\includegraphics[width=0.17\textwidth]{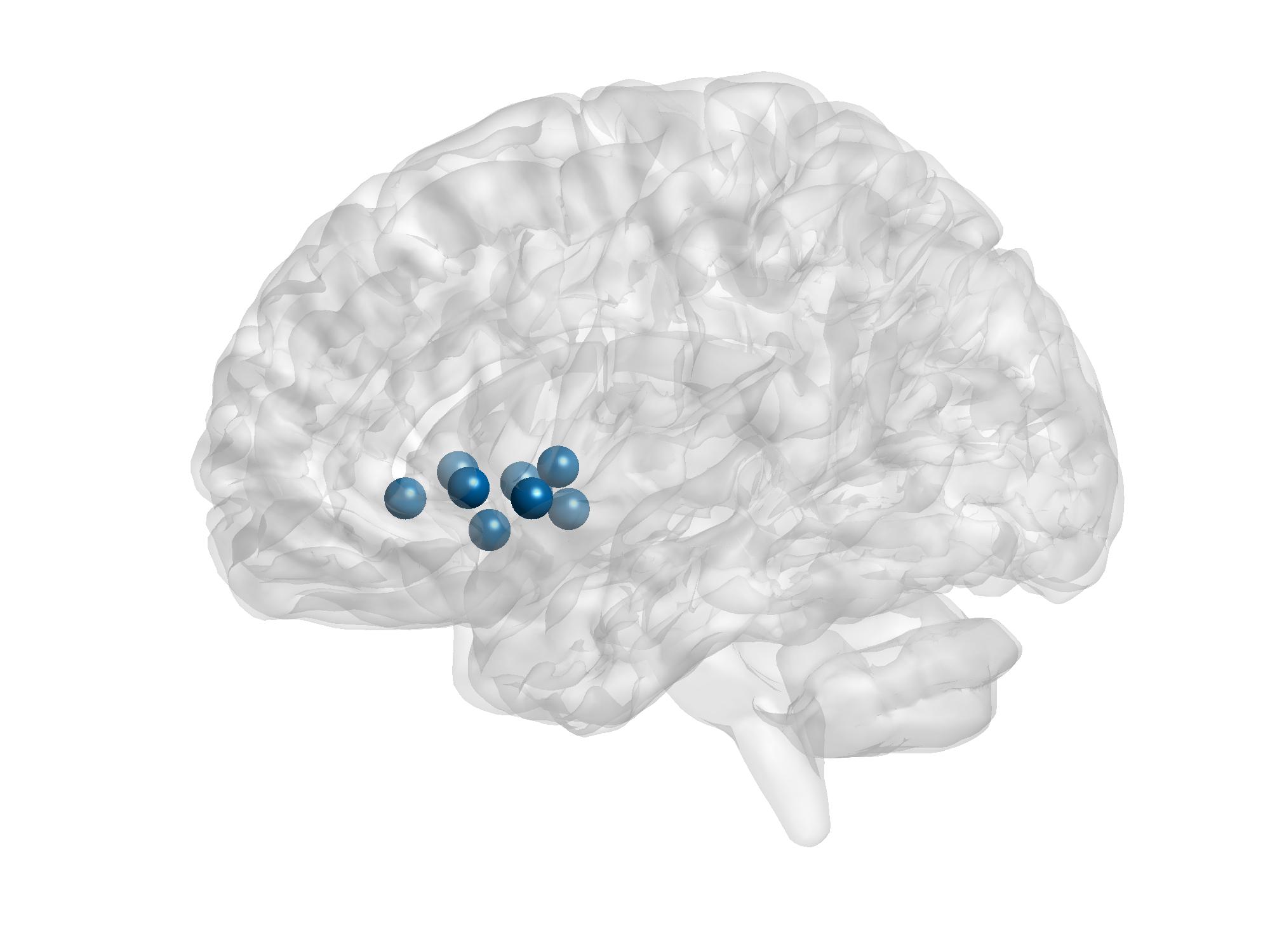}& \includegraphics[width=0.17\textwidth]{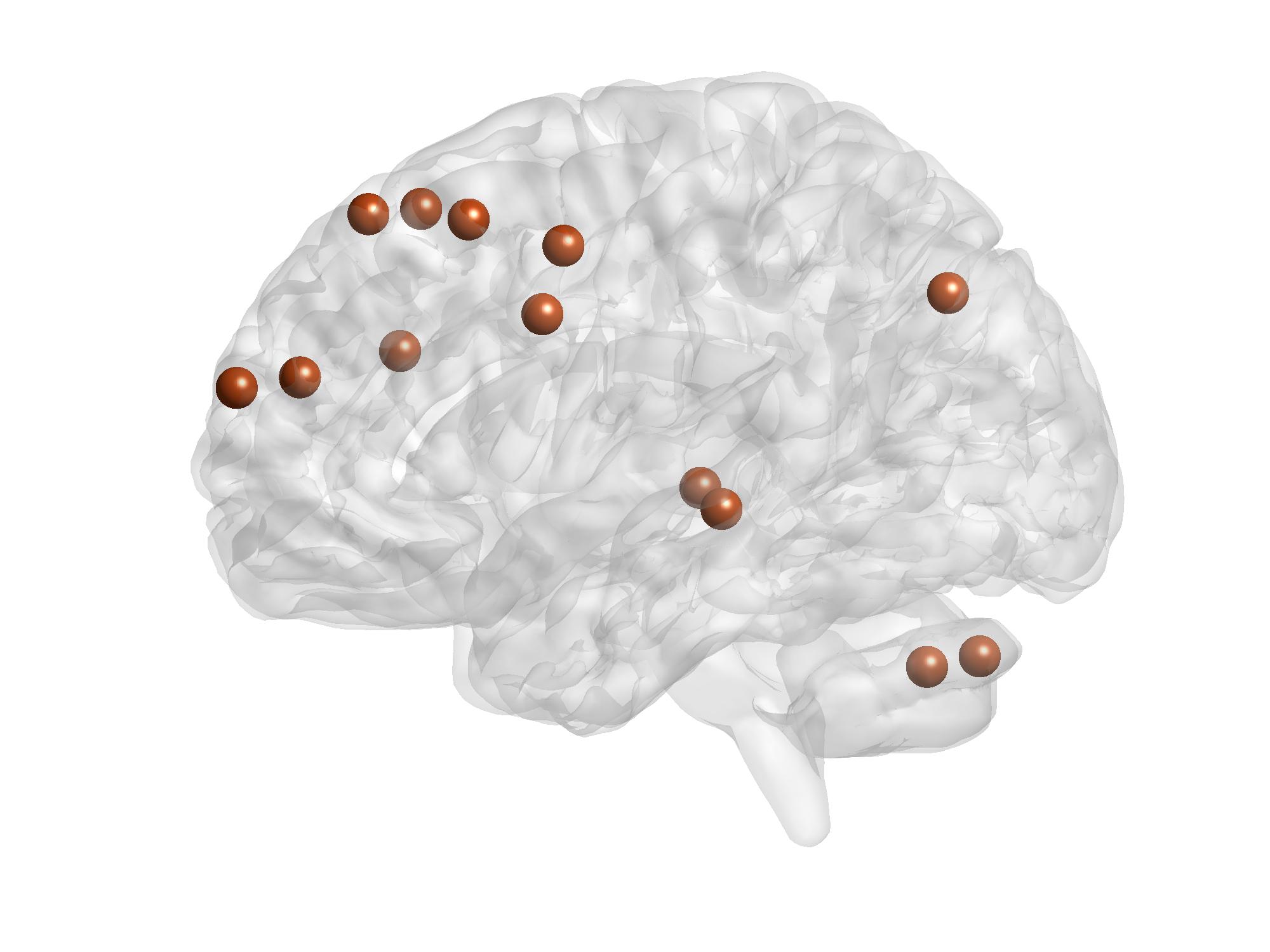} & \includegraphics[width=0.17\textwidth]{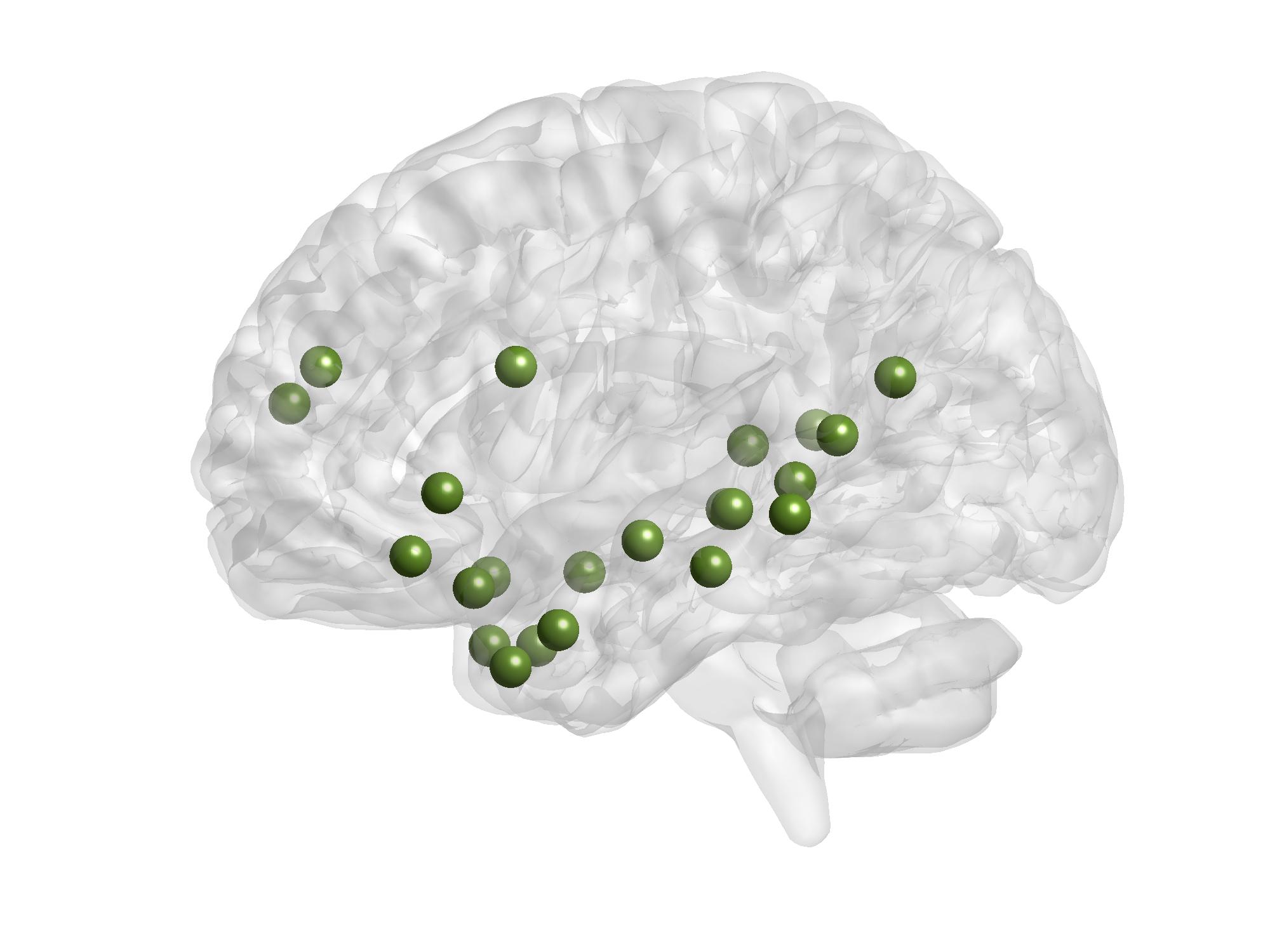} & \includegraphics[width=0.17\textwidth]{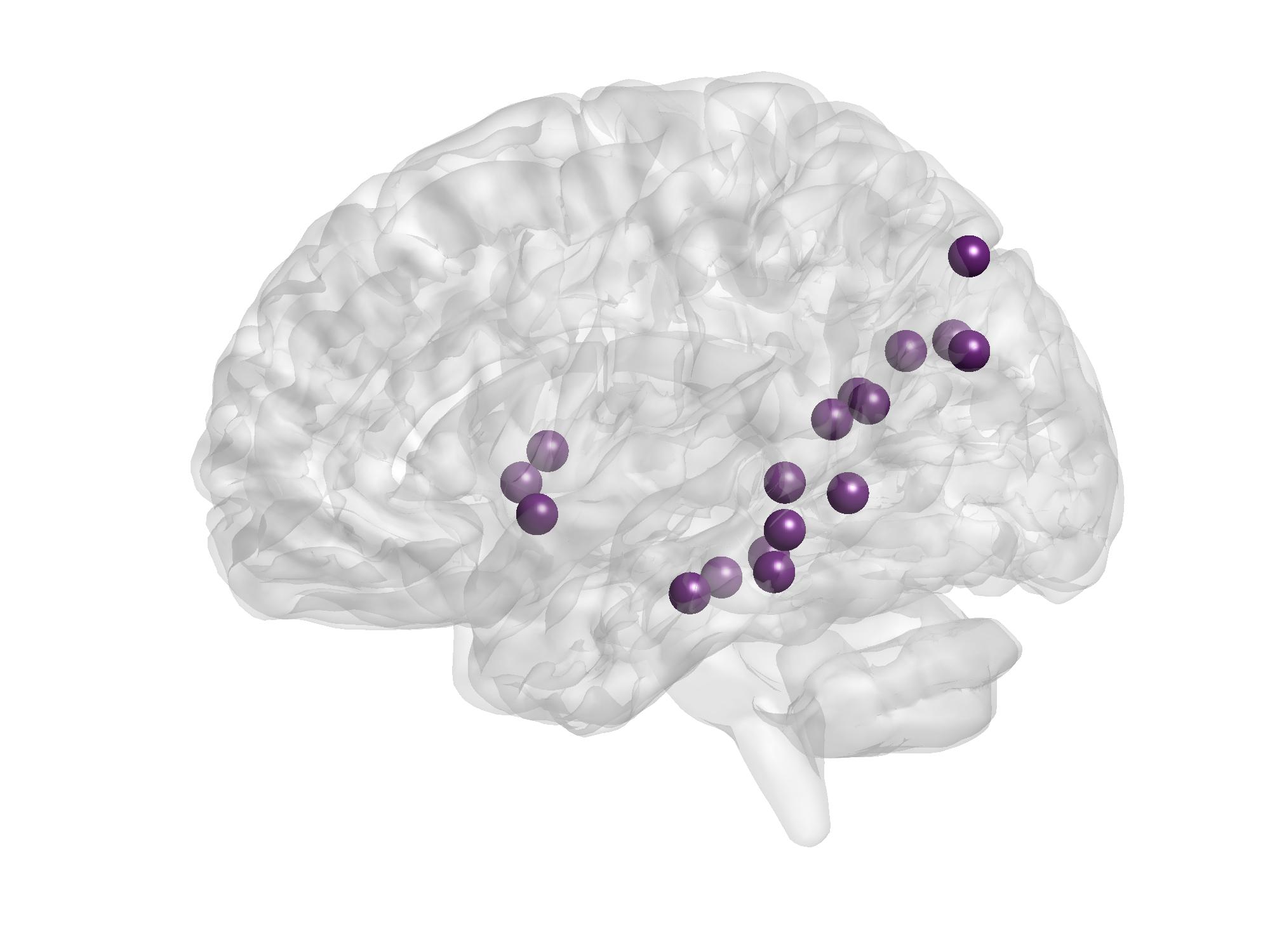} & \includegraphics[width=0.17\textwidth]{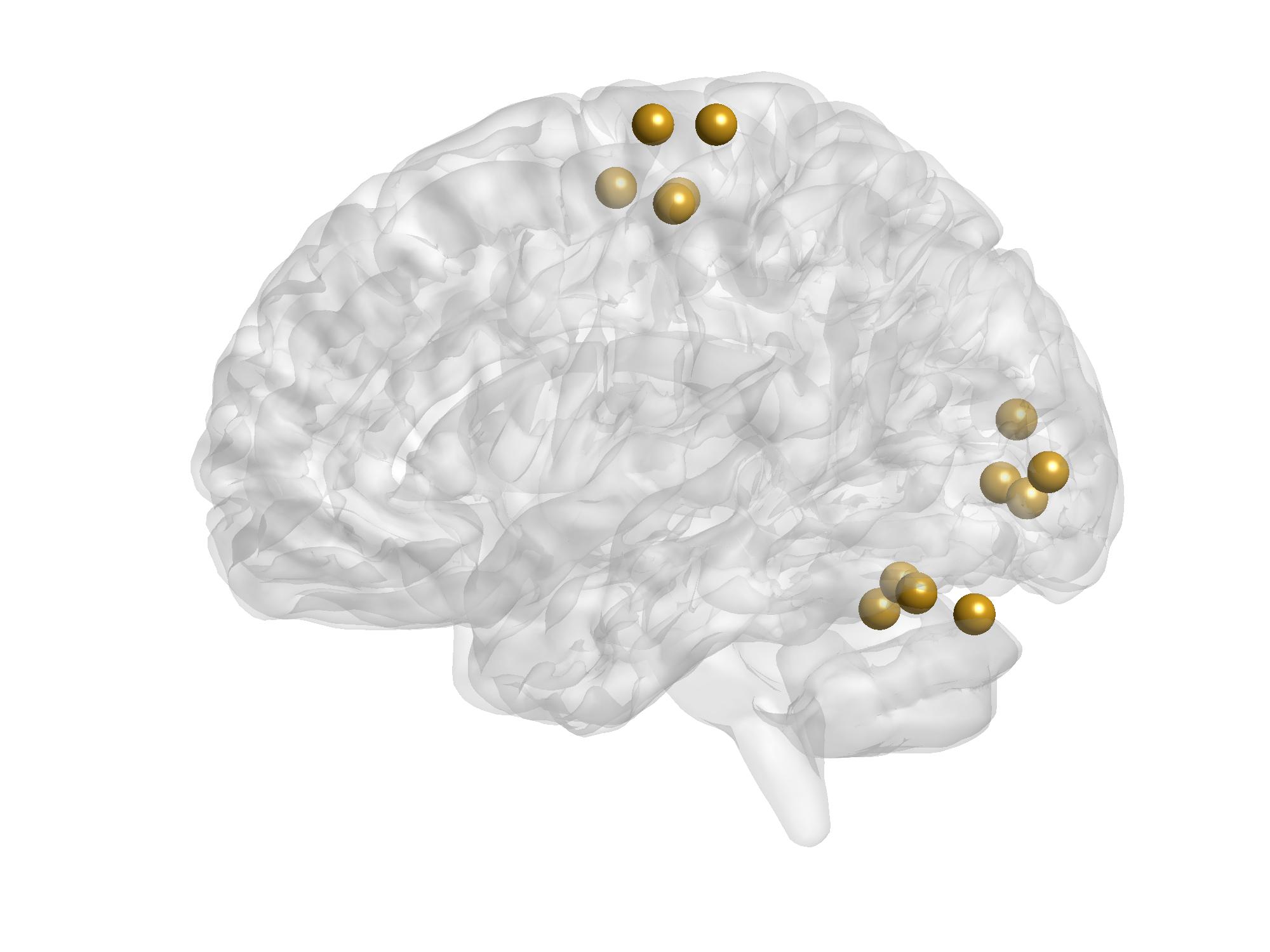} \\
		F & G & H & I & J\\ 
		\includegraphics[width=0.17\textwidth]{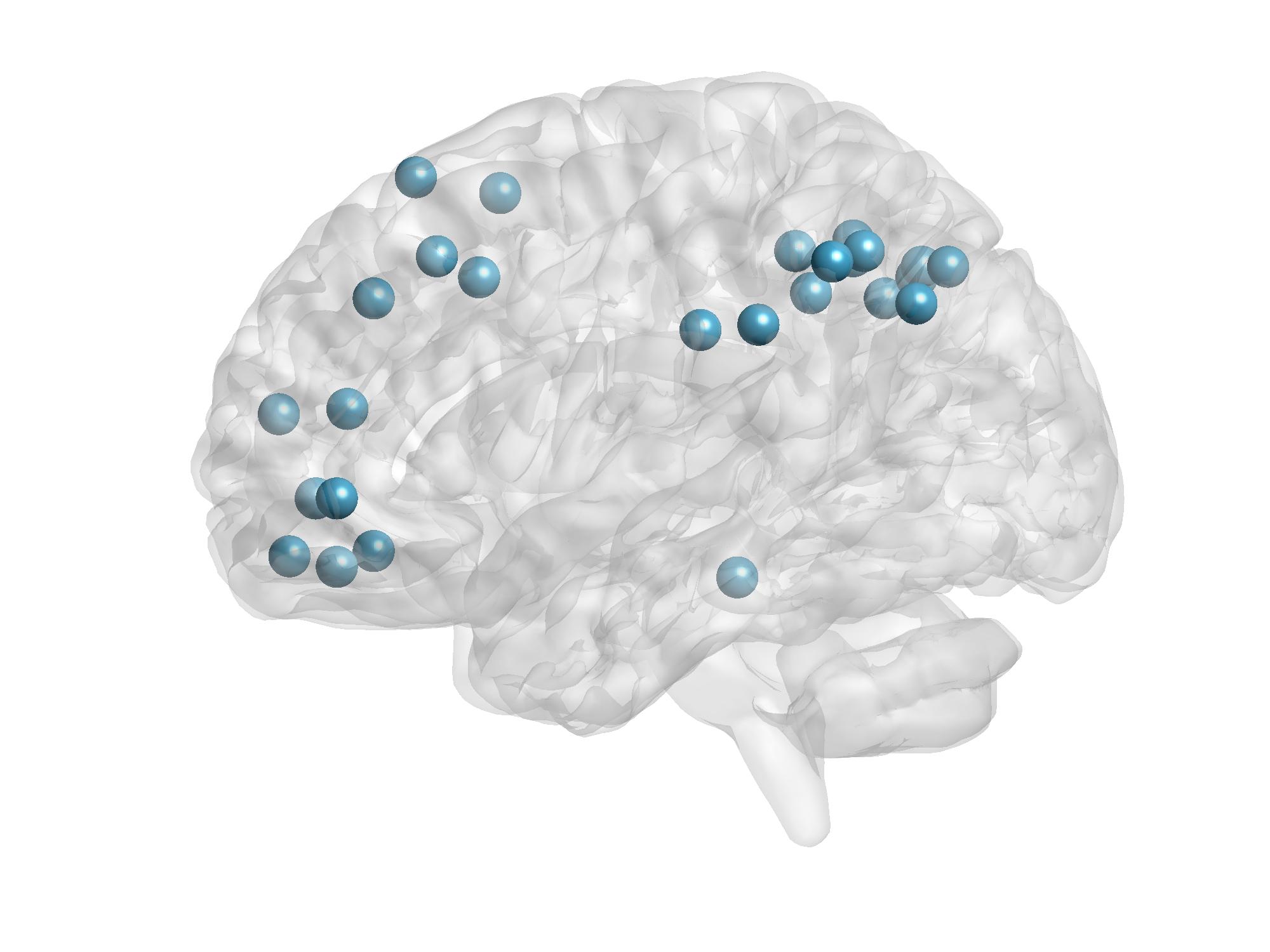}& \includegraphics[width=0.17\textwidth]{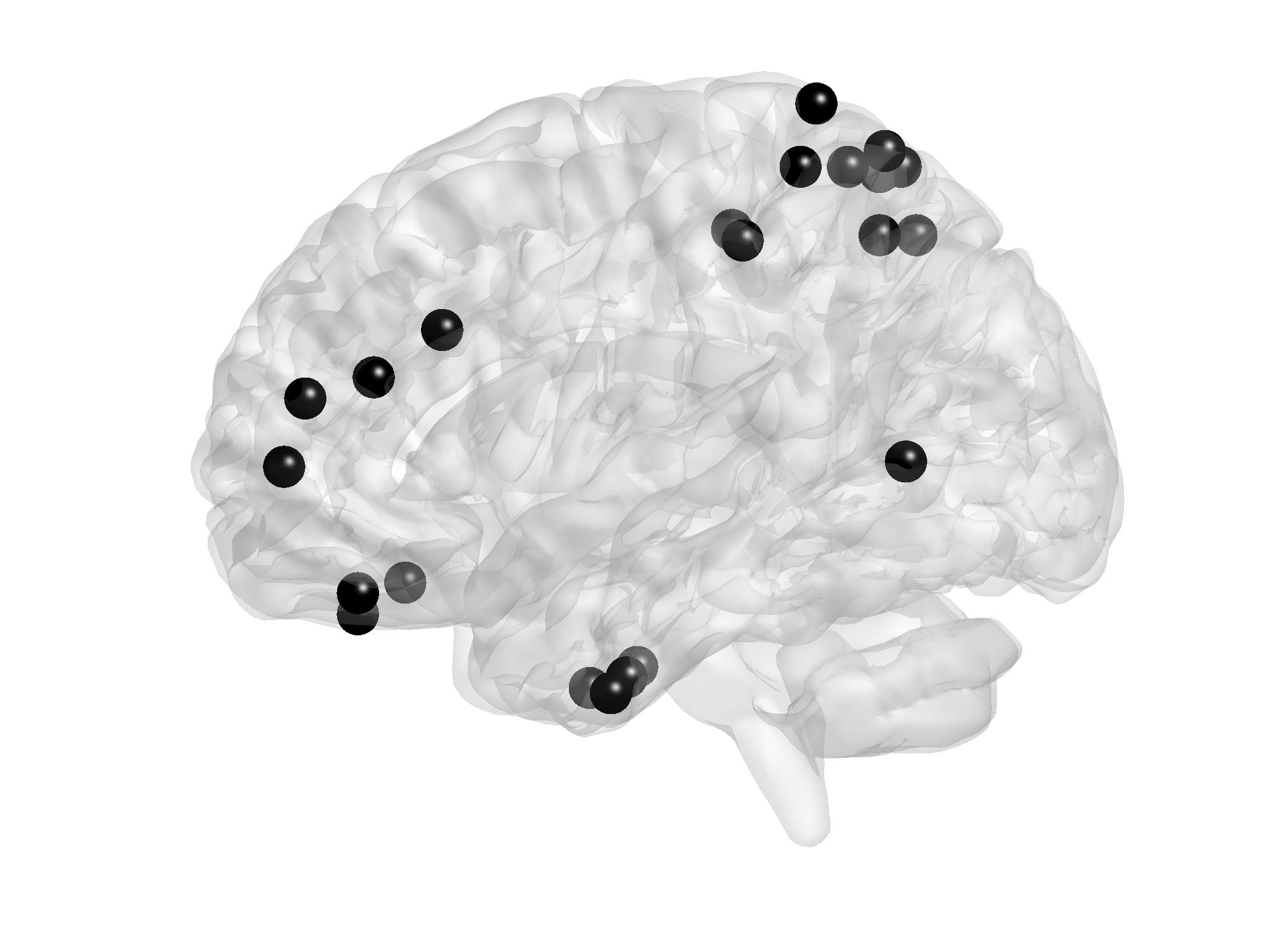} & \includegraphics[width=0.17\textwidth]{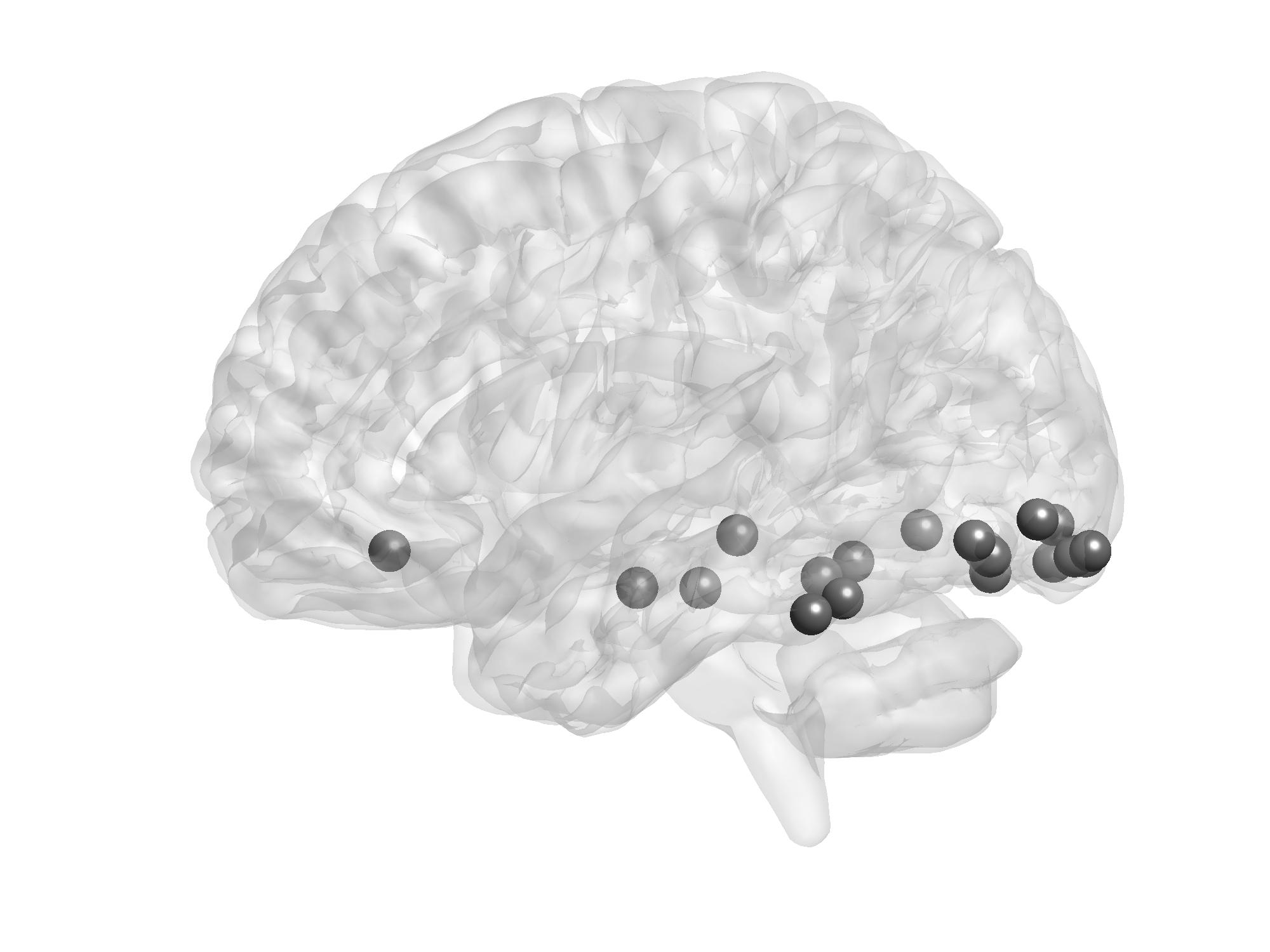} & \includegraphics[width=0.17\textwidth]{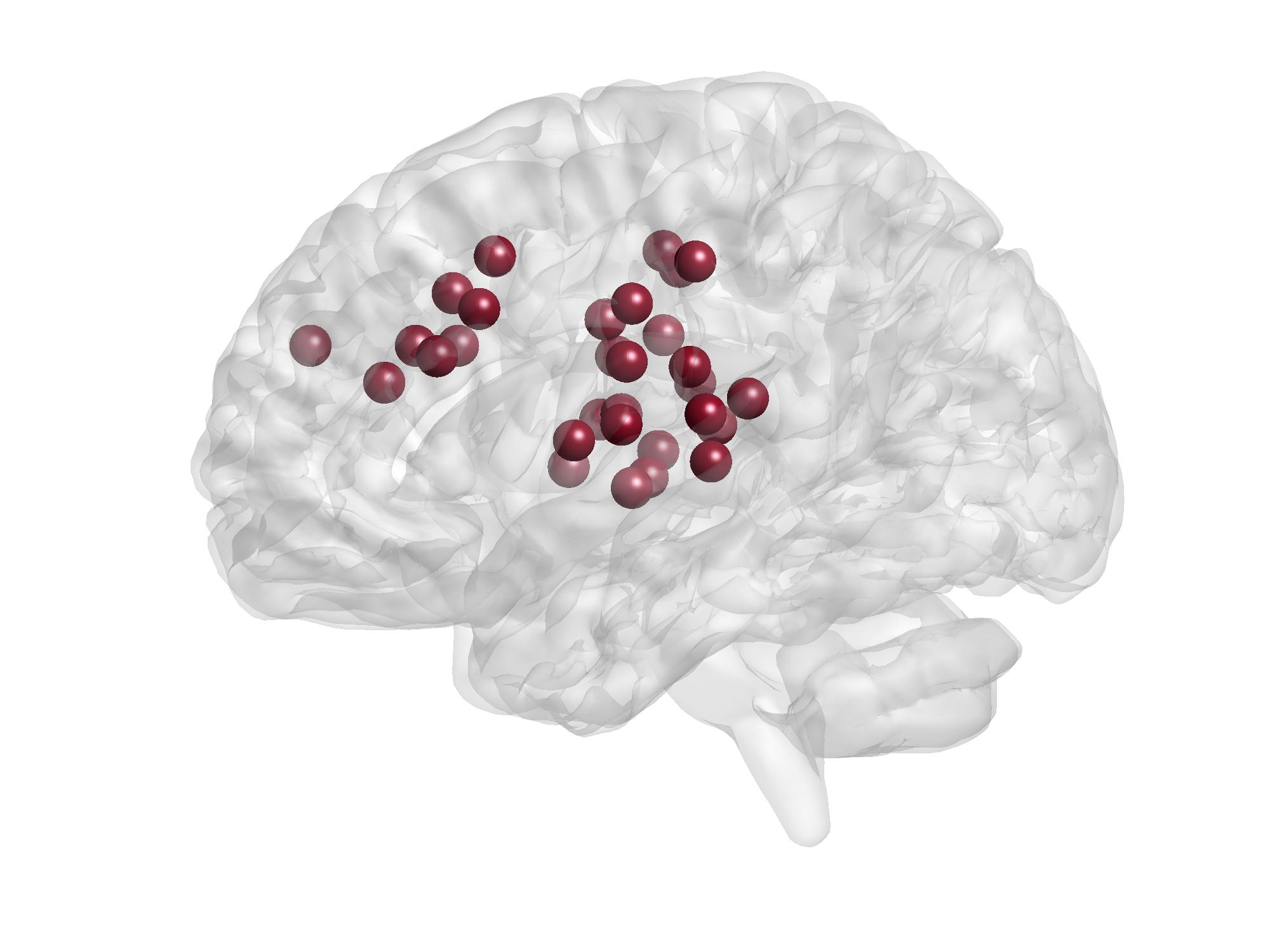} & \\
		K & L & M & N & \\ 
	\end{tabular} 
	\caption{Individual communities found by the supervised community detection method.}
	\label{fig:newcomms}
\end{figure}

\begin{figure}
	\centering
	\includegraphics[width=\textwidth]{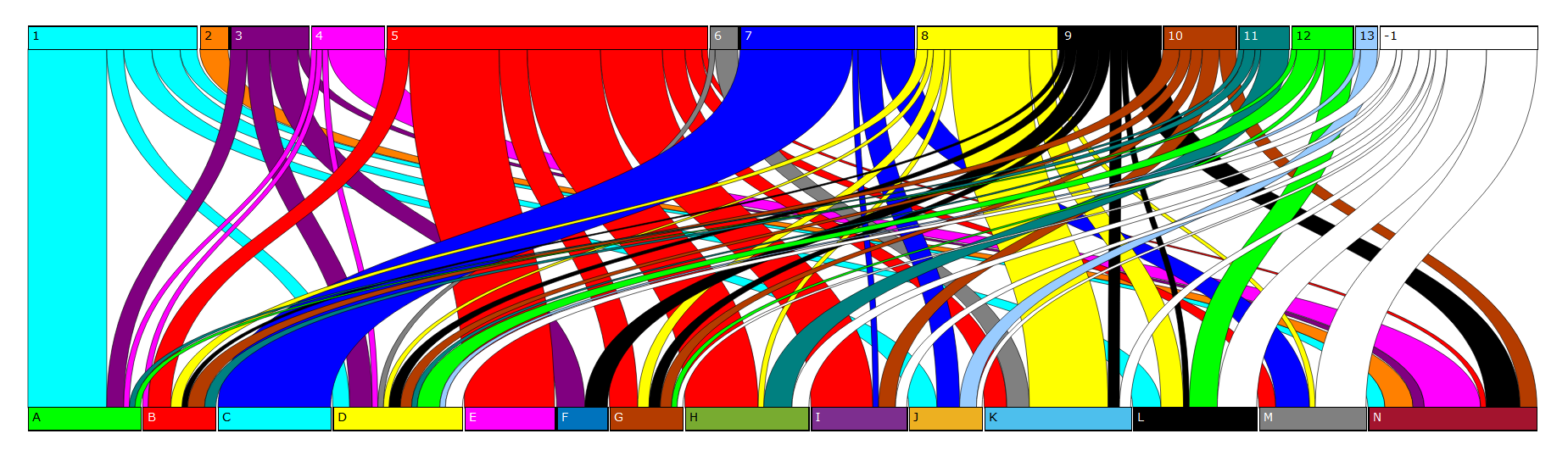}
	\caption{Sankey diagram of the node community assignment changes from the Power parcellation (top row) and the communities found by our method (bottom row).}
	\label{fig:sankey}
\end{figure}

\begin{figure}
	\centering
	\includegraphics[width=0.7\textwidth]{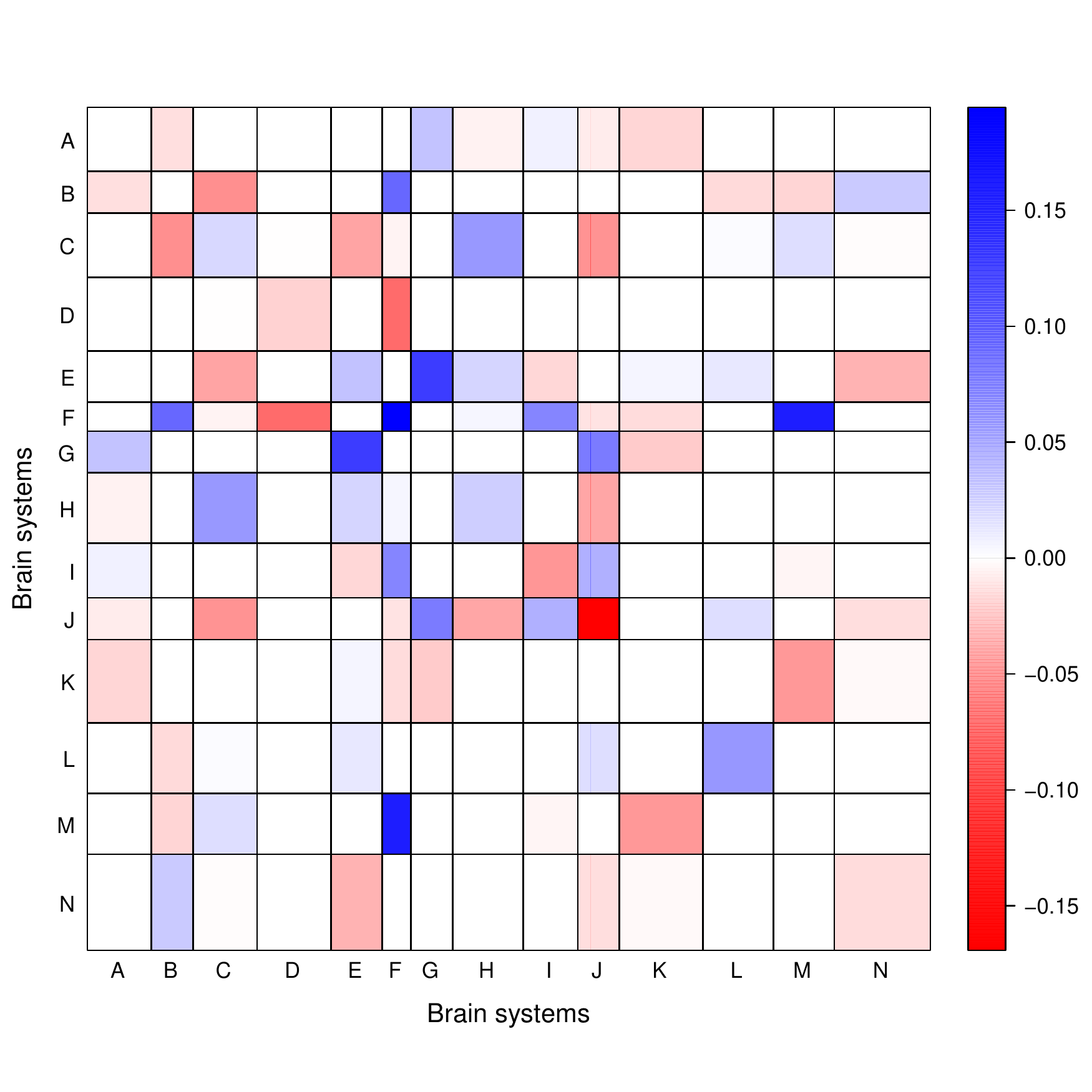}
	\caption{Matrix of fitted coefficients with the communities found by supervised community detection ($K=14$).}
	\label{fig:Bnew}
\end{figure}

Using 10-fold cross-validation, the average prediction error and number of non-zero different coefficients for a grid of values of $\lambda$ and $K$ are reported in Figure \ref{fig:CV-acc-number-SCD}.   Even when the number of communities is small, the accuracy of the supervised method is better than the baseline communities (Figure \ref{fig:powervsnew}), and as $K$ increases, the accuracy improves significantly. Comparing with other methods (see Table 1 in \cite{arroyo2016graphclass}), our method has good accuracy with a small number of highly interpretable parameters, and performs better than methods like lasso or DLDA, which are only able to select individual edges. 

\begin{figure}
	\centering
	\begin{minipage}{0.48\textwidth}
		\centering
		Cross-validation error
		\includegraphics[width=\textwidth]{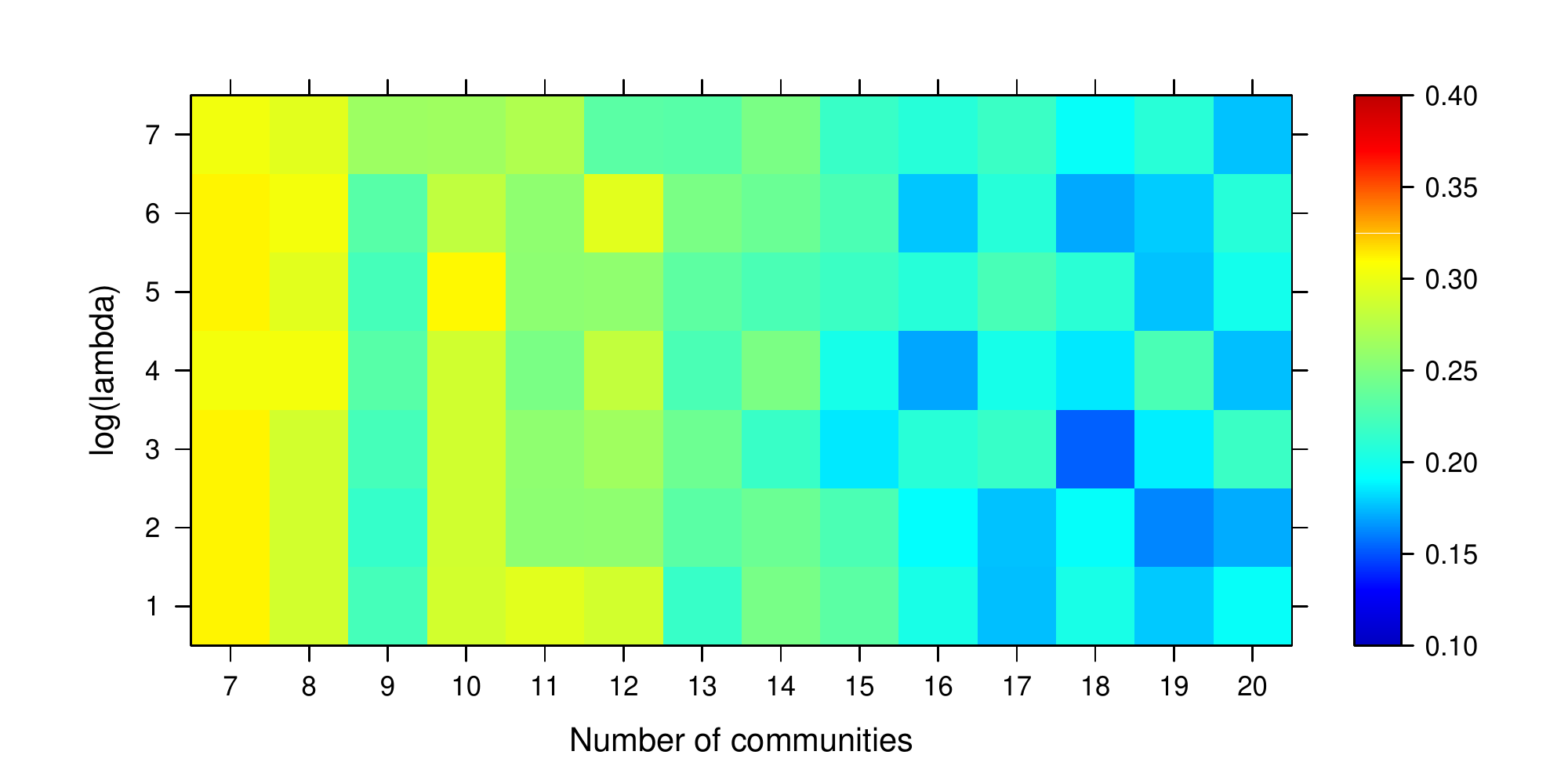}
	\end{minipage}
	\begin{minipage}{0.48\textwidth}
		\centering
		Number of non-zero coefficients
		\includegraphics[width=\textwidth]{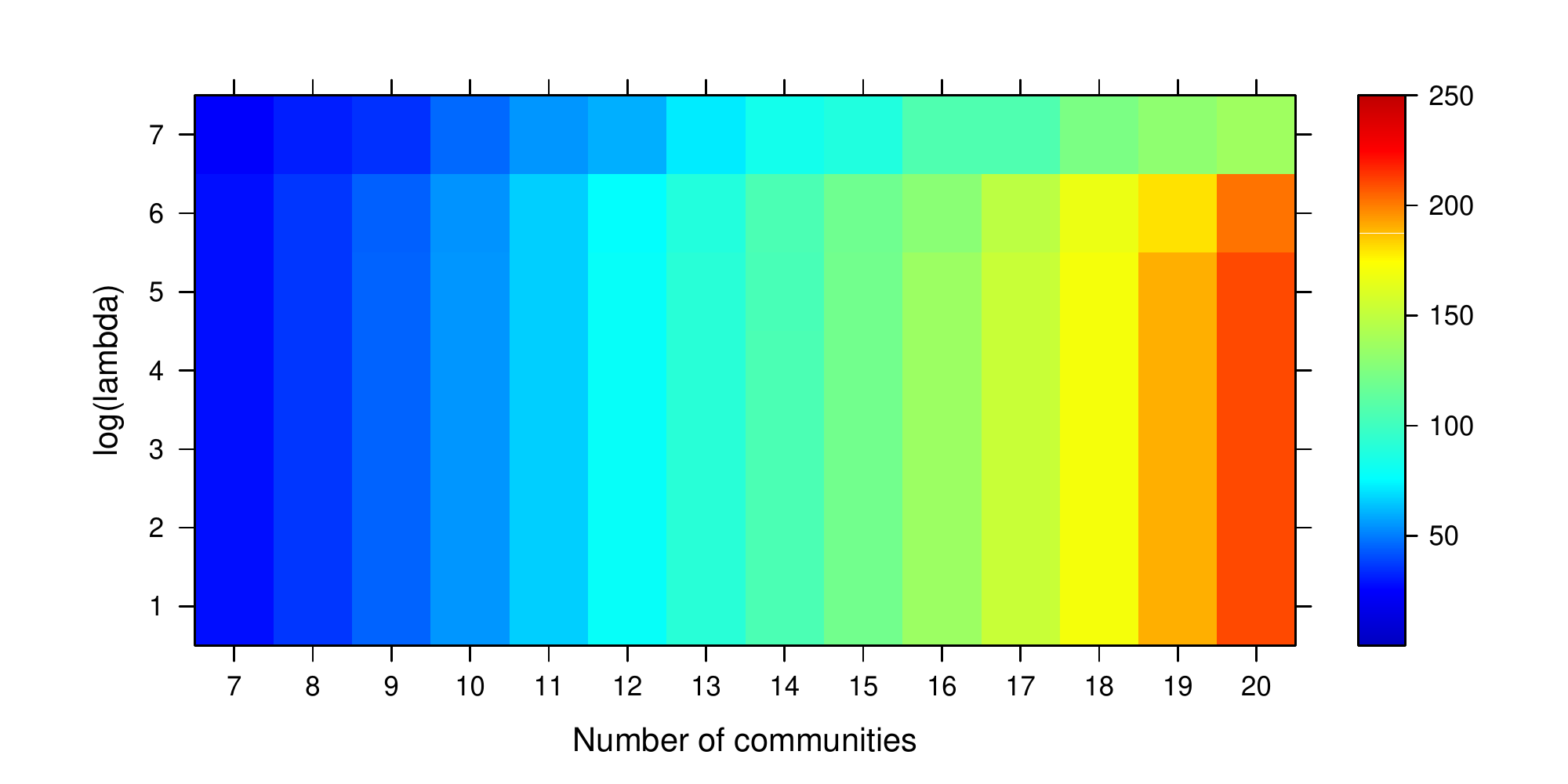}
	\end{minipage}
	\caption{Average cross-validation error (left) and average number of non-zero different coefficients (right) for a grid of $\lambda$ and $K$ values.}
	\label{fig:CV-acc-number-SCD}
\end{figure}

%% file: discussion.tex
\section{Discussion \label{sec:discus}}

Finding communities in networks is a much studied problem, but as with any unsupervised problem, different community detection algorithms often yield different results, and it is hard to compare them.   Even when comparisons with some ``ground truth'' are made, the ground truth is frequently just another network covariate which may or may not correspond to communities \citep{Peel2017}.  In contrast, our focus on community structure as a regularization tool in a prediction problem, with the structure of regulazation motivated by and arising from the underlying science, allows for straightforward evaluation and comparisons in terms of prediction error.   However, good prediction performance is only one of our goals; having a sparse and interpretable solution is even more important from the scientific point of view (provided, of course, the solution does have good prediction performace, as there is no point interpreting a bad solution).   We do not aim to achieve the lowest prediction error possible, as long as we can have comparable performance with a scientifically interpretable solution.    This is one reason that we imposed an equal coefficient constraint on all the edges within a cell, in additition to the practical advantages of simpler optimization and shorter computing times. 

From the statistical perspective, there is much more to be done.   More complex prediction rules  can be used instead of a linear function to obtain more flexible classifiers;  the same structured penalty extends easily to other methods such as polynomial regression, splines, generalized additive models, or anything else that fits coefficients to an expanded basis.  We can modify the loss function to not just evaluate the quality of prediction, but also the strength of discovered community structure;   this would allow a balance between finding the most predictive communities and the strongest communities purely in the network sense, which may or may not be the most predictive.     Community structure can be used as an approximation to more general models, for example, smooth graphons such as those in \cite{zhang2017estimating}, and those could alsobe leveraged in a prediction task. Valid statistical inference for this approach is an open question.  While there has been a lot of recent activity in high-dimensional post-selection inference \citep{van2014asymptotically,lee2016exact, lockhart2014significance}, the setting studied in this paper is a much harder framework of grouped rather sparse coefficients, and the groups themselves are learned from data, unlike other structured approaches such as group lasso \citep{yuan2006model}.  Developing an inference framework for assesing the group relationships between the variables is another future direction.   From the neuroscience prospective, applying methods like this to different clinical outcomes of interest, controlling for covariates such as age and gender, and understanding the relationship between community structures the brain organizes itself into for different tasks and under different conditions are all promising directions for future work.

%% file: appendix.tex
\appendix
\section{Appendix}

\subsection{Proofs}

\begin{proof}[Proof of Proposition \ref{proposition:expected-sigma-AY}]
 For any nodes $u$ and $v$ in $[n]$, $u<v$, the $(u,v)$ entry of the matrix $\hat{\Sigma}^{\mathcal{A}, \mathcal{Y}}$ can be expressed as
 \begin{align*}
     \Exp{\hat{\Sigma}^{\mathcal{A}, \mathcal{Y}}_{uv}} & = \e\left[\frac{1}{N}\sum_{m=1}^N A^{(m)}_{uv}\left\langle A^{(m)}, B\right\rangle + \sigma \epsilon_m A^{(m)}_{uv}\right] \nonumber\\
     &  = 2\sum_{(s,t)\in\mathcal{P}}\Sigma^{\mathcal{A}}_{[u,v][s,t]}
     B_{st}\\
     & = 2 \sum_{(s,t)\in\mathcal{P}}\left(\frac{1}{N}\sum_{m=1}^NR^{(m)}_{z_uz_v}R^{(m)}_{z_sz_t}\right)C_{z_sz_t}+ 2\Psi_{z_uz_v}C_{z_uz_v}\\
     & = F_{z_uz_v}.
 \end{align*}
 Note that the value of $\Exp{\hat{\Sigma}^{\mathcal{A}, \mathcal{Y}}_{uv}}$ is the same for all edges $(u,v)$ that are in the cell $(z_u,z_v)$, and hence it can be expressed through the membership matrix $Z$ as in Equation~\eqref{eq:expected=sigma-A-Y}.
\end{proof}


\begin{proof}[Proof of Proposition \ref{prop:condition-covariance}]
	Define the symmetric matrices $H,J\in\real^{K\times K}$ such that for  $(j,k)\in\mathcal{R}$ and
	for some $(u,v)\in\mathcal{P}$ with $z_u=j$ and $z_v=k$,
	\begin{equation*}
		H_{jk} := \sum_{\substack{(h,l)\in\ \mathcal{Q},\\ (h,l)\neq (j,k) }}\  \sum_{\substack{(s,t)\in\mathcal{P},\\(z_s, z_t) = (h,l)}}  \Sigma^{\mathcal{A}}_{[u,v][s,t]} B_{st} , 
	\end{equation*}
	\begin{equation*}
		J_{jk} := \Pi_{jk}\sum_{\substack{(s,t)\in\mathcal{P},\\(z_s, z_t) = (j,k)}}  B_{st}.
	\end{equation*}
	Then,
	\begin{equation*}
		F: = 2\left( J + \Psi \circ C + H\right).
	\end{equation*}
	Using Weyl's inequality,
	\begin{align}
		|\lambda_{\min}(F)| & \geq 2\left(|\lambda_{\min}\left( J + \Psi\circ C\right)| - |\lambda_{\max}(H)|\right). \label{proof:eigen1}
	\end{align}
	To bound the right hand side, by the Gershgorin disc theorem (see for example \cite{Bhatia1997}) and Equation~\eqref{eq:covarianceR-bound} we have
	\begin{align*}
		|\lambda_{\max}(H)|  \leq \max_{(j,k)\in\mathcal{Q}}K|H_{jk}| 
		 \leq (1-\delta)n^2_{\min}|\lambda_{\min}(\Pi\circ C')|.
	\end{align*}

	Let $\Theta = \operatorname{diag}(n_1, \ldots, n_K)$ be a diagonal matrix with the community sizes on the diagonal. Observe that
	\begin{equation*}
		2J = \Theta (\Pi\circ C')\Theta -    (I \circ \Theta \circ \Pi\circ C'),
	\end{equation*}
	and therefore,
		\begin{align*}
		2 |\lambda_{\min}(J)| & \geq  \left(\min_{j\in[K]}\Theta_{jj}^2\right)|\lambda_{\min}( \Pi\circ C')| - \max_{j\in[K]}(\Theta_{jj}\Pi_{jj}C'_{jj})\\
		& \geq n_{\min}^2|\lambda_{\min}( \Pi\circ C')| - n_{max}\max_{j\in[K]}(\Pi_{jj}C'_{jj})\\
		& = \Omega\left(n_{\min}^2|\lambda_{\min}( \Pi\circ C')| \right).
	\end{align*}
	Combining these results with Equation~\eqref{proof:eigen1},
		\begin{align*}
		|\lambda_{\min}(F)| & \geq 2\left(|\lambda_{\min}(J + \Psi\circ C)| - (1-\delta)n^2_{\min}|\lambda_{\min}(\Pi\circ C)|\right)\\
		& \geq 2\left( |\lambda_{\min}(J)| - \| 2\Psi\circ C\| - (1-\delta) n^2_{\min}|\lambda_{\min}(\Pi\circ C)| \right)\\
		& = \Omega(n^2_{\min}|\lambda_{\min}(\Pi\circ C')|).
	\end{align*}
\end{proof}

The next proposition will be needed to prove the theorems.
\begin{proposition} \label{prop:spectral-error}
	Suppose that networks $A^{(1)},\ldots,A^{(N)}$ satisfy Assumptions \ref{assump:network-distr} and \ref{assumption:centered}, and responses $Y_1,\ldots,Y_N$ satisfy Assumption \ref{assump:linear-model}. Then the expectation of the spectral norm of the difference between the
	matrix $\hat{\Sigma}^{\mathcal{A}, \mathcal{Y}}$, defined in Equation~\eqref{eq:sigmaAY-block}, and its expectation satisfies
	\begin{equation*}
	\Exp{\left\|\hat{\Sigma}^{\mathcal{A},\mathcal{Y}} - \Exp{\hat{\Sigma}^{\mathcal{A},\mathcal{Y}}}\right\|}  \lesssim 
	\sqrt{\frac{n}{N}}\left( K^2n_{\max}^2\pi^\ast\psi +
	\sigma +
	n_{\max}\psi^2 
	\right).
	\end{equation*}
	
\end{proposition}

\begin{proof}[Proof of Proposition \ref{prop:spectral-error}]
For each $m\in[N]$ and $u,v\in[n]$ define 
$$\tau^{(m)}_{uv} := A_{uv}^{(m)} - R^{(m)}_{uv}.$$
Note that $\tau^{(m)}_{uv}$ is sub-Gaussian, and $\Exp{\tau^{(m)}_{uv}}=0$.
 For any nodes $u$ and $v$ in $[n]$, the $(u,v)$ entry of the matrix $\hat{\Sigma}^{\mathcal{A}, \mathcal{Y}}$ can be expressed as
 \begin{align}
     \hat{\Sigma}^{\mathcal{A}, \mathcal{Y}}_{uv} & = \frac{1}{N}\sum_{m=1}^N\left[ A^{(m)}_{uv}\textbf{}\left\langle A^{(m)}, B\right\rangle + \sigma \epsilon_m A^{(m)}_{uv}\right] \nonumber\\
     &  = \frac{1}{N}\sum_{m=1}^N\left[\left(R^{(m)}_{z_uz_v} + \tau^{(m)}_{uv}\right) \sum_{s=1}^n\sum_{t=1}^n\left(R^{(m)}_{z_sz_t}+\tau^{(m)}_{st}\right)C_{z_sz_t} + \sigma \epsilon_m \left(R^{(m)}_{z_uz_v}+\tau_{uv}^{(m)}\right)\right].\label{eq:Sigma-decomp}
 \end{align}
 Observe that 
 \begin{equation*}
     \sum_{s=1}^n\sum_{t=1}^n R^{(m)}_{z_sz_t}C_{z_sz_t} = \left\langle M\circ C, R^{(m)}\right\rangle.
 \end{equation*}
 Also, for each $m\in[N]$ define $S^{(m)}\in\real^{K\times K}$ as the matrix with entries 
 \begin{equation*}
     S^{(m)}_{jk}:= \sum_{u:z_u=j}\sum_{v:z_v=k} \tau^{(m)}_{uv}.
 \end{equation*}
 Since the entries of $S^{(m)}$ are sums of independent sub-Gaussian random variables,  by Proposition 2.6.1 in \cite{Vershynin2010},
 \begin{equation*}
     \orlicz{S^{(m)}_{jk}} \leq \sqrt{M_{jk}}\psi \leq n_{\max}\psi.
 \end{equation*}
 Additionally, observe that
 \begin{equation*}
     \sum_{s=1}^n\sum_{t=1}^n \tau^{(m)}_{st}C_{z_sz_t} = \left\langle C, S^{(m)}\right\rangle,
 \end{equation*}
 which is also a sub-Gaussian random variable with
 \begin{equation*}
     \orlicz{\left\langle C, S^{(m)}\right\rangle}\leq \left(\sum_{j=1}^K\sum_{k=1}^KC_{jk}^2n^2_{\max}\psi^2\right)^{1/2} = n_{\max}\psi \|C\|_F.
 \end{equation*}
 For each $m\in[N]$, define matrices $E^{(m)}, T^{(m)}, U^{(m)}, W^{(m)}, X^{(m)}\in\real^{n\times n}$ by 
 \begin{align}
    E^{(m)}_{uv} & := R^{(m)}_{z_uz_v}\left\langle M\circ C, R^{(m)}\right\rangle + 2\Psi_{z_uz_v}C_{z_uz_v}, \label{eq:proof-E}\\
 	T^{(m)}_{uv} & :=  \tau^{(m)}_{uv} \left\langle M\circ C, R^{(m)}\right\rangle,\label{eq:proof-T}\\
 	U_{uv}^{(m)} & := R^{(m)}_{z_uz_v}\left\langle C, S^{(m)}\right\rangle ,\label{eq:proof-U}\\
 	V_{uv}^{(m)} & :=   \tau^{(m)}_{uv}\left\langle C, S^{(m)}\right\rangle - 2\left(\tau^{(m)}_{uv}\right)^2C_{z_uz_v}  +
 	2\left(\tau^{(m)}_{uv}\right)^2C_{z_uz_v} - 2\Psi_{z_uz_v}C_{z_uz_v} ,\label{eq:proof-V}\\
 	W_{uv}^{(m)} & :=  \sigma \epsilon_m \tau_{uv}^{(m)},\label{eq:proof-W}\\
 	X_{uv}^{(m)} & :=  \sigma \epsilon_m R^{(m)}_{z_uz_v}.\label{eq:proof-X}
 \end{align}
 Combining  with Equation~\eqref{eq:Sigma-decomp} , we get 
 \begin{equation*}
     \hat{\Sigma}^{\mathcal{A}, \mathcal{Y}} - \frac{1}{N}\sum_{m=1}^NE^{(m)} = \frac{1}{N}\sum_{m=1}^N\left(T^{(m)} + U^{(m)}+ V^{(m)}+W^{(m)}+X^{(m)} \right).
 \end{equation*}
  The expected value of the right hand side is zero, and therefore the expectation of $\hat{\Sigma}^{\mathcal{A}, \mathcal{Y}}$, given parameters $\{R^{(m)}\}_{m=1}^N$ and 
  $Z$, is 
 \begin{equation*}
 \Exp{\hat{\Sigma}^{\mathcal{A}, \mathcal{Y}}}=\frac{1}{m}\sum_{i=1}^mE^{(i)}.    
 \end{equation*}
 Define constants $c^\ast, r_m^\ast, \tilde{r}_m>0$ and as
 \begin{align*}
     c^\ast & := \max_{j,k\in[K]} |C_{jk}|,\\
     r_m^\ast & := \max_{j,k\in[K]} |R_{jk}^{(m)}|,\\
     \tilde{r}_m & := \|R^{(m)}\|_F,
 \end{align*}
 and set $r^\ast := \sqrt{\frac{1}{N}\sum_{m=1}^N(r_m^\ast)^2}$, $\tilde{r} := \sqrt{\frac{1}{N}\sum_{m=1}^N(\tilde{r}_m)^2}$. Observe that $r^\ast_m=O(1)$, $\tilde{r}=O(K)$, and $c^\ast = \Theta(1)$.
 With these definitions, we have
 \begin{equation}
     |\left\langle M\circ C, R^{(m)}\right\rangle| \leq \|M\circ C\|_F \|R^{(m)}\|_F \leq Kn^2_{\max}c^\ast \tilde{r}_m. \label{eq:max_QC-R}
 \end{equation}
 Define matrices $E,T, U, V, W, X\in\real^{n\times n}$ as the averages of the corresponding matrices in Equations~\eqref{eq:proof-E}-\eqref{eq:proof-X}, that is,
 $$E:=\frac{1}{N}\sum_{m=1}^NE^{(m)},$$
 and analogously for the others.
 Note that the variables defined in Equations~\eqref{eq:proof-T}-\eqref{eq:proof-X} have mean zero. The variables $T^{(m)}_{uv}, X^{(m)}_{uv}$ and $U^{(m)}_{uv}$ are sub-Gaussian, and hence,
 \begin{align*}
     \orlicz{T_{uv}^{(m)}} & \leq  |\left\langle M\circ  C, R^{(m)}\right\rangle|\psi
     \leq Kn_{\max}^2c^\ast \tilde{r}_m\psi, \\
     \orlicz{U_{uv}^{(m)}} & \leq n_{\max}|R^{(m)}_{z_uz_v}|\|C\|_F\psi
     \leq n_{\max}r^\ast_m \|C\|_F\psi, \\
     \orlicz{X_{uv}^{(m)}} & = \sigma |R^{(m)}_{z_uz_v}| \leq \sigma r^\ast_m.
 \end{align*}
 Then, Proposition 2.6.1 of \cite{Vershynin2010} implies
 \begin{align*}
     \orlicz{T_{uv}} := \orlicz{\frac{1}{N}\sum_{m=1}^mT_{uv}^{(m)}} & \leq  
     \frac{ Kn_{\max}^2c^\ast \tilde{r}\psi}{\sqrt{N}},\\
     \orlicz{U_{uv}}:= \orlicz{\frac{1}{N}\sum_{m=1}^mU_{uv}^{(m)}} & \leq 
     \frac{ n_{\max}r^\ast \|C\|_F\psi}{\sqrt{N}},\\
     \orlicz{X_{uv}} := \orlicz{\frac{1}{N}\sum_{m=1}^mX_{uv}^{(m)}} & \leq  \frac{\sigma r^\ast}{\sqrt{N}}.
 \end{align*}
 By Corollary 3.3 of \cite{Bandeira2016}, the expectation of the spectral norm of the matrices $T$, $U$ and $X$ can be bounded as
 \begin{align}
     \Exp{\|T\|} & \lesssim 
     \frac{\sqrt{n} Kn_{\max}^2c^\ast \tilde{r} \psi}{\sqrt{N}},\label{eq:p-ET}\\
     \Exp{\|U\|}  & \lesssim
     \frac{\sqrt{n} n_{\max}r^\ast \|C\|_F\psi}{\sqrt{N}},\label{eq:p-EU}\\
     \Exp{\|X\|}  & \lesssim
      \frac{\sqrt{n}\sigma r^\ast}{\sqrt{N}}.\label{eq:p-EX}
 \end{align}

 The variables $V^{(m)}_{uv}$ and $W_{uv}^{(m)}$ are sums of products of sub-Gaussian random variables, and hence, sub-exponential. To bound these variables, observe that by Lemma 2.7.6 of \cite{Vershynin2010},  for any $u,v,s,t\in[n]$, with $(u,v)\neq (s,t)$ and $(u,v)\neq (t,s)$ the following inequalities hold:
 \begin{align*}
     & \orliczexp{\tau^{(m)^2}_{uv}} \leq 2 \psi^2,\\
     &\orliczexp{\tau^{(m)}_{uv}\tau^{(m)}_{st}} =\psi^2,\\
     & \orliczexp{\epsilon_m \tau^{(m)}_{uv}} = \psi.
 \end{align*}
 Using these inequalities, the sub-exponential norms of the corresponding random variables in Equations  \eqref{eq:proof-V} and \eqref{eq:proof-W} are bounded above as
 \begin{align*}
     \orliczexp{V_{uv}^{(m)}} & \leq n_{\max}\|C\|_F\psi^2,\\
     \orliczexp{W_{uv}^{(m)}} & \leq  \sigma \psi,
 \end{align*}
 and by combining Theorem 2.8.2 of \cite{Vershynin2010} for sums of sub-exponential random variables, and Corollary 3.5 of \cite{Bandeira2016}, the expected value of the spectral norm of the corresponding matrices satisfies
  \begin{align}
     \Exp{\|V\|} := \Exp{\left\|\frac{1}{N}\sum_{m=1}^NV^{(m)}\right\|} & \lesssim \frac{\sqrt{n}n_{\max}\|C\|_F\psi^2}{\sqrt{N}},\label{eq:p-EV}\\
     \Exp{\|W\|} := \Exp{\left\|\frac{1}{N}\sum_{m=1}^NW^{(m)}\right\|} & \lesssim
     \frac{\sqrt{n}\sigma \psi}{\sqrt{N}}.\label{eq:p-EW}
 \end{align}
 Finally, combining Equations~\eqref{eq:p-ET}-\eqref{eq:p-EW} and~\eqref{eq:Sigma-decomp}, we obtain the  bound
 \begin{align*}
     \Exp{\left\|\hat{\Sigma}^{\mathcal{A},\mathcal{Y}}-E\right\|}   & \lesssim 
     \sqrt{\frac{n}{N}}\left( Kn_{\max}^2c^\ast \tilde{r}\psi +
     \sigma r^\ast +
     n_{\max}\|C\|_F\psi^2 + \sigma \psi
     \right)\\
     & \lesssim \sqrt{\frac{n}{N}}\left( K^2n_{\max}^2\pi^\ast\psi +
     \sigma +
     n_{\max}\psi^2 
     \right) . 
 \end{align*}
 The proof is completed by observing that $c^\ast$, $\|C\|_F$ are $\Theta(1)$, $\tilde{r} = O(K\pi^\ast)$,  $r^\ast$, $\psi$ are $O(1)$, with $r^\ast + \psi=\Theta(1)$.
 
 \end{proof}

 
 \begin{proof}[Proof of Theorem~\ref{thm:eigenvectors}]
 By the Davis-Kahan theorem in the form given by \cite{yu2014useful}, there exists an $O\in\mathcal{O}_K$ such that
 \begin{equation}
\|\hat{V} - VO\|_F \lesssim \frac{\sqrt{K}\|\hat{\Sigma}^{\mathcal{A},\mathcal{Y}} - \Exp{\hat{\Sigma}^{\mathcal{A},\mathcal{Y}}\|} }{\left|\lambda_{K}\left(\Exp{\hat{\Sigma}^\mathcal{A,\mathcal{Y}}}\right) - \lambda_{K+1}\left(\Exp{\hat{\Sigma}^\mathcal{A,\mathcal{Y}}}\right)\right|} . 
\label{eq:davis-kahan}
\end{equation} 
By Proposition~\ref{proposition:expected-sigma-AY}, $\Exp{\hat{\Sigma}^{\mathcal{A},\mathcal{Y}}}=ZFZ^T - \text{diag}(ZFZ^T)$, and by writing $Z=(Z\Theta^{-1/2})\Theta^{1/2}$, the denominator of Equation~\eqref{eq:davis-kahan} can be bounded from below using Weyl's inequality as
\begin{align*}
\left|\lambda_{K}\left(\Exp{\hat{\Sigma}^{\mathcal{A},\mathcal{Y}}}\right) - \lambda_{K+1}\left(\Exp{\hat{\Sigma}^{\mathcal{A},\mathcal{Y}}}\right)\right|    & \geq \left|\lambda_{K}\left(\Exp{\hat{\Sigma}^{\mathcal{A},\mathcal{Y}}}\right)\right| - 2\left|\lambda_{K+1}\left(\Exp{\hat{\Sigma}^{\mathcal{A},\mathcal{Y}}}\right)\right| \\ & \geq |\lambda_{\min}(\Theta^{1/2}F \Theta^{1/2})| -\max_{k\in [K]} |F_{kk}|\\
 & \geq  n_{\min} |\lambda_{\min}(F)| - O(1).
\end{align*}
By taking expectations on both sides of~\eqref{eq:davis-kahan}, and combining with Proposition~\ref{prop:spectral-error},
\begin{equation*}
\Exp{\|\hat{V} - VO\|_F} \lesssim \frac{K^{1/2}\sqrt{n}(K^2n_{\max}^2\pi^\ast\psi +
     \sigma +
     n_{\max}\psi^2 )} {\sqrt{N}n_{\min}|\lambda_{\min}(F)|}.
\end{equation*} 
 \end{proof}
 
 
 \begin{proof}[Proof of Theorem~\ref{theorem:communityerror}]
Recall that $\hat{V}$ and $V$ are the leading eigenvectors of $\hat{E}$ and $E$. Note that $V$ has $K$ unique rows, which we can arrange in a matrix $\mu\in\real^{K\times K}$ and write $V= Z\mu$. We apply $K$-means to the rows of $\hat{V}$ and obtain $K$ centroids, arranged in a matrix $\nu\in\real^{K\times K}$, and community assignments $\hat{Z}$. By Lemma 3.2 of \cite{rohe2011spectral}, a sufficient condition for a vertex $u$ to be correctly clustered is 
 \[\|\hat{Z}_{u\cdot}\nu - Z_{u\cdot}\mu \tilde{O}\|_2\leq \frac{1}{\sqrt{2n_{\max}}} \]
  for any $\mu_v \tilde{O}\neq \hat{Z}_{u\cdot}\nu$. 
 Following \cite{rohe2011spectral}, define $\mathcal{M}\subset\{1,\ldots,n\}$ as 
 \begin{equation*}
 \mathcal{M} := \left\{u: \|\hat{Z}_{u\cdot}\nu - Z_{u\cdot}\mu \tilde{O}\|_2\geq \frac{1}{\sqrt{2n_{\max}}}\right\},
 \end{equation*}
which contains all the misclusterd vertices, and hence $\min_{O\in\mathcal{P}_K}\|\hat{Z}- ZO\|_F^2\leq |\mathcal{M}|$
 By the same arguments as in the proof of Theorem 3.1 in \cite{rohe2011spectral}, 
 \begin{eqnarray*}
 |\mathcal{M}| & \leq & 
             8n_{\max}\|\hat{V} - V O\|_F^2.
 \end{eqnarray*}
 The condition on the community sizes implies that $n_{\min}\asymp n_{\max} \asymp n/K$, and therefore
 \begin{align*}
     \Exp{\min_{O\in\mathcal{P}_K} \|\hat{Z} - Z\|_F} & \leq 2^{3/2}\sqrt{n_{\max}}\Exp{\|\hat{V} - V O\|_F}\\
     & \lesssim \sqrt{\frac{n}{K}} \left(\frac{K^{3/2}(n^2\pi^\ast\psi + \sigma + n\psi^2/K)}{\sqrt{nN}}\right)\min\left\{\frac{K^2}{n^2\lambda^\ast}, 1\right\}
 \end{align*}
 
 \end{proof}


\subsection{Optimization with the elastic net penalty \label{sec:elastic-net}}
For the optimization problem \eqref{eq:blockConstrainedProblem}, consider adding an elastic net penalty of the form
\[\Omega(B) = \lambda\sum_{i,j}|B_{ij}|  + \frac{\gamma}{2}\|B\|_F^2.\]
Let $U^{(t)}=W^{(t-1)}-\frac{1}{\rho}V^{(t-1)}$. Then the first step of the ADMM algorithm in Equation~\eqref{eq:alg_lossstep} can be expressed as
\begin{eqnarray}
B^{(t)} & = & \argmin_B\left\{  \ell(B) + \lambda\sum_{i,j}|B_{ij}|  + \frac{\gamma}{2}\|B\|_F^2 + \frac{\rho}{2}\left\|B-U^{(t)}\right\|_F^2\right\}\\
& = & \argmin_B\left\{  \ell(B) + \lambda\sum_{i,j}|B_{ij}|  + \left(\frac{\gamma+ \rho}{2}\right)\left\|B -  \frac{\rho}{ \gamma+ \rho} U^{(t)}\right\|_F^2 \right\}.
\end{eqnarray}
This objective function can be written as a sum of two convex functions, one of them differentiable. We use an accelerated proximal algorithm to solve this problem \citep{parikh2013proximal}. In the special case of $\lambda=0$, the problem can be rewritten as ridge regression with an offset. Let $T^{(t)}_B = B -  \frac{\rho}{ \gamma+ \rho} U^{(t)}$. Then,
\begin{eqnarray}
B^{(t)}
& = & \argmin_{T^{(t)}_B}\left\{  \ell\left(T^{(t)}_B + \frac{\rho}{ \gamma+ \rho} U^{(t)}\right) + \left(\frac{\gamma+ \rho}{2}\right)\left\|T^{(t)}_B\right\|_F^2 \right\}.
\end{eqnarray}
We compute the solution for this special case using the R package \texttt{glmnet} \citep{friedman2009glmnet}.